\documentclass[final]{lmcs}

\keywords{terminal coalgebra, countable iteration, descending chain condition, Hausdorff functor, Vietoris functor, initial algebra}

\usepackage{amssymb,amsmath,mathrsfs}
\usepackage{mathtools}
\usepackage{multicol}
\usepackage{tikz}
\usetikzlibrary{cd,fit,calc,positioning,arrows}
\usepackage{ifthen}
\usepackage{bbold}
\usepackage[notref,notcite]{showkeys}

\usepackage{seqsplit}
\usepackage{xstring}
\newcommand{\defaultshowkeysformat}[1]{%
\StrSubstitute{#1}{ }{\textvisiblespace}[\TEMP]%
\parbox[t]{2.5cm}{\raggedright\normalfont\small\ttfamily\(\{\){\color{red!50!black}\expandafter\seqsplit\expandafter{\TEMP}}\(\}\)}%
}

\renewcommand*\showkeyslabelformat[1]{%
\noexpandarg%
\defaultshowkeysformat{#1}%
}

\usepackage{rotating}

\usepackage{ifdraft}

\newcommand{\resetCurThmBraces}{%
\gdef\curThmBraceOpen{(}%
\gdef\curThmBraceClose{)}}
\resetCurThmBraces
\newcommand{\removeThmBraces}{%
\gdef\curThmBraceOpen{}%
\gdef\curThmBraceClose{}}
\resetCurThmBraces

\newenvironment{notheorembrackets}{\removeThmBraces}{\resetCurThmBraces}

\usepackage{etoolbox}
\patchcmd{\thmhead}{(#3)}{\curThmBraceOpen #3\curThmBraceClose }{}{}

\newcommand{\DCC}{\text{DCC}}
\newcommand{\names}{\mathbb{A}}

\usepackage[inline]{enumitem}
\setlist[enumerate,1]{label=(\arabic*),font=\normalfont,align=left,leftmargin=0pt,labelindent=0pt,listparindent=\parindent,labelwidth=0pt,itemindent=!,topsep=3pt,parsep=0pt,itemsep=3pt,start=1}
\setlist[enumerate,2]{label=(\alph*),font=\normalfont,labelindent=*,leftmargin=*,start=1}

\usepackage{refcount}

\usepackage{hyperref}
\hypersetup{hidelinks,final,bookmarks}

\usepackage[capitalize,noabbrev,nameinlink]{cleveref}

\crefname{rem}{Remark}{Remarks}

\renewcommand\itemautorefname{Item}

\usepackage[footnote,marginclue,nomargin]{fixme}

\numberwithin{equation}{section}

\newcommand{\mysubsec}{\paragraph*}

\usepackage{color}

\usepackage{graphicx}
\usepackage{latexsym}
\usepackage{bussproofs}
\usepackage{proof}
\usepackage{tikz}
\usetikzlibrary{arrows}
\usepackage[all]{xy}

\usepackage{microtype}
\usepackage{graphicx}
\usepackage{amsbsy,eucal}  
\usepackage{dsfont}
\usepackage{xspace}
\newcommand{\labelling}{\ell}
\newcommand{\TT}{\mathcal{T}}

\newcommand{\treepartial}{\rho}

\renewcommand{\root}{\mbox{\sf root}}

\renewcommand{\phi}{\varphi}
\newcommand{\phibar}{\overline{\phi}}
\newcommand{\parto}{\rightrightarrows}

\newcommand{\rhobar}{\overline{\rho}}
\newcommand{\partialbar}{\overline{\partial}}
\newcommand{\erase}{\textsf{erase}}
\newcommand{\arrowsigma}{\,\lower1pt\hbox{$\abovearrow{\!\! \sigma}$}\! }
\newcommand{\abovearrow}[1]{\rightarrow\hspace{-.12in}\raisebox{1.0ex}
{$\scriptscriptstyle{#1}$}\hspace{.13in}}

\usepackage{ifdraft}
\ifdraft{
  \overfullrule=0mm 
}{}

\newcommand{\pullbackangle}[2][]{\arrow[phantom,to path={
                     -- ($ (\tikztostart)!1cm!#2:([xshift=8cm]\tikztostart) $)
                        node[anchor=west,pos=0.0,rotate=#2,
                        inner xsep = 0]
                        {\begin{tikzpicture}[minimum
                        height=1mm,baseline=0,#1]
    \draw[-] (0,0) -- (.5em,.5em) -- (0,1em);
                        \end{tikzpicture}}}]{}}

\tikzstyle{shiftarr}=[
        rounded corners,%
        to path={--([#1]\tikztostart.center)
                     -- ([#1]\tikztotarget.center) \tikztonodes
                     -- (\tikztotarget)},
]

\theoremstyle{plain}
\newtheorem{theorem}{Theorem}[section]
\newtheorem{proposition}[theorem]{Proposition}
\newtheorem{lemma}[theorem]{Lemma}
\newtheorem{corollary}[theorem]{Corollary}

\theoremstyle{definition}
\newtheorem{defn}[theorem]{Definition} 
\newtheorem{remark}[theorem]{Remark} 

\newtheorem{oproblem}[theorem]{Open Problem}
\newtheorem{notation}[theorem]{Notation}
\newtheorem{example}[theorem]{Example}

\DeclareMathOperator{\Sub}{\mathsf{Sub}}

\begin{document}
\FXRegisterAuthor{sm}{asm}{SM}
\FXRegisterAuthor{ja}{aja}{JA}
\FXRegisterAuthor{lm}{alm}{LM}

\title[Terminal Coalgebras in Countably Many Steps]{Terminal Coalgebras in Countably Many Steps}
\thanks{Stefan Milius acknowledges funding by the Deutsche Forschungsgemeinschaft (DFG, German  Research Foundation) -- project number 470467389}

\specialissue{Selected Papers of the 11th Conference on Algebra and Coalgebra in Computer Science (CALCO 2025)}

\author[J.~Ad\'amek]{Ji\v{r}\'i Ad\'amek\lmcsorcid{0000-0002-1721-3155}}[a]
\author[S.~Milius]{Stefan Milius\lmcsorcid{0000-0002-2021-1644}}[b]
\author[L.S.~Moss]{Lawrence S.~Moss\lmcsorcid{0000-0002-9908-5774}}[c]

\address{Computer Science Department, Czech Technical University, Czech Republic and
  Technical University Braunschweig, Germany}
\email{j.adamek@tu-braunschweig.de}

\address{Department of Computer Science, Friedrich-Alexander-Universität Erlangen-N\"{u}rnberg, Germany}
\email{stefan.milius@fau.de}

\address{Mathematics Department, Indiana University, Bloomington IN, USA}
\email{lmoss@iu.edu}

 
 \newcommand{\overbar}[1]{\mkern 1.5mu\overline{\mkern-1.5mu#1\mkern-1.5mu}\mkern 1.5mu}
%
\newcommand{\mybar}[3]{%
  \mathrlap{\hspace{#2}\overline{\scalebox{#1}[1]{\phantom{\ensuremath{#3}}}}}\ensuremath{#3}
}

\newcommand{\myhat}[3]{%
  \mathrlap{\hspace{#2}\widehat{\scalebox{#1}[1]{\phantom{\ensuremath{#3}}}}}\ensuremath{#3}
}

\newcommand{\mytilde}[3]{%
  \mathrlap{\hspace{#2}\widetilde{\scalebox{#1}[1]{\phantom{\ensuremath{#3}}}}}\ensuremath{#3}
}

\newcommand{\barF}{\mybar{0.6}{2.5pt}{F}}
\newcommand{\barM}{\mybar{0.7}{2.9pt}{M}}
\newcommand{\bpartial}{\overline{\partial}}
\newcommand{\tildeM}{\mytilde{0.7}{2.9pt}{M}}

\newcommand{\xra}[1]{\xrightarrow{~#1~}}
\newcommand{\xla}[1]{\ensuremath{\xleftarrow{~#1~}}}

%
%

\newcommand{\smooth}{smooth\xspace}
\newcommand{\opcit}[1][.]{\textit{op.cit#1}\xspace}

\newcommand{\ol}[1]{\overline{#1}}

\newcommand{\op}[1]{\operatorname{\mathsf{#1}}}
\newcommand{\id}{\op{id}}

\newcommand{\inj}{\op{in}}
\newcommand{\inl}{\op{inl}}
\newcommand{\inr}{\op{inr}}
\newcommand{\pr}{\op{pr}}

\newcommand{\monoto}{\ensuremath{\rightarrowtail}}
\newcommand{\epito}{\ensuremath{\twoheadrightarrow}}
\newcommand{\subto}{\ensuremath{\hookrightarrow}}

\newboolean{bolt}
\setboolean{bolt}{false}

\ifbolt
\usepackage{fontawesome}
\newcommand{\ACDCem}{\emph{\mbox{AC\hspace*{2pt}\faBolt\hspace*{-1pt}DC}}\xspace}
\newcommand{\ACDC}{\mbox{AC\,\faBolt\/DC}\xspace}
\else
\newcommand{\ACDC}{\mbox{AC/DC}\xspace}
\newcommand{\ACDCem}{\emph{\ACDC}}
\fi

\newcommand{\cat}[1]{\mathscr{#1}}
\def\A{\cat A}
\def\B{\cat B}
\def\C{\cat C}
\def\D{\cat D}
\newcommand{\E}{\mathcal{E}}
\newcommand{\FF}{\mathcal{F}}
\newcommand{\MM}{\mathcal{S}}
\newcommand{\Set}{\mathsf{Set}}
\newcommand{\Setp}{\mathsf{Set}_\mathsf{p}}
\newcommand{\F}{\mathcal{F}}
\newcommand{\Bij}{\mathcal{B}}
\newcommand{\Inj}{\mathcal{I}}
\newcommand{\Surj}{\mathcal{S}}
\newcommand{\SetF}{\Set^{\F}}
\newcommand{\SetB}{\Set^{\Bij}}
\newcommand{\Pos}{\mathsf{Pos}}
\newcommand{\JSL}{\mathsf{JSL}}
\newcommand{\Pfn}{\mathsf{Pfn}}
\newcommand{\SMod}{\Srg\text{-}\mathsf{Mod}}
\newcommand{\Gra}{\mathsf{Gra}}
\newcommand{\KVec}{K\text{-}\mathsf{Vec}}
\newcommand{\CPO}{\mathsf{CPO}}
\newcommand{\DCPO}{\mathsf{DCPO}}
\newcommand{\DCPOb}{\mathsf{DCPO}_\bot}
\newcommand{\CMS}{\mathsf{CMS}}
\newcommand{\Met}{\MS}
\newcommand{\Nom}{\mathsf{Nom}}
\newcommand{\MS}{\mathsf{Met}}
\newcommand{\UMet}{\mathsf{UMet}}
\newcommand{\Bool}{\ensuremath{\mathsf{Bool}}}
\newcommand{\Rel}{\mathsf{Rel}}
\newcommand{\Ab}{\mathsf{Ab}}
\newcommand{\Mod}{\mathsf{Mod}}
\newcommand{\CMet}{\CMS}
\newcommand{\Ord}{\mathsf{Ord}}
\newcommand{\Top}{\ensuremath{\mathsf{Top}}}
\newcommand{\KHaus}{\ensuremath{\mathsf{KHaus}}}
\newcommand{\Khaus}{\KHaus}
\newcommand{\Haus}{\mathsf{Haus}}
\newcommand{\Stone}{\mathsf{Stone}}
\newcommand{\ZSet}{\Z\text{-}\Set}

\newcommand{\Id}{\mathsf{Id}}
\newcommand{\Pow}{\mathscr{P}}
\def\pow{\Pow}
\newcommand{\Powf}{\Pow_\mathsf{f}} 
\def\powf{\Powf}
\def\powfin{\Powf}
\newcommand{\powcl}{\Pow_{\mathsf{cl}}}
\newcommand{\V}{\mathscr{V}}
\renewcommand{\H}{\mathcal{H}}

\newcommand{\M}{\mathcal{M}}

\newcommand{\N}{\mathds{N}}
\def\Nat{\N}
\newcommand{\R}{\mathds{R}}
\newcommand{\Z}{\mathds{Z}}
\newcommand{\Srg}{\mathds{S}}
\newcommand{\fpair}[1]{\ensuremath{\langle #1 \rangle}}

\newcommand{\Coalg}{\mathop{\mathsf{Coalg}}}
\newcommand{\Alg}{\mathop{\mathsf{Alg}}}
\newcommand{\colim}{\mathop{\mathsf{colim}}}

\newcommand{\ter}{\tau}
\newcommand{\ini}{\iota}

\newcommand*\cocolon{%
        \nobreak
        \mskip6mu plus1mu
        \mathpunct{}%
        \nonscript
        \mkern-\thinmuskip
        {:}%
        \relax
}

\newcommand{\HF}{H\!F}
\newcommand{\arity}{\mathop{\mathsf{ar}}}
\newcommand{\set}[1]{\{#1\}}

\newcommand{\Abar}{B}
\newcommand{\dbar}{d'}
\newcommand{\eps}{\varepsilon}
\newcommand{\opp}{\mathsf{op}}
\newcommand{\mhat}{\widehat{m}}
\newcommand{\ehat}{\widehat{e}}
\newcommand{\chat}{\widehat{c}}  
\renewcommand{\o}{\cdot}
\newcommand{\temptree}{\mathsf{tr}}
%
%
\newcommand{\takeout}[1]{\empty}

\newcommand{\Vbar}{\overline{V}}
\newcommand{\vbar}{\overline{v}}
\newcommand{\mbar}{\overline{m}}
\newcommand{\descto}[3][]{\arrow[phantom]{#2}[#1]{\text{\footnotesize{}\begin{tabular}{c}#3\end{tabular}}}}
\newcommand{\desctox}[4][]{\arrow[phantom,#2]{#3}[#1]{\text{\footnotesize{}\begin{tabular}{c}#4\end{tabular}}}}

\newcommand{\pair}[1]{\langle {#1} \rangle}


 \keywords{terminal coalgebra, countable iteration, descending chain condition, Hausdorff endofunctor, Vietoris endofunctor, }

\begin{abstract}
  We present a collection of results that imply that an endofunctor on a category has a
  terminal coalgebra obtainable as a countable limit of its terminal-coalgebra chain.
  This holds for finitary endofunctors on locally
  finitely presentable categories under conditions on both the functor and the category.
  We
  adapt finiteness arguments that were originally advanced by Worrell concerning terminal
  coalgebras for finitary set functors.  Examples include the categories of sets, posets,
  vector spaces, graphs, nominal sets, and presheaves on finite sets. Worrell also
  described, without proof, the terminal-coalgebra chain of the finite power-set functor. We
  provide a detailed proof following his ideas.

  We then turn to polynomial endofunctors on the categories of Hausdorff topological spaces
  and metric spaces.  The Vietoris space of compact subsets yields an endofunctor~$\V$ on
  the category of Hausdorff spaces. Vietoris polynomial endofunctors are built from~$\V$,
  the identity and constant functors by forming products, coproducts and compositions. Their
  terminal coalgebras are obtained in $\omega$ steps. We then turn to the class of Hausdorff
  polynomial functors on the category of metric spaces, which is analogous but uses in lieu
  of~$\V$ the Hausdorff functor $\H$.  We prove they have terminal coalgebras obtained in
  $\omega + \omega$ steps. Finally, we show that every finitary endofunctor on the category
  of vector spaces over a fixed field again has a terminal coalgebra obtained in
  $\omega+\omega$ steps.
\end{abstract}

\maketitle


\section{Introduction}
\label{S:intro}

Coalgebras capture various types of state-based systems in a uniform way by encapsulating the
type of transitions as an endofunctor on a suitable base category. Coalgebras also come with a
canonical behaviour domain given by the notion of a terminal coalgebra. So results on the
existence and construction of terminal coalgebras for endofunctors are at the heart of the
theory of universal coalgebra.  The topic is treated in our monograph~\cite{AMM25}.  A
well-known construction of the terminal coalgebra for an endofunctor was first presented by
Ad\'amek~\cite{A74} (in dual form) and independently by Barr~\cite{barr}. The idea is to
iterate a given endofunctor $F$ on the unique morphism $F1 \to 1$ to obtain the following
$\omega^\opp$-chain
\begin{equation}\label{eq:op-chain}
  1 \xla{!} F1 \xla{F!} FF1 \xla{FF!} FFF1 \xla{FFF!} \cdots
\end{equation}
and then continue transfinitely.  For each ordinal $i$, we write $V_i$ for the $i$th
iterate. Hence, 
\begin{equation}\label{eq:trans}
  V_0 = 1,\qquad V_{i+1} = FV_i, \qquad
  \text{and}
  \qquad
  \textstyle\text{$V_i = \lim_{j<i} V_j$ when $i$ is a limit ordinal};
\end{equation}
the connecting morphisms are as expected. In particular, for every ordinal~$i$, we have a
morphism $FV_i \to V_i$. If the transfinite chain converges in the sense that this morphism
is an isomorphism for some $i$, then its inverse is the structure of a terminal coalgebra for
the functor~$F$~\cite[dual of second prop.]{A74}.  This happens for a limit ordinal $i$
provided that~$F$ preserves the limit $V_i$. However, in general, this transfinite chain does
not converge at all (e.g.~for the power-set functor), and moreover, if it does converge, then
the number of iterations needed to obtain the terminal coalgebra can be arbitrarily large. For
example, the set functor~$\pow_{\alpha}$, which assigns to a set the set of all subsets of
cardinality smaller than $\alpha$, requires $\alpha + \omega$ iterations~\cite{almms}.

A famous result by Worrell~\cite{worrell:05} states that a finitary set functor needs at most
$\omega+\omega$ iterations to converge. We generalize this result to other base categories by isolating
properties of the category of sets and endofunctors on it that entail it:
\begin{enumerate}
\item The \emph{descending chain condition} (DCC), which states that for every finitely
presentable object (a category-theoretic generalization of the notion of a finite set) every
strictly decreasing chain of subobjects or strong quotient
objects is finite.

\item The preservation of \emph{nonempty binary intersections}, that is, pullbacks of two monomorphisms such that the domain is not a strict initial object (cf.~\cref{D:nonempty}).
\end{enumerate}

The first condition is inspired by the descending chain condition in algebra. Regarding the
second one, it was shown by Trnkov\'{a} that every set functor preserves nonempty binary
intersections~\cite{trnkova69}.  In addition, every finitary set functor preserves
\emph{all} nonempty intersections~\cite[Thm.~4.4.3]{AMM25}, that is, wide pullbacks of
  monomorphisms whose domain is not a strict initial object.


Our first main result (\cref{T:main}) holds for locally finitely presentable categories
satisfying the DCC: for every finitary endofunctor preserving nonempty binary intersections,
the terminal-coalgebra chain converges in $\omega + \omega$ steps. We also show that the DCC is
satisfied by a large number of categories of interest, such as sets, posets, graphs, vector
spaces, boolean algebras, nominal sets, and presheaves on finite sets.

The category of metric spaces and non-expanding maps is not locally finitely presentable, and
so \cref{T:main} is not applicable to it. Nevertheless, we provide in~\cref{T:MetOplusO}
a sufficient condition which implies that the terminal-coalgebra chain of an endofunctor converges in
$\omega + \omega$ steps: the endofunctor should be finitary and preserve nonempty binary intersections
(as in our~\cref{T:main}), and it also should preserve isometric embeddings.

We are also interested in other variations on Worrell's method.  We take an endofunctor~$F$
which has a set of desirable properties and consider \emph{polynomials in $F$}.  By this we
mean the functors built from~$F$ and the constant functors using product, coproduct,
and composition.  These variations are presented in \Crefrange{S:Kripke}{S:Hausdorff}.
They use our sufficient condition (\Cref{P:modest}) but not $\DCC$-categories.

On $\Set$, we will be concerned with polynomials in $\powfin$; we call these \emph{Kripke
  polynomial functors}, following Jacobs~\cite{jacobs}.  On the category~$\Top$ of topological
spaces, a good analog of~$\powf$ is the \emph{Vietoris functor} $\V$ assigning to every space
$X$ the space of all compact subsets equipped with the Vietoris topology (\Cref{S:Vietoris}).
The resulting class of \emph{Vietoris polynomial} functors was first defined by Hofmann et
al.~\cite{HNN19}.  We also study the category $\MS$ of metric spaces and non-expanding
maps. The role of the Vietoris functor is played there by the \emph{Hausdorff functor} $\H$
assigning to every space~$X$ the space $\H X$ of all compact subsets with the Hausdorff metric.
We shall see that Kripke and Hausdorff polynomial functors have $\nu F = V_{\omega+\omega}$,
whereas Vietoris polynomials functors have the stronger bound $\nu F = V_{\omega}$.

A concrete example of a terminal coalgebra related to metric labelled transition systems is
presented in \Cref{S:example}.

  
\mysubsec{Other contributions.}
On the category $\Haus$ of Hausdorff spaces we prove that $\V\colon \Haus \to \Haus$ preserves
limits of $\omega^{\opp}$-chains.  Suppose that a Vietoris polynomial functor $F$ has the
property that all the constants involved in its construction are complete spaces (or Hausdorff
spaces, compact spaces).  Then $\nu F$ again turns out to have this property.  We present a
proof of the description of $\nu\powfin$ and $V_\omega$ for $\powfin$ in terms of trees
mentioned by Worrell~\cite{worrell:05} (the latter without a proof).  We give a concrete
representation of a terminal coalgebra of an endofunctor on metric spaces that again uses
trees.  We simplify a proof of a known negative result: the variation of $\H$ obtained by
moving from compact sets to closed sets has no fixed~points.

\mysubsec{Related work.}
Previous conference papers have appeared containing the material in
\Crefrange{S:Kripke}{S:Hausdorff}~\cite{AdamekEA23} and
\Cref{S:finitary,S:omega+omega}~\cite{AdamekEA25}. \Cref{S:Nom-SetF} presents two new and
important examples of $\DCC$-categories. In addition, the material in
\cref{S:compactness,S:cofree-comonad} as well as \cref{S:example} are new.

As already stated, our main result in \cref{S:omega+omega} generalizes Worrell's theorem beyond the category of sets. 
Our work in \cref{S:Vietoris} is more general and hence improves results by
Abramsky~\cite{abramsky} and Hofmann et al.~\cite{HNN19}.

Our DCC condition was inspired by the Noetherian condition used by Urbat and
Schr\"{o}der~\cite{UrbatS20}. They say that an object is \emph{Noetherian} if its posets of
subobjects and quotients satisfy the ascending chain condition. Since they use the reverse
order on quotients, their condition on quotients is the same as ours on strong quotients; on
subobjects, their condition is dual to ours. Moreover, the results here are disjoint from
the ones in op.~cit.

A slightly stronger condition than our DCC was introduced in previous work~\cite{AdamekSousa24}. The relationship of the two conditions is discussed in \cref{S:DCC}.

Another related result concerns the category of complete metric spaces: for every locally contracting
endofunctor $F$ on this category satisfying $F\emptyset \neq \emptyset$, the terminal-coalgebra
chain converges in $\omega$ steps~\cite{are} (see also~\cite[Cor.~5.2.18]{AMM25}). Moreover,
the ensuing terminal coalgebra is then also an initial algebra.



\section{Preliminaries}
\label{S:preliminaries}

We assume that readers are familiar with basic notions of category theory as well as
algebras and coalgebras for an endofunctor. In \Cref{S:finitary}, we assume familiarity with
locally finitely presentable categories. The material presented in this section
(except~\cref{T:initial}) is standard, and we present it to facilitate reading the
subsequent sections.

We denote by $\Set$ the category of sets and functions, $\Top$ is the category of topological
spaces and continuous functions, and $\Met$ is the category of \emph{(extended) metric spaces}
(so we might have $d(x,y) = \infty$) and non-expanding maps: the functions $f\colon X \to Y$
where $d(f(x),f(x')) \leq d(x,x')$ holds for every pair $x, x' \in X$.  Note that this class of
morphisms is smaller than the class of continuous functions between metric spaces.  Finally,
$\KVec$ is the category of vector spaces over a fixed but arbitrary field $K$, using linear
maps as morphisms.

An epimorphism
    $e\colon A \to B$ is \emph{strong}
    if it satisfies the following
    \emph{diagonal fill-in property}:
    given a
    monomorphism \mbox{$m\colon C \monoto D$} and morphisms $f\colon A \to C$
    and $g\colon B \to D$ such that~\mbox{$m\o f = g\o e$} (so that the outside
    of the square below commutes),  there exists a
    unique `diagonal'~$d\colon B \to C$ such that the diagram below commutes:
    \begin{equation}\label{diag:diag}
      \begin{tikzcd}
        A
        \arrow[->>]{r}{e}
        \arrow{d}[swap]{f}
        &
        B
        \arrow{d}{g}
        \arrow[dashed]{ld}[swap]{d}
        \\
        C
        \arrow[>->]{r}{m}
        &
        D
      \end{tikzcd} 
    \end{equation}

We write $S\monoto X$ for monomorphisms and $X \epito E$ for strong epimorphisms.  Given an
endofunctor~$F$, we write $\nu F$ for its terminal coalgebra, if it exists.

Regarding the $\omega^\opp$-chain in~\eqref{eq:op-chain}, let
$\ell_n\colon V_\omega \to F^n 1$ ($n < \omega$) be the limit cone.  We obtain a unique
morphism \mbox{$m\colon FV_\omega \to V_\omega$} such that for all $n \in \omega^\opp$, we have
\begin{equation}
  \label{diag:m}
  \begin{tikzcd}
    FV_\omega
   \ar{rr}{m}
    \ar{rd}[swap]{F\ell_n}
    & &
    V_\omega
    \arrow{dl}{\ell_{n+1}}
    \\
    & F^{n+1} 1
  \end{tikzcd}
\end{equation} 
This is the connecting morphism from $V_{\omega+1} = FV_\omega$ to $V_\omega$ in the
transfinite chain~\eqref{eq:trans}.

If $F$ preserves the limit $V_\omega$, then $m$ is an isomorphism (and conversely).
Therefore, its inverse yields the terminal coalgebra
$m^{-1}\colon V_\omega \to FV_\omega$~\cite[dual of second prop.]{A74}; shortly
$\nu F = V_\omega$. This is used e.g.~in the proof of \cref{T:Vietoris}.

This technique of \emph{finitary iteration} is the most basic and prominent construction of
terminal coalgebras.  However, it does \emph{not} apply to the finite power-set
functor~$\powf$.  For that functor $FV_{\omega} \not\cong
V_{\omega}$~\cite[Ex.~3(b)]{ak95}. However, a modification of finitary iteration does apply, as
shown by Worrell~\cite[Th.~11]{worrell:05}.  One needs a \emph{second infinite iteration},
iterating $F$ on the morphism $m\colon FV_{\omega} \to V_{\omega}$ rather than on
$!\colon F1 \to 1$, obtaining the $\omega^\opp$-chain
\begin{equation}\label{diag:mm}
  V_\omega
  \xla{m}
  V_{\omega + 1}
  \xla{Fm}
  V_{\omega + 2}
  \xla{FFm}
  \cdots.
\end{equation}
Its limit is denoted by
\begin{equation}\label{eq:omega+omega}
  V_{\omega+\omega} = \lim_{n < \omega} V_{\omega+n}.
\end{equation}

Again, if $F$ preserves this limit, then we obtain that $V_{\omega + \omega}$ carries a
terminal coalgebra; shortly $\nu F = V_{\omega +\omega}$. Worrell proved that this holds for
every finitary set functor.

In general the \emph{terminal-coalgebra chain} is defined by transfinite
recursion: its objects are given by~\eqref{eq:trans} and the connecting morphisms are defined
by 
\[
  \begin{array}{l}
    \text{$v_{1,0}\colon V_1 \to 1$ is unique,}
    \qquad
    \text{$v_{k+1,j+1} = Fv_{k,j}\colon FV_k \to  FV_j$, and}
    \\
    \text{$v_{i,j}$ ($i>j$) is the limit cone for every limit ordinal $i$.}
  \end{array}
\]
We say that the terminal-coalgebra chain \emph{converges in $\lambda$ steps} if
$v_{\lambda+1,\lambda}$ is an isomorphism.

\subsection{Limits of $\omega^\opp$-chains}
We shall frequently use the following characterization of limits of $\omega^\opp$-chains,
e.g.~in \cref{S:Kripke,S:Vietoris,S:Hausdorff}.
\begin{remark}\label{P:chain}
  Consider an $\omega^\opp$-chain
  \begin{equation}\label{Xsequence}
    X_0 \xla{f_0} X_1 \xla{f_1} X_2 \xla{f_2} \cdots.
  \end{equation}
  In $\Set$, $\Top$, $\Met$, and $\KVec$, the limit $L$ is carried by the set of all sequences
  $(x_n)_{n < \omega}$, $x_n \in X_n$ that are \emph{compatible}: $f_{n}(x_{n+1}) = x_n$ for
  every $n$.  The limit projections are the functions $\ell_n\colon L \to X_n$ defined by
  $\ell_n((x_i)) = x_n$.
  \begin{enumerate}
  \item\label{R:Vietoris-tech:1} In $\Top$, the topology on $L$ has as a base the sets
    $\ell_n^{-1}(U)$, for $U$ open in $X_n$.
    
  \item\label{R:Haus:4} In $\Met$,  the metric on $L$
    is defined by $d((x_n), (y_n)) = \sup_{n <\omega} d(x_n,y_n)$.
    
  \item \label{R:Kvec:4} In $\KVec$, the limit $L$ is a subspace of the space $\Pi_i X_i$, and
   the projections $\ell_n$ are linear.
  \end{enumerate}
\end{remark}

\subsection{Cofree Comonads}
\label{S:introduce-cofree-comonads}

A closely related topic to terminal coalgebras are cofree comonads.  Given an endofunctor $F$
on a category $\A$, a comonad $F_{\sharp}$ together with a natural transformation
$\eps\colon F_{\sharp}\to F$ is \emph{cofree} provided that for every comonad $C$ and every
natural transformation $\phi\colon C\to F$, there is a unique comonad morphism
$\phibar\colon C\to F_{\sharp}$ for which the triangle below commutes:
\[
  \begin{tikzcd}
    F_{\sharp}
    \ar{r}{\eps}
    &
    F
    \\
    C
    \arrow[dashed]{u}{\phibar}    
    \ar{ru}[swap]{\phi}  
  \end{tikzcd}
\]

  \begin{proposition}[Dual to Barr~\cite{barr0}]\label{P:cofree}
    An endofunctor $F$ on a cocomplete category generates cofree
    comonads iff all cofree coalgebras exist, that is,
    $U_F\colon \Coalg F \to \A$ has a right adjoint.
  \end{proposition}
%

Moreover, if the given category has finite products, then a cofree
coalgebra on an object~$Y$ is precisely a terminal coalgebra for the
endofunctor $F(-) \times Y$.

\begin{corollary}\label{C:cofree}
  Let $\A$ be a cocomplete category with finite products. An
  endofunctor $F$ generates a cofree comonad iff for every object $Y$
  of $\A$ a terminal coalgebra for $F(-) \times Y$ exists.
  Moreover, the cofree comonad $F_\sharp$ is given by
  \(
  F_\sharp Y = \nu(F(-)\times Y).
  \)
\end{corollary}


From the above the following construction of $F_\sharp$ was
derived~\cite[Thm.~23.3]{kelly}. Let $F$ be an endofunctor on a
complete category $\A$. Define an ordinal-indexed chain $F_i$
($i \in \Ord^\opp$) of endofunctors and connecting natural
transformations $f_{i,j}\colon F_i \to F_j$ ($i\geq j$) by transfinite
recursion: put
\[
  \begin{array}{l@{\,}l@{\qquad}l}
    F_0 &= \Id,\\
    F_{j+1} &= FF_j \times \Id& \text{for all ordinals $j$,}\\
    F_j &= \lim_{i<j} F_i & \text{for all limit ordinals $j$, and}
    \\[10pt]
    \multicolumn{3}{l}{
      \text{$f_{1,0}\colon F_1 = F\times \Id \to \Id$ is the projection,}}\\
    \multicolumn{3}{l}{f_{k+1,j+1} = Ff_{k,j} \times \id\colon FF_k
      \times \Id \to FF_j \times \Id,}\\
    \multicolumn{3}{l}{\text{$f_{j,i}\colon F_j \to F_i$ ($j>i$) is the limit cone for
        every limit ordinal $j$.}}
  \end{array}
\]
This generalizes the terminal-coalgebra chain of $F$.
If $F$ preserves the limit $F_\omega X = \lim_{n<\omega} F_nX$ for
every object $X$, then it generates a cofree comonad carried by
$F_\omega$, and we say that the cofree comonad is \emph{obtained in
  $\omega$ steps.}

Analogously, if $F$ preserves the limit $F_{\omega + \omega} X=
\lim_{n< \omega} F_{\omega +n} X$ for every object $X$, then the
cofree comonad is carried by $F_\sharp = F_{\omega+\omega}$, and we
say that it is \emph{obtained in $\omega+\omega$ steps}.

\subsection{Locally finitely presentable categories}
\label{S:lfp}

We continue with a terse review of locally finitely presentable categories needed in
\cref{S:finitary}; see~\cite{ar} for background. A diagram $\D\to \A$ is \emph{directed} if
its domain~$\D$ is a directed poset (i.e.~nonempty and such that every pair of elements has
an upper bound). A functor is \emph{finitary} if it preserves directed colimits. An object
$A$ of a category $\A$ is \emph{finitely presentable} if its hom-functor
$\A(A,-)\colon \A \to \Set$ preserves directed colimits. A category is \emph{locally
  finitely presentable} (lfp, for short) if it is cocomplete and has a set of finitely
presentable objects such that every object is a directed colimit of objects from that set.

\begin{example}\label{E:lfp}
  We list a number of examples of lfp categories.
  \begin{enumerate}
  \item The category $\Set$ of all sets and $\Setp$ of pointed sets; the finitely presentable
    objects are precisely the finite sets.
    
  \item The category $\Gra$ of graph and their homomorphisms as well as $\Pos$ of posets and
    monotone maps;  finitely presentable objects are precisely the finite graphs or
    posets, respectively.

  \item Every finitary variety, that is, any category of algebras specified by operations of
    finite arity and equations; the finitely presentable objects are precisely those
    algebras which have a presentation by finitely many generators and relations (in the usual
    sense of universal algebra).
    The following three items are instances of this one.
     
  \item The category $\Bool$ of Boolean algebras and their homomorphisms;  the finitely
    presentable objects are precisely the finite Boolean algebras. The same holds for every
    locally finite variety, e.g.~join-semilattices or distributive lattices.

  \item The category $M$-$\Set$ of sets with an action of a monoid
    $M$, and equivariant maps; the finitely presentable objects are precisely the
    orbit-finite $M$-sets (i.e.~those having finitely many orbits).

  \item The category $\KVec$ of vector spaces over a field $K$ and linear maps;
  the finitely presentable objects are precisely the finite-dimensional vector spaces.

    More generally, given a semiring $\Srg$, the category $\SMod$ of all $\Srg$-semimodules is lfp.

  \item The category $\Nom$ of nominal sets and equivariant maps; the finitely
    presentable objects are precisely the orbit-finite nominal sets.

  \item The category $\Set^\A$ of (covariant) presheaves on a small category $\A$. 
    
  \item A poset, considered as a category, is lfp iff it is an algebraic lattice: a complete lattice in which every element is a join of compact ones.  (An element $x$ is \emph{compact} if for every subset~$S$, $x \leq \bigvee S$ implies that $x \leq\bigvee S'$ for some  finite $S' \subseteq S$.)
  \end{enumerate}
\end{example}

\begin{remark}\label{R:subobj}
  We next recall definitions concerning subobjects.
  \begin{enumerate}
  \item For a fixed object $A$, the monomorphisms with codomain $A$ have a natural preorder:
    given $c\colon C\monoto A$ and $c'\colon C'\monoto A$, we say that $c \leq c'$ iff
    $c = c' \o m$ for some monomorphism $m\colon C\to C'$.  A \emph{subobject} of $A$ is an
    equivalence class of monomorphisms under the induced equivalence relation. We write
    representatives to denote subobjects.
    
  \item A subobject (represented by) $c\colon C \monoto A$ is \emph{finitely presentable} if
    its domain $C$ is a finitely presentable object.
  \end{enumerate}
\end{remark}

\begin{remark}\label{rem-LP}
  We recall properties of an lfp category $\A$ used in the proof of \Cref{T:main}:
  \begin{enumerate}
  \item\label{rem-LP-1} $\A$ is complete~\cite[Rem~1.56]{ar} (and cocomplete by definition).

  \item\label{rem-LP-2} $\A$ has a $(\text{strong-epi}, \text{mono})$-factorization system~\cite[Rem.~1.62]{ar}.

  \item\label{rem-LP-3} Every morphism from a finitely presentable object to a
    directed colimit factorizes through one of the colimit maps.
    
  \item\label{rem-LP-4} Suppose that finitely presentable objects are closed under strong
    quotients in $\A$, that is, for every strong epimorphism $A \epito B$, if $A$ is
    finitely presentable, then so is $B$. (Equivalently, the class of all finitely
    presentable objects coincides with that of all finitely generated
    ones~\cite[Prop.~1.69]{ar}, where we recall~\cite[Def.~1.67]{ar} that an object $A$ is
    \emph{finitely generated} if its hom-functor $\A(A,-)\colon \A \to \Set$ preserves
    directed colimits of diagrams of monomorphisms.)  Then every object is the colimit of
    the canonical directed diagram of all of its finitely presentable
    subobjects~\cite[Lemma~3.1]{AdamekSousa24}.
    Moreover, given any finitely
    presentable subobject $c\colon C \monoto A$, it is easy to see that the object $A$ is
    the colimit of the diagram of all its finitely presentable subobjects
    $s\colon S \monoto A$ such that $c \leq s$.
    
  \item\label{rem-LP-5} The collection of all finitely presentable objects, up to
    isomorphism, is a set. It is a \emph{generator} of $\A$; it follows that a morphism
    $m\colon X \to Y$ is monic iff for every pair $u, v\colon U \to X$ of morphisms with a
    finitely presentable domain $U$, we have that $m \cdot u = m \cdot v$ implies $u = v$.
  %
  
  \end{enumerate}
\end{remark}

\begin{notheorembrackets}
  \begin{defn}[\cite{mpw20}]\label{D:nonempty}
    An initial object $0$ is \emph{strict} if every
    morphism with codomain $0$ is an isomorphism.   A monomorphism $A\monoto B$ is \emph{empty} 
    if its domain is a strict initial object; it is \emph{nonempty} if it is not empty.%

    An intersection (a wide pullback of monomorphisms) is \emph{empty} if its domain is a
    strict initial object, that is, the limit cone is formed by empty monomorphisms; the
    intersection is nonempty if it is not empty.\footnote{There is no condition on the 
      (non-)emptiness of the \emph{family} of monomorphisms which is intersected here.}
    
    An endofunctor $F\colon\A\to\A$ \emph{preserves nonempty intersections} if $F$ takes a nonempty
    intersection to a (not necessarily nonempty) wide pullback.
  \end{defn}
\end{notheorembrackets}
\begin{remark}\label{R:pres}
  Every endofunctor preserving nonempty binary intersections preserves non\-empty
  monomorphisms. This holds since a morphism is monic iff the pullback along itself is formed
  by a pair of identity morphisms.
\end{remark}

\begin{example}\label{E:nonempty}
  \begin{enumerate}
  \item In $\Set$, the initial object $\emptyset$ is strict.  A nonempty intersection is an
    intersection of subsets having a common element.  Trnkov\'{a}~\cite{trnkova69} proved that
    every set functor preserves nonempty binary intersections. It follows that every finitary
    set functor preserves nonempty intersections~\cite[Thm.~4.4.3]{AMM25}.
    
  \item The initial object $\set{0}$ in $\KVec$ is not strict.  Thus all subobjects are nonempty.
    Every endofunctor on $\KVec$ preserves finite intersections~\cite[Ex.~4.3]{AdamekSousa24}.
    
  \item In $\Gra$ and $\Pos$ nonempty intersections are, as in $\Set$, intersections of
    subobjects having a common element.
  \end{enumerate}
\end{example}

\begin{remark}
  \begin{enumerate}
  \item Unlike on $\Set$ and $\KVec$, on most everyday categories finitary endofunctors may
    fail to preserve nonempty intersections. For example, consider the category $\Gra$ of
    graphs.  We exhibit a finitary endofunctor not preserving nonempty binary intersections. We
    denote by $1$ the terminal graph, a single loop, and by $S$ a single node which has no
    loop. Let~$F$ be the extension of the identity functor with $FX=X$ if $X$ has no loop, else
    $FX=X+1$. The graph $1+1$ has subobjects $S+1$ and $1+S$ with the nonempty intersection
    $S+S$, but~$F$ does not preserve it.
    
  \item The collection of all finitary endofunctors on lfp categories preserving nonempty
    intersections is, nevertheless, large. It contains constant functors, finite power-functors
    $(-)^n$, for $n \in \N$, and it is closed under finite products and composites. It is also
    closed under coproducts provided that they commute with pullbacks (which holds in
    categories such as $\Set$, $\Pos$, $\Gra$, and $\Nom$).

  \item On the category $\Nom$, the abstraction functor (cf.~\cite[Thm.~4.12]{Pitts13}) and the
    finite power-set functor preserve intersections.%
    \smnote{An example of an endofunctor on $\Nom$ not preserving nonempty binary intersections
      is easy (I write it here to not forget it, but we have decided not to include it): define
      $FX = 1$ if $X$ has an element with empty support (i.e.~$X = X' +1$ for some nominal set
      $X'$) and $FX = X$ else. Then $1 + \names + 1$ has two subobjects $1 + \names$ and
      $\names + 1$ whose nonempty intersection is $\names$. But $F$ does not preserve this
      intersection.}
\end{enumerate}
\end{remark}

\begin{remark}\label{R:nonempty}
  Let $A$ be an object of a locally finitely presentable category in which finitely
  presentable objects are closed under strong quotients. 
  \begin{enumerate}
    \takeout{ 
    \item\label{R:nonempty:1} Every pair of finitely presentable objects $m_i\colon B_i \monoto A$
      ($i = 1,2$) has a union $m_1\cup m_2$: this is the smallest subobject which includes $m_1$
      and $m_2$.  Indeed, factorize $[m_1,m_2]\colon B_1 + B_2 \to A$ as a strong epimorphism
      $e\colon B_1+ B_2\epito B$ followed by a monomphism $m\colon B\monoto A$.  Then we take $m$
      to be $m_1\cup m_2$.  If either $m_i$ is nonempty, so is $m$: use the unique morphism
      $n\colon B_i \monoto B$ with $m_1 = m\o n$.}
    
  \item\label{R:nonempty:2} If $A$ is not strictly initial, then it has a nonempty finitely
    presentable subobject.  To see this, let $c_i\colon C_i \monoto A$ $(i\in I)$ be the
    colimit cocone of the diagram in~\cref{rem-LP}\ref{rem-LP-4}.  If each~$C_i$ is strictly
    initial, then so is the colimit $A$. Indeed, the colimit of any diagram of strict initial
    objects is itself strict initial.

  \item\label{R:nonempty:2a} Moreover, if $A$ is not strictly initial, then it is the directed
    colimit of the canonical diagram of all its \emph{nonempty} finitely presentable
    subobjects. To see this, combine \cref{rem-LP}\ref{rem-LP-4} with the previous item.
    
  \item\label{R:nonempty:3} If for some ordinal $i \leq \omega+\omega$ the object $V_i$ is
    strictly initial, then the terminal-coalgebra chain converges in $\omega+\omega$ steps by
    default.  Indeed, recall the transfinite chain $V_j$ from~\eqref{eq:trans}.  The connecting
    morphism from $V_{i+1} = FV_i$ to $V_i$ is an isomorphism, whence $\nu F = V_i$.
  \end{enumerate}
\end{remark}

\subsection{Smooth Monomorphisms}
In addition to terminal coalgebras, we also study initial algebras for the functors
of interest in this paper.  For this, we call on a general result which allows one
to infer the existence of the initial algebra for an endofunctor $F$ from the existence
of a terminal coalgebra for $F$ (or in fact of any algebra with monic
structure).

For a class $\M$ of monomorphisms we denote by $\Sub_\M(A)$ the
collection of subobjects of~$A$ represented by monomorphisms from
$\M$ (cf.~\cref{R:subobj}). To say that this is a dcpo means that it is a set which (when
ordered by factorization in the usual way) is a poset having directed joins.

\removeThmBraces
\begin{defn}[{\cite[Def.~3.1]{amm21}}]\label{D:smooth}
  Let $\M$ be a class of monomorphisms closed under isomorphisms and
  composition.
  \begin{enumerate}
  \item\label{D:smooth:1} We say that an object $A$ has \emph{\smooth
      $\M$-subobjects} provided that $\Sub_\M(A)$ is a dcpo with
    bottom~$\bot$, where the least element
    and directed joins are given by colimits of the corresponding
    diagrams of subobjects.

  \item\label{D:smooth:2} The class $\M$ is \emph{smooth} if every
    object of $\A$ has smooth $\M$-subobjects.
  \end{enumerate}
  A category has \emph{smooth monomorphisms} if the class of all
  monomorphisms is smooth.
\end{defn}
\resetCurThmBraces
\begin{example}\label{E:cms}
  \begin{enumerate}
  \item \label{E:cms:set} The categories $\Set$ and $\Top$ have smooth
    monomorphisms, and so does the full subcategory of Hausdorff
    spaces. This is easy to see.
    
  \item\label{E:cms:met}
    The category $\Met$ also has smooth monomorphisms (these are the injective
    non-expanding maps)~\cite[Lemma~A.1]{amm21}.
    
    The full subcategory $\CMS$ of complete metric spaces does not have
    smooth monomorphisms. However, strong monomorphisms 
    are smooth in both $\Met$ and $\CMS$~\cite[Lemma~A.2]{amm21}.
These are the \emph{isometric embeddings}: morphisms which preserve distances.
    
  \item Strong monomorphisms (subspace embeddings) in $\Top$ are not
    smooth~\cite[Ex.~3.5]{ahrt23}. 

  \item\label{E:cms:khaus} In the full subcategory of $\Top$ given by compact Hausdorff
    spaces, monomorphisms are not smooth. To see this, consider the $\omega$-chain of all
    finite discrete spaces $\set{0,\ldots, n-1}$ ($n < \omega$) whose colimit is $\beta\N$,
    the Stone-\v{C}ech compactification of the natural numbers. Now for the one-point
    compactification $\N_\infty$ of $\N$, the inclusions
    $\set{0,\ldots,n-1} \subto \N_\infty$ ($n <\omega$) form a cocone of monomorphisms;
    however, the unique factorizing morphism $\beta\N \to \N_\infty$ is not monic.
  \end{enumerate}
\end{example}
\removeThmBraces
\begin{theorem}[{\cite[Cor.~4.4]{amm21}}]\label{T:initial}
  Let $\M$ be a smooth class of monomorphisms. If an endofunctor~$F$
  preserving~$\M$ has a terminal coalgebra, then it has an initial algebra.
\end{theorem}
\resetCurThmBraces
\noindent
Note that loc.~cit.~states more: given any algebra
$m\colon FA \monoto A$ where $m$ lies in $\M$, the initial algebra
exists and is a subalgebra of $(A,m)$.

\subsection{A Sufficient Condition for $\nu F = V_{\omega+\omega}$}
\label{S:finitary}

We first present a simple result that holds for all  endofunctors of all categories.
This result will then be used several times in the sequel.
In it, we recall the notation $V_{\omega+\omega}$ from \eqref{diag:mm}.
Following this, we introduce $\DCC$-categories, and prove a generalization of Worrell's result for them
(\Cref{T:main}).

\begin{proposition}\label{P:modest}
  Let $\A$ be a category with terminal object $1$.  Let $F\colon\A\to\A$ be an endofunctor, let
  $\M$ be a class of monomorphisms closed under composition, and suppose that the following
  hold:
  \begin{enumerate}
  \item\label{P:modest:1}
    The limit $V_{\omega}$ of the $\omega^{\opp}$-chain~\eqref{eq:op-chain}
    exists, and the canonical morphism $m\colon FV_{\omega}\to V_{\omega}$ belongs to~$\M$. 
  \item\label{P:modest:2} $F$ preserves $\M$: if $m$ belongs to $\M$, so does $Fm$.
  \item\label{P:modest:3} $\A$ has and $F$ preserves intersections of
    $\omega^\opp$-limits of $\M$-subobjects.
  \end{enumerate}
  Then $V_{\omega + \omega}$ exists, and it is preserved by $F$.  Moreover,
  $\nu F = V_{\omega+\omega}$.
\end{proposition}
\begin{proof} 
  Let $V_i$ be defined for all ordinals $i$ by $V_0 = 1$, $V_{i+1} = FV_i$, and
  $V_i = \lim_{j<i} V_j$ for limit ordinals $i$.  The $\omega^{\opp}$-chain
  \eqref{eq:op-chain} is its beginning, \eqref{diag:m} defines the connecting morphism
  $m\colon V_{\omega+1}\to V_{\omega}$, and the $\omega^{\opp}$-chain~\eqref{diag:mm}, repeated
  below, is the continuation of the chain in~\eqref{eq:op-chain} up to
  $V_{\omega+\omega} = \lim_{i< \omega+\omega}V_{i}$:
  \[
    V_\omega
    \xla{m}
    V_{\omega + 1}
    \xla{Fm}
    V_{\omega + 2}
    \xla{FFm}
    \cdots.
  \]
  From \cref{P:modest:1,P:modest:2}, the morphisms in the chain~\eqref{diag:mm} belong to $\M$.
  Since $\M$ is closed under composition, that chain is an $\omega^{\opp}$-chain of subobjects
  of $V_{\omega}$ carried by members of~$\M$. Its limit is its intersection.  By
  \cref{P:modest:3}, $F$ preserves this limit, and it follows that $\nu F =
  V_{\omega+\omega}$~\cite[dual of second prop.]{A74}.
\end{proof}

\section{$\DCC$-Categories}\label{S:DCC}
We introduce lfp categories satisfying a descending chain condition, shortly $\DCC$-categories.
Examples are presented and the related condition of graduatedness is discussed.
We prove that $\omega^{\opp}$-limits in $\DCC$-categories are finitary.
In  \cref{S:omega+omega}, we  prove that $\nu F = V_{\omega +\omega}$ for all finitary
endofunctors on $\DCC$-categories preserving nonempty binary intersections.


We have already seen the order of subobjects of a fixed object $A$ (cf.~\cref{R:subobj}).
(This corresponds to the preordered collection in the slice category $\A \downarrow A$.)
Dually, we use the order on strong quotients, represented by strong epimorphisms
$e\colon A \epito E$: given $e'\colon A \epito E'$, we have $e\leq e'$ iff $e' = u\o e$ for some
$u\colon E\to E'$.   This corresponds to the preordered collection in the slice category $A\downarrow\A$.
In the literature, the opposite order on quotients is also used.   For
example, Urbat and Schr\"{o}der~\cite{UrbatS20}, whose work has inspired our 
next definition, use that opposite order.  So readers of papers in this area should be careful.

\begin{defn}\label{D:DCC}
  A locally finitely presentable category $\A$ is a
  \emph{DCC-category} if finitely presentable objects are closed under strong quotients and
  every finitely presentable object~$A$ satisfies the following \emph{descending chain
    condition}: Every strictly descending chain of subobjects or strong quotients of $A$ is
  finite.
\end{defn}
\begin{remark}
  Recall~\cite[Def.~1.67]{ar} that an object $A$ of a category $\A$ is \emph{finitely
    generated} if its hom-functor $\A(A,-)\colon \A \to \Set$ preserves directed colimits of
  diagrams of monomorphisms. The condition in \cref{D:DCC} that for every strong epimorphism
  $A \epito B$, if $A$ is finitely presentable, then so is $B$ is equivalent to stating that
  the class of all finitely generated objects and that of all finitely presentable ones
  coincide; this follows from \cite[Prop.~1.69]{ar}.
\end{remark}

Our notion is also related to the stronger notion of graduatedness~\cite{AdamekSousa24}: a
locally finitely presentable category is \emph{graduated} if to every every finitely
presentable object $A$ a natural number $n$ is assigned, called the \emph{grade of $A$}, such
that every (proper) subobject and every (proper) strong quotient is finitely presentable, and
with a grade at most (smaller than, respectively) the grade of $A$.

\begin{proposition}\label{P:grad-DCC}
  Every graduated category is a $\DCC$-category.
\end{proposition}
\begin{proof}
  Let $A$ be a finitely presentable object having grade $n$. Then every strictly decreasing chain of
  subobjects or strong quotients of $A$ has length at most $n+1$. 
\end{proof}

The converse of~\cref{P:grad-DCC} does not hold:
\begin{example}
  Here is a $\DCC$-category which is not graduated.  Consider the
  poset $A$ with top element $\top$, bottom element $\bot$, and
  elements $a_{nm}$ ($n\leq m < \omega$) ordered as follows:
  $a_{ij} \leq a_{nm}$ iff $i = n$ and $j \leq m$.
  \[
    \tikzcdset{scale cd/.style={every label/.append style={scale=#1},
        cells={nodes={scale=#1}}}}
    \begin{tikzcd}[scale cd=.9,sep=small]
      & \top & & \\
      a_{00} \ar[-]{ur }& a_{11}\ar[-]{u}  & a_{22} \ar[-]{ul}& \\
      & a_{10}  \ar[-]{u}& a_{21}\ar[-]{u} & \cdots \\
      & & a_{20}\ar[-]{u} & \\
      & \bot\ar[-]{uu} \ar[-]{uuul}  \ar[-]{ur}& \\
    \end{tikzcd}   
  \]
  This is a complete lattice with all elements compact (i.e.~finitely
  presentable).  Thus, it is a locally finitely presentable category.
  The $\DCC$ condition is obvious.  But $\top$ cannot have a (finite)
  grade: its grade would have to be at least $2$, due to
  $\bot < a_0 < \top$, and at least $3$ due to
  $\bot < a_{10}< a_{11} < \top$, etc.
\end{example}
\takeout{ 
\begin{example}
  \begin{enumerate}
  \item The category $\Ab$ of abelian groups is locally finitely presentable but not
    $\DCC$. The group $\Z$ of integers is finitely presentable, but it has the following
    descending sequence of proper subgroups $\Z\supset 2 \Z\supset 4 \Z \supset \cdots$.

  \item Here is a $\DCC$-category which is not graduated: the category $\ZSet$ of sets
    with an action of the additive group $\Z$ of integers and equivariant maps. To give a
    $\Z$-set is to give a set $X$ equipped with a bijective operation
    $\sigma\colon X \to X$; the corresponding action $\Z \times X \to X$ then takes $(k,x)$
    to $\sigma^k(x)$. Examples are $\Z$ itself with the successor operation or the $n$-cycle
    $C_n = \set{0, 1, \ldots, n-1}$ with the successor modulo $n$.

    It is easy to see that $\Z$ and all $C_n$ are finitely presentable $\Z$-sets. Moreover,
    every finitely presentable $\Z$-set is a finite coproduct of copies of those $\Z$-sets.
    \begin{enumerate}
    \item $\ZSet$ is \DCC. Every finitely presentable $\Z$-set has finitely many
      subobjects. This follows from the fact that $\Z$ and each $C_n$ have no nonempty
      proper subobjects. Moreover, the proper quotients of $\Z$ are precisely all $C_n$,
      where quotients of $C_n$ are precisely all $C_m$ where $m$ divides $n$. It follows
      that $X$ has no infinite chain of quotients. 
    \item $\ZSet$ is not graduated. Indeed, $\Z$ has a chain of quotients of length $k$, for
      every $k \in \N$:
      \[
        \Z \epito C_{2^k} \epito C_{2^{k-1}} \epito \cdots \epito C_1.
      \]
      If $\Z$ had a grade $n$, then we would have $n > k$, for all $k \in \N$, a
      contradiction.
    \end{enumerate}
\end{enumerate}
\end{example}
}

\begin{notheorembrackets}
  \begin{example}[{\cite{AdamekSousa24}}] Here are examples of
    graduated categories.  In the first four items the grade is the
    cardinality of the underlying set.
    \begin{enumerate}
    \item $\Set$, and $\Set_p$ (pointed sets).   
    \item Boolean algebras and homomorphisms.
    \item Modules over a finite semiring and linear maps.
    \item $M$-$\Set$, sets with an action of a finite monoid $M$, and equivariant maps.
    \item $\Gra$, the category of graphs and homomorphisms.  The grade of a graph on $n$
      vertices with $k$ edges is $n+k$.
    \item $\KVec$, the category of vector spaces over a fixed field $K$ and linear
      maps. The grade of~$A$ is its dimension.
      
    \item $\Pos$, the category of posets and monotone maps.  Let $\Nat\times \Nat$ be
      the poset of pairs of natural numbers ordered lexicographically,
      and let $\mathbb{P}$ be the subposet of pairs $(n,k)$ with
      $k\leq n^2$.  There is an isomorphism
      $\phi\colon \mathbb{P}\to \Nat$.  The grade of a poset on $n$
      elements which contains $k$ comparable pairs is $\phi(n,k)$.
    \end{enumerate}
  \end{example}
\end{notheorembrackets}

\ifbolt
\begin{remark}
  While graduatednes is a stronger notion, $\DC$ categories are
  thunderstruck while graduated categories are too old to rock'n'roll
  but too young to die (cf.~Ian Anderson et al.~\cite{JethroTull76}).
\end{remark}
\fi

\subsection{Nominal Sets and Presheaves on Finite Sets}\label{S:Nom-SetF}

In this section we present two important examples of 
graduated categories which are not included in the previous work~\cite{AdamekSousa24}. The
first one is $\Nom$, the category of nominal sets and equivariant maps, and the second one is
the presheaf category $\SetF$, where $\F$ is the category of all finite sets and maps between
them. Both categories are used in well-known approaches to the category-theoretic study of
syntax with variable binding operations~\cite{GabbayPitts02,fpt}, such as the $\lambda$- or
$\pi$-calculus, nominal sets are also used as a framework for automata for data
languages~\cite{BojanczykEA14,skmw17}.

\mysubsec{Nominal sets.} We first present a proof that $\Nom$ is graduated based on ideas by
Urbat and Schröder~\cite{UrbatS20}. We assume that readers are familiar with basic notions
(like orbit and support) from the theory of nominal sets, see Pitts~\cite{Pitts13}.
\begin{proposition}
  The category $\Nom$ is graduated, therefore it is a $\DCC$-category.
\end{proposition}

\begin{proof}
  \begin{enumerate}
  \item The finitely presentable objects of $\Nom$ are precisely the orbit-finite nominal
    sets~\cite[Prop.~2.3.7]{petrisanphd}. Since subobjects of a nominal set $X$ are given by
    a number of orbits of $X$ due to equivariance, orbit-finite nominal sets are closed
    under subobjects, and the descending chain condition on subobjects of an orbit-finite
    nominal set clearly holds.

  \item Recall that in $\Nom$ all quotients are strong, and they are represented by the
    surjective equivariant maps. Therefore, orbit-finite nominal sets are closed under
    strong quotients. For the descending chain condition for strong quotients, we first
    consider single-orbit nominal sets and recall that the supports of elements of an orbit
    all have the same cardinality. We also recall the standard
    fact~\cite[Exercise~5.1]{Pitts13} that every single-orbit nominal set $X$ whose elements
    have supports of cardinality $n$ (this is the \emph{degree} of $X$) is a quotient of the
    nominal set $\names^{\# n} = \set{(a_1, \ldots, a_n) : |\set{a_1, \ldots, a_n}| = n}$,
    where $\names$ denotes the set of names (or atoms). Now observe that a quotient of
    $\names^{\# n}$ having degree $n$ is determined by a subgroup $G$ of the symmetric
    group~$S_n$ (of all permutations on $n$). More specifically, the quotient determined
    by~$G$ identifies $(a_1, \ldots, a_n)$ and $(a_{\pi(1)}, \ldots, a_{\pi(n)})$ for every
    $(a_1, \ldots, a_n) \in \names^{\#n}$ and every $\pi \in G$. Conversely, given a
    quotient $e\colon \names^{\#n} \epito X$ we obtain $G$ as consisting of all those $\pi$
    for which~$e$ identifies the above two $n$-tuples for every $a_1, \ldots, a_n$ in
    $\names$. We conclude that every strictly descending chain of quotients of
    $\names^{\#n}$ all having degree $n$ corresponds to a strictly descending chain of
    subgroups of $S_n$; the same holds of course for every single-orbit nominal set of
    degree $n$. For $n\geq 2$, such a chain of subgroups of~$S_n$ has length at most
    $2n-3$~\cite{Babai86} (and for $n=1$, $S_n$ is trivial, of course, so chains of
    subgroups have length 0).

  \item \label{threenumbers} Given a general orbit-finite
    set $X$, we now conclude that for every proper strong
    quotient of $X$, one of three numbers strictly
    decreases: the number of orbits, the degree of some orbit of $X$,
    or the maximum length of the above chain of subgroups of $S_n$ for
    some orbit.  We conclude that $\Nom$ is a $\DCC$-category.

  \item To see that $\Nom$ is even graduated, observe that we can
    assign to each orbit-finite nominal set~$X$ the sum of the
    three numbers mentioned in point \ref{threenumbers} above.
    It is then clear that for every proper nominal subset or quotient
    of $X$ the grade is strictly smaller.  \qedhere
  \end{enumerate}
\end{proof}

\mysubsec{Presheaves on finite sets.} We now turn to $\SetF$, the category of presheaves on
finite sets; Fiore et al.~\cite{fpt} have named this the category of sets in context. We 
just speak of presheaves.

\begin{remark}
  Note that $\SetF$ is equivalent to the category of all finitary set functors. In one
  direction, this equivalence is given by restricting the domain of a finitary set functor
  to~$\F$, and in the other direction one takes the left Kan extension of a given presheaf along
  the inclusion functor~$\F \subto \Set$.
\end{remark}

Like every presheaf category, $\SetF$ is locally finitely presentable.  We proceed to describe
the finitely presentable presheaves.  In what follows we identify natural numbers with finite
ordinals $n = \set{0, \ldots, n-1}$, as usual. The following definition was given for finitary
set functors by Ad\'amek and Trnkov\'a~\cite{ATbook}; we state it for presheaves.

\removeThmBraces
\begin{defn}\label{D:super}
  A presheaf $P\colon \F \to \Set$ is \emph{super-finitary} if there is a \emph{generating
    (natural) number} $n$, that is, (1)~$Pn$ is a finite set, and~(2) for every finite set $X$
  and $x \in PX$, there exists a map $f\colon n\to X$ such that $x \in Pf[Pn]$. (This implies
  that $PX$ is finite.)
\end{defn}

A strong quotient of a presheaf $P$ is represented by a natural transformation
$q\colon P \epito Q$ whose components $q_X$ are surjective maps (i.e.~epic in $\Set$); we
simply speak of quotients subsequently. The next result has appeared in previous
work~\cite[Lemma~3.30 and Cor.~3.34]{amsw19_1}.%
%
\begin{proposition}\label{P:super}
  For every presheaf $P$, the following are equivalent:
  \begin{enumerate}
  \item $P$ is finitely presentable,
  \item $P$ is super-finitary,
  \item\label{P:super:3} $P$ is a quotient of a presheaf $X \mapsto A\times X^n$ for a finite
    set $A$ and a natural number~$n$.
  \end{enumerate}
\end{proposition}
\noindent
It follows from the proof that $n$ in \cref{P:super:3} can be chosen to be any
generating number of~$P$.
\begin{remark}\label{R:super}
  \begin{enumerate}
  \item Note that for finite $A$, every presheaf $A\times (-)^n$ is clearly super-finitary with
    $n$ as its generating number.

  \item\label{R:super:1a} It is easy to infer from~\cref{P:super} that a presheaf has
    generating number $0$ iff it is a constant functor with a finite set~$A$ as value.

  \item\label{R:super:1b} From \cref{D:super} we see that if $n$ is a generating number of $P$, then so is every $n' \geq n$. 

  \item\label{R:super:2} The proof of \cref{P:super} given in op.~cit.~shows that $n$ is a generating
    number for a presheaf $P$ precisely when $P$ is a quotient of $A \times (-)^n$ for
    some finite set $A$.

  \item\label{R:super:3} Super-finitary presheaves are closed under subpresheaves and
    quotients~\cite[Prop.~3.31 and Ex.~3.29(4)]{amsw19_1}. Moreover, for a quotient presheaf
    $P \epito Q$ or a subpresheaf $Q \subto P$, every generating number for $P$ is also
    generating for $Q$. For quotients, see again~\cite[Ex.~3.29(4)]{amsw19_1}. For
    subpresheaves of $(-)^n$, see~\cite[Ex.~3.29(3)]{amsw19_1}, and for subpresheaves of
    $A\times (-)^n$, the argument is the same as in op.~cit.
    
    For a general subpresheaf $Q \subto P$ and a generating natural number $n$ of $P$, write
    $P$ as a quotient $q\colon A\times (-)^n \epito P$ for some finite set $A$, using
    \cref{R:super:2}. Now consider the pullback
    \[
      \begin{tikzcd}
        Q'
        \pullbackangle{-45}
        \ar[hook]{d}
        \ar[->>]{r}
        &
        Q
        \ar[hook]{d}
        \\
        A \times (-)^n
        \ar[->>]{r}{q}
        &
        P
      \end{tikzcd}
    \]
    Then $n$ is a generating natural number for $Q'$ and whence for $Q$ by what we have
    already stated above. 
  \end{enumerate}
\end{remark}
\begin{lemma}\label{L:super}
  Let $P$ be a super-finitary presheaf with a generating number $n$, and put $k = 2n$. For every
  \mbox{$\ell \geq k$} and every pair $x_1, x_2 \in F\ell$, there exists an injective map
  $f\colon k \monoto \ell$ such that $x_1, x_2 \in Ff[Fk]$.
\end{lemma}
\begin{proof}
  \begin{enumerate}
  \item\label{L:super:1} Suppose first that the presheaf $P$ is of the form $A \times
    (-)^n$. Suppose that $\ell \geq k$, and let $x_1,x_2 \in P\ell$. Then we have
    $x_i = (a_i, t_i)$ for some $a_i \in A$ and $t_i \in \ell^n$, for $i = 1,2$. Let
    $M \subseteq \ell$ be the set of all components of the $n$-tuples $t_1, t_2$. Then clearly
    $|M| \leq k$. Thus, we can find an injective map $f\colon k \to \ell$ and $n$-tuples
    $t_i' \in k^n$ such that $f^n(t_i') = t_i$ for $i = 1,2$. Hence, for $y_i = (a_i, t_i')$ we
    have $Pf(y_i) = x_i$ for $i = 1,2$, and so we are done.

  \item Given any super-finitary presheaf $P$, we write it as a quotient
    $q\colon A \times (-)^n \epito P$ for some finite set $A$.
    Given $x_1, x_2 \in P\ell$, choose $y_1, y_2 \in A\times \ell^n$ such that
    $q_\ell(y_i) = x_i$ for $i = 1,2$. By \cref{L:super:1} we have an injective map $f\colon k
    \monoto \ell$ and elements $y_1',y_2'\in A \times k^n$ such that $(A \times f^n)(y_i') = y_i$
    for $i = 1,2$. Now due to the naturality square
    \[
      \begin{tikzcd}[column sep = 40]
        A\times k^n
        \ar{d}[swap]{q_k}
        \ar{r}{A\times f^n}
        &
        A\times \ell^n
        \ar{d}{q_\ell}
        \\
        Pk
        \ar{r}{Pf}
        &
        P\ell
      \end{tikzcd}
    \]
    we have that $x_i' = q_k(y_i')$ satisfies $Pf(x_i') = Pf(q_k(y_i')) = q_\ell((A\times
    f^n)(y_i')) = q_\ell(y_i) = x_i$ for $i = 1,2$. Thus, $x_1,x_2 \in Pf[Pk]$ as desired.
    \qedhere
  \end{enumerate}
\end{proof}
\begin{lemma}\label{L:super-quo}
  Every super-finitary presheaf has only finitely many quotients, which are all super-finitary (\cref{R:super}\ref{R:super:3}). 
\end{lemma}
\begin{proof}
  Let $P$ be a super-finitary presheaf, let $k$ be the natural number as in \cref{L:super}. Using
  \cref{R:super}\ref{R:super:1b}, we may assume that $k > 0$. We shall prove that every
  quotient $q\colon P \epito Q$ is determined by the (finite) collection of equivalence
  relations $\ker q_n = \set{(x,y) : q_n(x) = q_n(y)}$ on the sets~$Pn$ for $n \leq k$. Indeed,
  the kernel of $q_\ell$ for $\ell > k$ is derived from $\ker q_k$ as follows. We show below
  that elements $x_1,x_2 \in P\ell$ are merged by $q_\ell$ iff there is a monic
  $f\colon k \monoto \ell$ and elements $y_1,y_2 \in Pk$ such that $x_i = Pf(y_i)$, for $i =1,2$,
  and which are merged by $q_k$:
  \begin{enumerate}
  \item If $f$ and $y_1,y_2$ exist, then $q_\ell(x_1) = q_\ell(x_2)$, since
    \[
      q_\ell(x_i) = q_\ell(Pf(y_i)) = Qf(q_k(y_i))
      \qquad
      \text{for i = 1,2},
    \]
    due to the naturality of $q$.
  \item Conversely, if $q_\ell(x_1) = q_\ell(x_2)$, then
    \[
      Qf(q_k(y_1)) = q_\ell(Pf(y_1)) = q_\ell(x_1) = q_\ell(x_2) = q_\ell(Pf(y_2)) =
      Qf(q_k(y_2)),
    \]
    again due to the naturality of $q$. Since $k >0$, we see that $f$ is split monic, whence so
    is $Qf$. We conclude that $q_k(y_1) = q_k(y_2)$. \qedhere
  \end{enumerate}
\end{proof}
\takeout{
\begin{lemma}\label{L:super-sub}
  Every presheaf has finitely many subpresheaves, and they are all super-finitary.
\end{lemma}
\begin{proof}
  \begin{enumerate}
  \item\label{L:super-sub:1} First, we consider the presheaf $(-)^n$. Observe that every subpresheaf
    $P \subto (-)^n$ defines a set $E$ of equivalence relations on $n$:
    \[
      E = \set{\ker h : \text{$h\colon n \to n$ lies in $Pn$}}.
    \]
    We verify that for every finite set $S$ and every map $f\colon n \to S$ we have
    \begin{equation}\label{eq:ker}
      f\in PS \qquad\text{if and only if}\qquad \ker f\in S.
    \end{equation}
    For `if', suppose that $\ker f \in E$ and choose any $h \in Pn$ such that $\ker h = \ker
    f$. Define $u\colon n \to S$ by $u(h(x)) = f(x)$, and for $y \not\in h[n]$,
    $u(y)$ is arbitrary. Since $\ker f = \ker h$, the map $u$ is well-defined. Using that $P$ is a
    subfunctor of $(-)^n$, we conclude that $f = u  \cdot h = Pu(h) \in PS$, as desired.

    For `only if', let $f \in PS$ and choose any $h\colon n \to n$ such that $\ker h = \ker
    f$. Define $v\colon S \to n$ by $u(f(x)) = h(x)$, and for $y \not\in f[n] \subseteq S$,
    $u(y)$ is arbitrary. Since $\ker f = \ker h$, the map $u$ is well-defined. Using that $P$
    is a functor, we conclude that $h = u \cdot f = Pu(f) \in Pn$. Thus, $\ker h \in E$.

    It follows from~\eqref{eq:ker} that $E$ is uniquely determined by $P$; different
    subfunctors of $(-)^n$ yield different sets $E$. Hence, since we only have finitely many
    different sets $E$ of equivalence relations on $n$, we see that there are only finitely
    many different subfunctors of $(-)^n$. 
      
  \item\label{L:super-sub:2} A subpresheaf of $P = A\times (-)^n$, for a finite set $A$ and a
    natural number $n$, is given by a subset of $A$ and a subpresheaf of $(-)^n$. Using
    \cref{L:super-sub:1}, we see that $P$ has only finitely many subpresheaves.

  \item Now let $P$ be any super-finitary presheaf. So $P$ is a quotient $q\colon A \times (-)^n
    \epito P$ for a finite set $A$ and a natural number $n$ according to \cref{P:super}. Given
    a subpresheaf $Q\subto P$ we obtain a subpresheaf $Q' \subto A \times (-)^n$ by forming the
    pullback below:
    \[
      \begin{tikzcd}
        Q'
        \pullbackangle{-45}
        \ar[hook]{d}
        \ar[->>]{r}
        &
        Q
        \ar[hook]{d}
        \\
        A \times (-)^n
        \ar[->>]{r}{q}
        &
        P
      \end{tikzcd}
    \]
    Since pullbacks in $\SetF$ are formed object wise and epimorphisms are stable under
    pullback in $\Set$, we see that $Q$ is quotient of $Q'$.  By \cref{L:super-sub:2}, there
    are only finitely many possible $Q'$ each of which is super-finitary. Hence, by
    \cref{L:super-quo}, each such $Q'$ has only finitely many quotients $Q$, which are all
    super-finitary. So there are only finitely many subpresheaves of $P$, which are all
    super-finitary. 
    \qedhere
  \end{enumerate}
\end{proof}
}
\begin{theorem}
  The category $\SetF$ is graduated, and therefore it is a $\DCC$-category.
\end{theorem}
\begin{proof}
  We have already seen in \cref{R:super}\ref{R:super:3} that super-finitary presheaves are
  closed under subobjects and quotients. Given a super-finitary presheaf $P$, let $n$ be the smallest
  generating number of $P$. According to \cref{R:super}\ref{R:super:2}, $n$ is the
  least number such that we have a quotient $A \times (-)^n \epito P$ for
  some finite set $A$. 

  The grade of $P$ is defined as $|A|+\sum_{k\leq 2n} |Pk|$. The reasons why this works are
  that (1)~$P$ is a quotient of $A \times (-)^n$ and~(2)~that quotients of $P$ are determined by
  quotients of $Pk$, for $k \leq 2n$, according to the proof of \cref{L:super-quo}. If
  $n = 0$, then $P$ is constant with value~$A$ (\cref{R:super}\ref{R:super:1a}), and so the
  grade of any proper subfunctor or quotient of $P$ is clearly smaller than that of~$P$.

  Observe further that if $Pn = \emptyset$, then $P$ is constant with value
  the empty set; we have
  \[
    PX = \bigcup_{f\colon n \to X}Pf[Pn] = \emptyset.
  \]
  Again, we are done in this case. So from now on we assume that $n > 0$ and that
  $Pn \neq \emptyset$.

  Let $Q\subto P$ be a proper subpresheaf. If the least generating number $m$ of $Q$ is smaller
  than $n$, then the grade of $Q$ is smaller than that of $P$, because $|Qk| \leq |Pk|$ for all
  $k \leq 2n$ and $|P(2n)| >0$, since $|Pn| > 0$. If $m = n$, then $Qn \subseteq Pn$ must be a
  proper subset; for otherwise we have the following naturality square of the inclusion $Q \subto P$
  \[
    \begin{tikzcd}
      Qn
      \ar[equals]{r}
      \ar{d}[swap]{Qf}
      &
      Pn
      \ar{d}{Pf}
      \\
      QX
      \ar[hook]{r}
      &
      PX\rlap{,}
    \end{tikzcd}
  \]
  which we use in the penultimate step of the following computation to obtain a contradiction:
  \[
    QX
    =
    \bigcup_{f\colon n \to X} Qf[Qn]
    =
    \bigcup_{f\colon n \to X} Qf[Pn]
    =
    \bigcup_{f\colon n \to X} Pf[Pn]
    =
    PX
    \quad
    \text{for every finite set $X$}.
  \]
  Thus, the grade of~$Q$ is strictly smaller than that of~$P$.

  Now let $q\colon P\epito Q$ be a proper quotient. Recall from the proof of \cref{L:super-quo}
  that the kernel $\ker q_\ell$ for every $\ell > 2n$ is determined by the morphisms $q_k$ for
  $k \leq 2n$; note that $2n$ is the number from (the proof of) \cref{L:super}. Equivalently,
  the quotient $q_\ell\colon P\ell \epito Q\ell$, for every $\ell >k$, is determined by the
  quotients $q_k\colon Pk \epito Qk$ for $k \leq 2n$. Hence, since $Q$ is a proper quotient of
  $P$, we have $|Qk| < |Pk|$ for some $k \leq 2n$. Thus, the grade of $Q$ is smaller than that
  of $P$.
\end{proof}
\begin{example}
  \begin{enumerate}
  \item For the category $\Bij$ of finite sets and bijections, the presheaves also form a
    graduated category. Indeed, consider the full subcategory $\Bij_0$ given by natural
    numbers. Then $\Bij_0$ is equivalent to $\Bij$, so $\Set^{\Bij_0}$ is equivalent to
    $\SetB$. Moreover, $\Bij_0 = \coprod_{n \in \N} S_n$ is the coproduct of all symmetric
    groups $S_n$. Since each $S_n$ is finite, $\Set^{\Bij_0}$ is
    graduated~\cite[Ex.~2.4(2)]{AdamekSousa24}. Thus, $\SetB$ is graduated, as claimed.

  \item For the category $\Surj$ of finite sets and surjections, the presheaves also form a
    graduated category. Indeed, consider the full subcategory $\Surj_0$ given by natural
    numbers. Then $\Surj_0$ is equivalent to $\Surj$, so $\Set^{\Surj_0}$ is equivalent to
    $\Set^{\Surj}$. We show that a presheaf is finitely generated iff it has finitely many
    elements. To see this, recall~\cite[Ex.~3.5(3) and Prop.~3.11]{ar} that a presheaf
    $P\colon \Surj_0 \to \Set$ is finitely generated iff it is generated by finitely many
    elements $x_i \in Pn_i$ ($i = 1, \ldots, k$): every element $x \in Pm$ has the form
    $x = Pf(x_i)$ for some $i$ and some $f\colon n_i \epito m$. Since $f$ is epic, we have
    $m \leq n_i$. Thus, $P$ has only finitely many elements. This shows `if' above, and
    `only if' is clear.

    Consequently, finitely generated presheaves are clearly closed under
    quotients, finite products, subobjects, and hence finite limits. Hence, the classes of
    all finitely presentable and finitely generated presheaves
    coincide~\cite[Prop.~3.32]{amsw19_1}. So using that finitely presentable presheaves are
    precisely those having finitely many elements, we obtain that $\Set^\Surj$ is graduated.
    
  \item In contrast, for the category $\Inj$ of finite sets and injections, the category of
    presheaves is not a $\DCC$-category, even though it is known that finitely presentable
    objects are closed under subobjects and strong
    quotients~\cite[Prop.~5.2]{FrankMU23}. Indeed, consider the full subcategory $\Inj_0$
    given by natural numbers, so $\Set^{\Inj_0}$ is equivalent to $\Set^\Inj$. Now take the
    inclusion $\Inj_0 \subto \Set$. This is the naturally equivalent to the representable
    presheaf $\Inj_0(1,-)$ and therefore finitely presentable.%
    \footnote{A presheaf is finitely presentable iff it is a finite colimit of representable
      ones~\cite[Thm.~5.37 and Cor.~5.41]{Kelly82}.}
    However, the following
    subpresheaves $P_n\colon \Inj_0 \to \Set$ ($n \in \N$) form an infinite descending
    chain:
    \[
      P_n(k) = \begin{cases}
        \emptyset & \text{if $k \leq n$}, \\
        k & \text{else}.
      \end{cases}
    \]
    Thus, $\Set^\Inj$ is not a $\DCC$-category.

  \item The category $\Ab$ of abelian groups is locally finitely presentable but not
    a $\DCC$-category. The group $\Z$ of integers is finitely presentable, but it has the following
    descending sequence of proper subgroups $\Z\supset 2 \Z\supset 4 \Z \supset \cdots$.
  \end{enumerate}
\end{example}
\subsection{Limits of $\omega^{\opp}$-chains in  $\DCC$-Categories}
The main point about locally finitely presentable $\DCC$-categories in this paper is that 
they have finitary $\omega^{\opp}$-limits as we define them just below.  This will be used 
subsequently to prove our first main result \Cref{T:main}.

\begin{defn}
\label{D:fO}
A category has \emph{finitary $\omega^{\opp}$-limits} provided that for every limit
$\ell_n \colon L\to A_n$ of an $\omega^\opp$-chain and every finitely presentable
  subobject $m\colon M\monoto L$, there is some $k\in \omega$ such that
  $\ell_k\o m: M \to A_k$ is monic.
\end{defn}
\begin{proposition}\label{P:ACDC}
  Every $\DCC$-category has finitary $\omega^{\opp}$-limits.
\end{proposition}
\begin{proof}
  Let $\ell_n\colon L \to A_n$ be a limit cone of an $\omega^{\opp}$-chain $D =
  (A_n)$ with connecting morphisms $a_{n+1}\colon A_{n+1} \to
  A_n$.  Given a finitely presentable object $M$ and a monomorphism $m\colon M \monoto
  L$, factorize $\ell_n\o m$ as a strong epimorphism $e_n\colon M\epito
  B_n$ followed by a monomorphism $u_n \colon B_n\monoto
  A_n$ (\Cref{rem-LP}\ref{rem-LP-2}).  We obtain a subchain $(B_n)$ of
  $(A_n)$ with connecting maps~$b_n$ given by the diagonal fill-ins, as shown below:
  \begin{equation}\label{FP-om}
    \begin{tikzcd}
      M \ar[->>]{dd}[swap]{e_{n}}
      \ar[->>]{r}{e_{n+1}}  
      &  B_{n+1}   \ar[>->]{d}{u_{n+1}} \ar[->>,ddl, "b_n" description]\\ 
      & A_{n+1}   \ar{d}{a_{n}}  \\
      B_n \ar[>->]{r}[swap]{u_n}  & A_{n} 
    \end{tikzcd}
  \end{equation}
    
  Notice that $b_n$ is a strong epimorphism, since so is $e_n$.  We thus have a descending
  chain~$(B_n)$ of strong quotients of the finitely presentable object $M$:
  $e_0 \geq e_1 \geq e_2 \geq \cdots$.  By the DCC condition, there is some $k$ such that for
  $n\geq k$, $b_n$ is an isomorphism.  For $n\geq k$, let $b_{n,k}\colon B_n\to B_k$ be the
  evident composition, and note that $b_{n+1,k} = b_{n,k}\o b_n$.  Thus, for all $n\geq k$, the
  triangle below commutes, where the lower part commutes by~\eqref{FP-om}:
  \begin{equation}\label{allfigs}
    \begin{tikzcd}[column sep = 10]
      &  &  B_k \ar{dl}[swap]{b_{n,k}^{-1}} \ar{dr}{b_{n+1,k}^{-1}} & & \\
      & B_n  \ar{dl}[swap]{u_n} & & \ar{ll}{b_n}  B_{n+1}\ar{dr}{u_{n+1}} & \\
      A_n & && & \ar{llll}{a_n} A_{n+1} 
    \end{tikzcd}
  \end{equation}

  Let $D'$ be the $\omega^{\opp}$ chain $(A_n)_{n\geq k}$.  This is a shortening of our
  original $\omega^{\opp}$-chain $D$, and so its limit is $\ell_n\colon L\to A_n$ ($n \geq k)$.
  The commutativity of all figures \eqref{allfigs} shows that we have a cone
  $(u_n\o b_{n,k}^{-1})_{n\geq k}$. Thus, there exists $b\colon B_k\to L$ such that
  \[
    \ell_n \o b = u_n\o b_{n,k}^{-1} \qquad (n\geq k).
  \]
  Consider the following diagram for $n\geq k$:
  \begin{equation}\label{sqtri}
    \begin{tikzcd}
      & B_k \ar{d}[swap]{b} \ar{r}{b^{-1}_{n,k}} & B_n  \ar{d}{u_{n}}
      \\
      M \ar{r}{m}  \ar{ru}{e_k} 
      & 
      L \ar{r}{\ell_n}  
      & A_n
    \end{tikzcd}
  \end{equation}
  The square commutes, and we now prove that so does the outside.  We show that for all
  $n \geq k$ and all $0\leq i \leq n -k$,
  \begin{equation}\label{ind}
    u_{n} \o b_{n, n-i}^{-1}\o e_{n-i} = \ell_{n}\o m.
  \end{equation}
  We argue by induction on $i$.  For $i = 0$, this holds using $b_{n,n} = \id$ and the
  factorization $u_n \o e_n = \ell_n \o m$.  Assume \eqref{ind} for $i$.  Fix $n \geq k$ such
  that $n -k \geq i+1$.  Then
  \[
    \begin{array}{cll}
      & u_{n} \o b_{n, n-(i+1)}^{-1}\o e_{n-(i+1)}  \\
      = &  
      u_{n} \o b_{n, n-i}^{-1}\o b_{n-i-1}^{-1}\o e_{n-i -1}
      & \mbox{since $b_{n,n-i-1} =   b_{n-i-1} \o b_{n,n-i } $}
      \\
      = & u_{n} \o b_{n, n-i}^{-1}\o e_{n-i} 
      & \mbox{since $e_{n-i-1} = b_{n-i-1} \o e_{n-i}$}
      \\
      = & \ell_{n}\o m  & \mbox{by induction hypothesis}
    \end{array}
  \]
  The induction completed, we take $i =n-k$ in~\eqref{ind} to see the commutativity of the
  outside of \eqref{sqtri} for all $n$.  Since the limit cone $(\ell_n)_{n\geq k}$ is
  collectively monic, the triangle commutes: $m = b\o e_k$.  As $m$ is monic, so is $e_k$.
  Thus, $\ell_k \o m = u_k \o e_k$ is also monic, as desired.
\end{proof}

\section{Terminal Coalgebras in $\omega+\omega$ Steps}
\label{S:omega+omega}

We are ready to state and prove the first main theorem of this paper, presenting conditions
under which the terminal-coalgebra chain for a finitary functor $F$ converges in
$\omega + \omega$ steps. Without extra conditions it is known that a terminal coalgebra
exists, and assuming that~$F$ preserves monomorphisms, its terminal-coalgebra chain
converges in $n$ steps for some ordinal~$n$. However, this ordinal can be arbitrarily large
(see \cref{E:n-steps}).  The following can, for example, be found in our book~\cite[Thm.~11.2.18 and~11.2.22]{AMM}. 
\begin{theorem}
  Let $\A$ be a locally finitely presentable category. Every finitary endofunctor~$F$ has a
  terminal coalgebra. If $F$ preserves monomorphisms, the terminal-coalgebra chain converges
  in $n$ steps for some ordinal $n$.
\end{theorem}
\begin{theorem}\label{T:main}
  For every $\DCC$-category and every finitary endofunctor $F\colon\A \to\A$ preserving
  nonempty binary intersections, then the terminal-coalgebra chain converges in
  $\omega+\omega$~steps.
\end{theorem}
\begin{proof}
  We will apply \Cref{P:modest}. Due to \cref{R:nonempty}\ref{R:nonempty:3} we can assume
  without loss of generality that $V_i$ is not strictly initial for any $i \leq \omega + \omega$.
  \begin{enumerate}
  \item\label{T:main:1} We first show that the canonical morphism
    $m\colon V_{\omega+1}\to V_{\omega}$ is monic.  Consider a parallel pair
    $q,q'\colon Q \parto FV_{\omega}$ such that $m\o q = m\o q'$.  We prove that $q = q'$. By
    \Cref{rem-LP}\ref{rem-LP-5}, we may assume that $Q$ is a finitely presentable
    object. Using that $V_\omega$ can be assumed not to be strictly initial and
    \cref{R:nonempty}\ref{R:nonempty:2a}, we may
    express~$V_{\omega}$ as a directed colimit of nonempty finitely presentable subobjects, say
    $m_t\colon M_t\monoto V_{\omega}$ ($t\in T$).  Since~$F$ is finitary,
    $Fm_t\colon FM_t\to FV_{\omega}$ is also a directed colimit.  Hence, $q$ and $q'$ factorize
    through~$Fm_t$ for some $t$.  We denote the factorizing morphisms by $r$ and $r'$,
    respectively. It is sufficient to show that they are equal. To this end consider the
    following diagram:
    \begin{equation}\label{eq:twotri}
      \begin{tikzcd}[column sep = 35]
        & FM_t\ar{d}{Fm_t}
        &&
        V_{\omega} \ar{d}{v_{\omega,i+1}}
        \\
        Q
        \ar[yshift=2pt]{r}{q}
        \ar[yshift=-2pt]{r}[swap]{q'}
        \ar[yshift=4pt]{ur}[near end, inner sep = 1.5]{r}
        \ar[yshift=0pt,xshift=2pt]{ur}[swap,near end, inner sep=0]{r'}    
        &
        FV_{\omega} 
        \ar{rr}[swap]{v_{\omega+1,i+1} = Fv_{\omega,i}}
        \ar{rru}{m}
        &&
        FV_{i} = V_{i+1} 
      \end{tikzcd}
    \end{equation}
    The limit $v_{\omega,i}\colon V_{\omega}\to V_i$ is finitary (\Cref{P:ACDC}).  Thus, there
    is some $i$ so that \mbox{$v_{\omega,i}\o m_t\colon M_t \to V_{i}$} is monic, and this
    monomorphism is nonempty.  Since $F$ preserves nonempty binary intersections, it preserves
    nonempty monomorphisms (\cref{R:pres}).  Hence, the following morphism is monic:
    \[
      \ell = \big(
      FM_t
      \xra{Fm_t}
      FV_\omega
      \xra{Fv_{\omega,i}}
      FV_{i}
      \big).
    \]
    It is enough to show that $\ell$ merges $r$ and
    $r'$. The triangle on the right in~\eqref{eq:twotri} commutes. Thus we obtain
    \[
      \ell = F v_{\omega,i}\o Fm_t = v_{\omega,i+1}\o m \o Fm_t. 
    \]
    Using that $m$ merges $q$ and $q'$, we see that $\ell$ merges $r$ and
    $r'$:
    \begin{align*}
      \ell\o r & = v_{\omega,i+1}\o m \o Fm_t \o r \\
      &= v_{\omega,i+1}\o m \o q \\
      &= v_{\omega,i+1}\o m \o q'\\
      &= v_{\omega,i+1}\o m \o Fm_t \o r'\\
      &= \ell \o r'.
    \end{align*}
    Since $\ell$ is monic, we have $r = r'$ whence $q = q'$, as desired.
    
  \item\label{T:main:2} Next, we prove that $F$ preserves nonempty intersections of
    $\omega^\opp$-chains of subobjects.  Consider such a chain $a_i\colon A_{i+1}\monoto
    A_i$, and let its limit cone be $\ell_i\colon L\monoto A_i$, where $L$ is
    not strictly initial. It follows that neither is any of the $A_i$. Take a cone
    \[
      q_i \colon Q\to FA_i \quad (i<\omega).
    \]
    Our task is to find a morphism $q\colon Q \to FL$ such that
    $q_i = F\ell_i \o q$ for all $i$.
    (This is unique: all maps $\ell_i$ are nonempty monic, whence all $F\ell_i$ are monic.)

    Using \Cref{rem-LP}\ref{rem-LP-4}, we can assume, without loss
    of generality, that $Q$ is finitely presentable: for a
    general object $Q$, express it as a colimit of finitely
    presentable subobjects~$Q_t$, and use the result which we prove for each $Q_t$.

    Choose a nonempty, finitely presentable subobject $c\colon C \monoto L$
    (\cref{R:nonempty}\ref{R:nonempty:2}). Note that this gives nonempty, finitely
    presentable subobjects
    \[
      c_i = \big(
      \begin{tikzcd}[cramped,column sep=20]
        C \ar[>->]{r}{c}
        &
        L
        \ar[>->]{r}{\ell_i}
        &
        A_i
      \end{tikzcd}
      \big)
      \qquad\text{for every $i < \omega$},
    \]
    which, moreover, form a cone: $c_i = a_i \cdot c_{i+1}$ for every $i < \omega$. 
    
We recursively define a subchain $(B_i)$ of $(A_i)$ by taking intersections as
      shown below:
    \[
      \begin{tikzcd}
        B_0  \ar[>->]{d}[swap]{u_0}
        &
        B_1
        \ar[>->]{l}[swap]{b_0}
        \ar[>->]{d}{u_1}
        \pullbackangle{225}
        &
        B_2
        \ar[>->]{l}[swap]{b_1}
        \ar[>->]{d}{u_2}
        \pullbackangle{225}
        &
        \cdots 
        \ar[>->]{l}[swap]{b_2}
        \\
        A_0
        &
        A_1
        \ar[>->]{l}[swap]{a_0}
        &
        A_2
        \ar[>->]{l}[swap]{a_1}
        &
        \cdots 
        \ar[>->]{l}[swap]{a_2} 
      \end{tikzcd}
    \]
Together with this we define a cone $p_i\colon Q\to FB_i$ such that
    $F u_i \o p_i = q_i$ and a cone $m_i\colon C \monoto B_i$ such that $c_i = u_i \cdot
    m_i$; this shows that all the intersections are nonempty. 

    To define $B_0$ and $u_0$, express $A_0$ as a directed colimit of all its finitely
    presentable subobjects $u\colon B\monoto A_0$ that contain $c_0$
    (\Cref{rem-LP}\ref{rem-LP-4}).  Then use that $F$ preserves this colimit: for the morphism
    $q_0\colon Q\to FA_0$ we may find a subobject $u_0\colon B_0\monoto A_0$ containing~$c_0$
    such that $q_0$ factorizes through $F u_0$ via some $p_0\colon Q\to FB_0$, say:
    \[
      \begin{tikzcd}
        &
        F B_0
        \ar[>->]{d}{Fu_0}
        \\
        Q
        \ar[dashed]{ru}{p_0}
        \ar{r}{q_0}
        &
        F A_0 
      \end{tikzcd}
    \]
    Since $u_0$ contains the subobject $c_0$, we have a monomorphism $m_0\colon C_0
    \monoto B_0$ such that $c_0 = u_0 \cdot m_0$.

    In the induction step we are given $B_i$, $u_i$, $p_i$, and $m_i$. Form the intersection of~$u_i$
    and $a_i$ to obtain $B_{i+1}$, $b_i$, and $u_{i+1}$ as shown in the left-hand square below:
    \[
      \begin{tikzcd} 
        B_i
        \ar[>->]{d}[swap]{u_i}
        &
        B_{i+1}
        \pullbackangle{225}
        \ar[>->]{l}[swap]{b_i}
        \ar[>->]{d}{u_{i+1}}
        &
        C
        \ar[>->,dashed]{l}[swap]{m_{i+1}}
        \ar[>->]{ld}{c_{i+1}}
        \ar[>->,shiftarr = {yshift=20}]{ll}[swap]{m_i}
        \\ 
        A_i
        &
        A_{i+1} 
        \ar[>->]{l}[swap]{a_i}  
      \end{tikzcd}
    \]
    The outside commutes by induction hypothesis: $u_i \cdot m_i = c_i = a_i \cdot
    c_{i+1}$. Hence, we obtain the monomorphism $m_{i+1}$ as indicated such that the upper part
    and right-hand triangle commute, as desired. Since $C$ is not strictly initial, neither is
    $B_{i+1}$, whence the intersection of $a_i$ and $u_i$ is nonempty.
    
    So by hypothesis, $F$ preserves the above pullback. Since the square below commutes
    \[
      \begin{tikzcd}
        F B_i  \ar[>->]{d}[swap]{Fu_i}
        &\ar{l}[swap]{p_i}  Q   \ar{d}{q_{i+1}} \ar[dl, "q_i" description]\\ 
        FA_i &\ar[>->]{l}[swap]{F a_i}  FA_{i+1} 
      \end{tikzcd}
    \]
    there is a unique morphism $p_{i+1}\colon Q\to FB_{i+1}$ such that 
    \[
      p_i = Fb_i \o  p_{i+1}
      \qquad\text{and}\qquad
      q_{i+1}= F u_{i+1}\o p_{i+1}.
    \]

    For all $i\leq j <\omega$, we form the composite morphism
    \[
      b_{j,i} = \big(
      \begin{tikzcd}[cramped, column sep = 30]
        B_j
        \ar[>->]{r}{b_{j-1}}
        &
        B_{j-1}
        \ar[>->]{r}{b_{j-2}}
        &
        \cdots
        \ar[>->]{r}{b_{i+1}}
        &
        B_{i+1}
        \ar[>->]{r}{b_i}
        &
        B_i
      \end{tikzcd}
      \big).
    \]
    We obtain a descending chain of subobjects
    $b_{j,0} \colon B_j \monoto B_0$ $(j < \omega)$ of the finitely
    presentable object $B_0$.  Since $\A$ is a $\DCC$-category, there is some
    $k^*<\omega$ such that $b_{k^*,0}$ represents the same subobject as
    $b_{j,0}$ for every $j\geq k^*$. Hence,  the morphism $b_{j,k^*}$ is
    an isomorphism.
    
    The shortened $\omega^{\opp}$-chain $(A_i)_{i\geq k^*}$ has the limit cone $(\ell_i)_{i \geq k^*}$.  The morphisms
    \[
      h_i =
      \big(
      B_{k^*}
      \xra{b^{-1}_{i,k^*}}
      B_i
      \xra{u_i}
      A_i
      \big)
      \qquad (i\geq k^*)
    \]
    form a cone: we see that $h_i = a_i \o h_{i+1}$ from the commutativity of the diagram below:
    \[
      \begin{tikzcd}[column sep = 20]
        &
        B_{k^*}
        \ar[>->]{ld}[swap,inner sep = 0]{b_{i,k^*}^{-1}}
        \ar[>->]{rd}[inner sep = 0]{b_{i+1,k^*}^{-1}}
        \arrow[rounded corners, to path={
          -- ([xshift=-64.5]\tikztostart.center)
          -- ([xshift=-20]\tikztotarget.center)\tikztonodes  
          -- (\tikztotarget)}]%
        {ldd}[swap]{h_i}
        \arrow[rounded corners, to path={
          -- ([xshift=74.5]\tikztostart.center)
          -- ([xshift=25]\tikztotarget.center)\tikztonodes  
          -- (\tikztotarget)}]%
        {rdd}{h_{i+1}}
        \\
        B_i
        \ar{d}[swap]{u_i}
        &&
        B_{i+1}
        \ar[>->]{ll}{b_i}
        \ar{d}{u_{i+1}}
        \\
        A_i
        &&
        A_{i+1}
        \ar{ll}{a_i}
      \end{tikzcd}
    \]
    So there is a unique morphism $h\colon B_{k^*} \to L$ such that
    $\ell_i \o h = u_i\o b^{-1}_{i,k^*}$ for $i\geq k^*$.
    
    The desired morphism is 
    \[
      q =
      \big(
      Q
      \xra{p_{k^*}}
      FB_{k^*}
      \xra{Fh}
      FL
      \big).
    \]
    In order to verify that $q_i = F\ell_i \o q$, it is sufficient to show this for
    $i\geq k^*$; it then follows also for all $i < k^*$, since the $q_i$ and $\ell_i$ form
    cocones:
    \[
      q_i = Fa_{k^*,i} \cdot q_{k^*} = Fa_{k^*,i} \cdot F\ell_{k^*} \o q = F\ell_i \o q
      \qquad
      \text{for $i < k^*$.}
    \]
    Now observe first that since~$(p_i)$ form a cone of
    $(FB_i)$, we have
    \[
      Fb_{i,k^*}\o p_i = p_{k^*}.
    \]
    By definition of $h$, we also have 
    $u_i = \ell_i \o h \o b_{i,k^*}$.
    Therefore for all $i\geq k^*$, we obtain
    \[
      q_i = Fu_i \o p_i = F\ell_i \o Fh \o Fb_{i,k^*}\o p_i = F\ell_i \o
      Fh \o p_{k^*} = F\ell_i \o q.
    \]
    This extends to all $i < k^*$, the argument is as above.
  \end{enumerate}
  Having checked all the conditions in~\Cref{P:modest}, we are done.
\end{proof}


In two cases it does:
\begin{corollary}\label{C:set-kvec}
  For every finitary endofunctor on $\Set$ or $\KVec$, the terminal-coalgebra chain
  converges in $\omega+\omega$ steps. 
\end{corollary}

\noindent
Indeed, every set functor preserves nonempty binary intersections~\cite[Prop.~2.1]{trnkova69},
and every endofunctor on $\KVec$ preserves finite intersections~\cite[Ex.~4.3]{AdamekSousa24}.

Recall that the assumptions in~\autoref{T:main} are that the category be a DCC-category and the
endofunctor is finitary and preserves nonempty binary intersections. The next example shows that 
the latter cannot be dropped. It also shows that there is no uniform bound on
the convergence of the terminal-coalgebra chain for finitary functors on DCC-categories.
\begin{example}\label{E:n-steps}
  For every ordinal $n$, we present a DCC-category and a finitary
  endofunctor which needs $n$ steps for its terminal-coalgebra chain to converge. The
  category is the complete lattice of all subsets of $n$ (considered as the set of all
  ordinals $i<n$).  To see that this is a DCC-category, note that the finitely
  presentable objects of this category are the finite subsets of $n$ and that, in
  every poset, strong epimorphisms are identities. Thus, there are no infinite strictly descending
  chains of subobjects or strong quotients of a finitely presentable object. 

  The functor is the monotone map $F$ defined by $F\emptyset = \emptyset$, and on all other
  sets $X \subseteq n$,
  \[
    FX = X \setminus\set{\min X}.
  \]
 This is monotone, since given $X \subseteq Y$, if $X$ contains $\min
 Y$, then $\min X = \min Y$. The only coalgebra for $F$ is empty; thus $\nu F = \emptyset$.

 The functor $F$ is finitary because for every directed union $X= \bigcup_{t\in T} X_t$ of
 nonempty subsets of $n$, $\min X$ lies in some $X_t$. Since the union is directed, $X$ is also
 a union of all~$X_s$ where $s \geq t$. Then $\min X$ is contained in all $X_s$. It follows
 that $\min X_s = \min X$, thus $FX_s = X_s \setminus \set{\min X}$ for all $ s \geq
 t$. Consequently,
 \[
   \colim FX_s = \colim X_s \setminus\set {\min X} = X\setminus
   \set{\min X} = FX.
 \]
 To see that $F$ does not preserve nonempty binary intersections, consider
 $n = \set{0,1,2,3}$, $X = \set{0,2,3}$ and $Y = \set{1,2,3}$. Then $F(X \cap Y) = \set{3}$
 while $FX \cap FY = \set{2,3}$.
 
 The terminal-coalgebra chain $V_i$ is given by $V_0 = n$ and $V_i= n\setminus i$ for all
 $ 0<i<n$, which is easy to prove by transfinite induction. Thus, that chain takes precisely
 $n$ steps to converge to the empty set, the terminal coalgebra for $F$.
\end{example}

As an open problem we leave the question whether the condition that $\A$ be a DCC-category
can be dropped from \cref{T:main}.
\begin{oproblem}
  \smnote{Chase Ford wrote to me on mattermost that they have such an example on the
    category of $\Z$-modules.}
  Find a finitary endofunctor preserving nonempty binary intersections on a locally finitely
  presentable category whose terminal-coalgebra chain does not converge in $\omega+\omega$
  steps.
\end{oproblem}

\subsection{Finitary Endofunctors on Metric Spaces}

We have seen in~\cref{T:main} a result which gives a sufficient condition for an endofunctor to
have a terminal coalgebra in $\omega+\omega$ steps.  This result does not apply to $\Met$, the
category of extended metric spaces (distances can be $\infty$) and non-expanding maps, 
since that category is not locally finitely presentable; in fact,
the empty space is the only finitely presentable object~\cite[Rem.~2.7]{AdamekR22}.  However,
for finitary endofunctors on $\Met$, we are able to prove an analogous result.  To do this, we
work with finite spaces in lieu of finitely presentable objects. Moreover, note that there is a
bijective correspondence between subobjects of~$M$ represented by isometric embeddings and
subspaces of $M$ (i.e.~subsets $S\subseteq M$ equipped with the metric of $M$ restricted to
$S$): indeed, for every subspace $S \subseteq M$, the inclusion $S \subto M$ is an isometric
embedding, and conversely, if $f\colon M' \to M$ is any isometric embedding, then it is monic
and represents the same subobject of $M$ as the inclusion map $f[M'] \subto M$ of the subspace
on the image of $f$. We need the following fact.

\begin{lemma}\label{L:finsub}
  Every metric space is a directed colimit of the diagram of all its
  finite subspaces.
\end{lemma}
\begin{proof}
  Fix a metric space $M$. Let $f_S\colon S \to A$ be a cocone of the diagram of all finite
  subspaces $m\colon S \subto M$ of $M$. Then there is a unique map $f\colon M \to A$ which
  restricts to $f_S$ for each finite subspace: $f \o m = f_S$. This map is non-expanding: given
  elements $x,y \in M$, let $S$ be the subspace of $M$ given by $\set{x,y}$. Since
  $f_S \colon S \subto A$ is non-expanding, the distance of $f(x)$ and $f(y)$ in $A$ is at most
  equal to the distance of $x$ and $y$ in $M$, that is, in $S$.
\end{proof}
\begin{remark}
  One easily derives that, given a metric space $M$ and a finite subspace $S \subto M$, the
  metric space $M$ is the directed colimit of the diagram of all its finite subspaces containing~$S$
  (cf.~\cref{rem-LP}\ref{rem-LP-4}).
\end{remark}
\begin{proposition}
  The category $\Met$ has finitary $\omega^\opp$-limits in the following sense: for every limit
  $l_n\colon L \to A_n$ ($n < \omega$) of an $\omega^\opp$-chain and every finite subobject
  $m\colon M \monoto L$, some morphism $l_k \cdot m\colon M \to A_k$ is monic.
\end{proposition}
\begin{proof}
  This follows since $\Set$ has finitary $\omega^\opp$-limits (\cref{P:ACDC}) because the
  forgetful functor into $\Set$~(1) preserves limits and (2)~preserves and reflects
  monomorphisms.
\end{proof}
\begin{lemma}\label{L:emb}
  Let $F$ be a finitary endofunctor on $\Met$ preserving isometric
  embeddings. For every non-expanding map $q\colon Q \to FM$ where $Q$
  is finite, there exists a factorization through~$Fm$ for some finite
  subspace $m\colon S \subto M$:
  \[
    \begin{tikzcd}
      &
      FS
      \ar[>->]{d}{Fm}\\
      Q
      \ar{r}{q}
      \ar[dashed]{ru}
      &
      FM
    \end{tikzcd}
  \]
\end{lemma}
\begin{proof}
  \begin{enumerate}
  \item\label{L:emb:1} Given a directed diagram $D$ in $\Met$ of metric spaces $A_i$
    ($i \in I$) and subspace embeddings $a_{i,j}\colon A_i \subto A_j$
    ($i \leq j$), the colimit $C$ is the union $\bigcup_{i\in I} A_i$
    with the metric inherited from the subspaces: given
    $x, y \in \bigcup_{i\in I} A_i$, the distance $d(x,y)$ in $C$ is
    their distance in $A_i$ for any $i \in I$ such that $x,y \in A_i$.

    An analogous statement holds for a directed diagram whose connecting morphisms are
    isometric embeddings.

  \item
    Given $q\colon Q \to FM$, let $D_M$ be the directed diagram of all
    finite subspaces of $M$ and all inclusion maps. Its colimit is $M$
    using \cref{L:emb:1}. Since $F$ is finitary, $FM$ is the colimit
    of $FD_M$, which is a directed diagram of isometric
    embeddings. The image $q[Q]$ is a finite subspace of $FM$. By
    \cref{L:emb:1}, there exists a finite subspace $m\colon S \subto
    M$ such that the colimit injection $Fm$ of $FC = \colim FD_M$
    satisfies $q[Q] \subseteq Fm[FS]$. Let $q'\colon Q \to FS$ be the
    unique map such that $q = Fm \o q'$. Then $q'$ is non-expanding
    because so is $q$, and $Fm$ is an isometric embedding. 
    \qedhere
  \end{enumerate}
\end{proof}

The following theorem has a proof analogous to that of
\Cref{T:main}. Recall that a functor preserving nonempty binary
intersections also preserves monomorphisms. This time, we need the extra
condition that also isometric embeddings are preserved:
\begin{theorem}\label{T:MetOplusO}
  For every finitary endofunctor $F$ on $\Met$ preserving nonempty binary intersections and
  isometric embeddings, the terminal-coalgebra chain converges in $\omega + \omega$~steps.
\end{theorem}
\begin{proof}
  We again use \Cref{P:modest}. By \cref{R:nonempty}\ref{R:nonempty:3}, we may assume without
  loss of generality that all $V_i$, $i \leq \omega + \omega$ are nonempty.
  \begin{enumerate}
  \item The morphism $m \colon V_{\omega+1} \to V_\omega$ is monic:
    given a nonempty space $Q$ and $q, q'\colon Q \to FV_\omega$ such that $m \o q = m\o q'$,
    we prove that $q = q'$. By \Cref{L:finsub}, we may assume that $Q$ is
    finite. Thus, there exists a nonempty finite subspace
    $m_t\colon M_t \subto Q$ such that both $q$ and~$q'$ factorize
    through $Fm_t$: we have morphisms $r, r'\colon Q \to FM_t$ such
    that $q = Fm_t \o r$ and $q' = Fm_t \o r'$. As in
    \cref{T:main:1} of the proof of
    \Cref{T:main}, we derive $r = r'$. Since $Fm_t$ is
    monic (because $F$ preserves nonempty binary intersections), this proves $q = q'$.
    
  \item We prove that $F$ preserves nonempty limits of $\omega^\opp$-chains of 
    monomorphisms
    \[
      a_i\colon A_{i+1} \monoto A_i
      \qquad\text{($i < \omega$)}.
    \]
    Let $\ell_i\colon L \to A_i$ be the limit cone. Given a cone
    $q_i\colon Q \to FA_i$ ($i < \omega$), we only need to find a
    morphism $q\colon Q \to FL$ such that $q_i = F\ell_i \o q$ ($i <
    \omega$).

    Using \Cref{L:finsub}, we may assume that $Q$ is finite. We define a subchain $(B_i)$ of
    $(A_i)$ carried by nonempty binary subspaces $u_i\colon B_i \subto A_i$, together with
    cones $p_i \colon Q \to FB_i$ and $m_i\colon C\to B_i$ such that $Fu_i \o p_i = q_i$ and
    $c_i = u_i \o m_i$.  We use recursion analogous to that in \cref{T:main:2} of the proof of
    \Cref{T:main}.  In order to define $B_0$, $u_0$, and $p_0$, use \Cref{L:emb}: there is a
    nonempty binary subspace $u_0\colon B_0 \subto A_0$ and a morphism $p_0\colon Q \to FB_0$
    such that $q_0 = Fu_0 \o p_0$.
    The induction step and the rest of the proof is as in \Cref{T:main}. 
    \qedhere
  \end{enumerate}
\end{proof}
\begin{example}\label{E:Hausdorff}
  \begin{enumerate}
  \item\label{E:Hausdorff:1} The Hausdorff functor $\H\colon \Met\to \Met$ maps a metric
    space~$X$ to the space of all compact subsets of $X$ equipped with the Hausdorff
    distance\footnote{The definition goes back to Pompeiu~\cite{Pompeiu05} and was popularized by
      Hausdorff~\cite[p.~293]{Hausdorff14}.} given by
    \[
      \bar{d}(S,T)
      =
      \max \left(\sup\nolimits_{x \in S} d(x,T), \sup\nolimits_{y \in T}d(y,S)\right),
      \qquad\text{for $S, T \subseteq X$ compact},
    \]
    where $d(x,S) = \inf_{y \in S} d(x,y)$. In particular
    $\bar d(\emptyset, T) = \infty$ for nonempty compact sets $T$. For
    a non-expanding map $f\colon X \to Y$ we have $\H f\colon S \mapsto
    f[S]$.

    We now show that the Hausdorff functor preserves isometric embeddings. Indeed, if
    $m\colon X \subto Y$ is the inclusion of a subspace, then $\H m$ preserves distances:
    given compact subsets $S$ and $T$ of the metric space $X$, then for every $x \in S$, we
    have that the distance $d(x,T)$ is the same in $X$ and~$Y$. By symmetry, the Hausdorff
    distance of $S$ and $T$ is also the same in~$\H X$ and in~$\H Y$.

    The functor $\H$ also preserves nonempty binary intersections: for every pair $A_1, A_2$
    of subsets of a metric space $B$, the compact subsets of the space $A_1 \cap A_2$ are
    precisely the compact subsets of $B$ contained in $A_i$ for $i = 1, 2$. However, $\H$ is
    not finitary on $\Met$: consider any infinite compact space $M$ (e.g.~$M = [0,1]$) and
    write it as a directed colimit. Then we have $M \in \H M$, but $M$ is not contained
    in~$\H S$ for any finite subspace $S \subto M$. This shows that \cref{L:emb} does not
    hold for $\H$ and $q\colon Q = \set{1} \to \H M$ which picks $M$. Note that as an
    endofunctor on $\CMS$, the Hausdorff functor is finitary~\cite[Ex.~3.13]{AdamekEA15}.

  \item The finite power-set functor has a lifting
    $\powf\colon \Met \to \Met$ that maps a metric space~$X$ to the space~$\powf X$ of all
    finite subsets of $X$ equipped with the Hausdorff distance. This functor is clearly
    finitary; in fact, it is the free algebra monad for the variety of quantitative
    semilattices~\cite[Sec.~9]{MardareEA16}.
    It preserves isometric embeddings because it is a subfunctor of $\H$. Finally, $\powf$
    preserves nonempty binary intersections because it is a lifting of a set functor and
    intersections of metric spaces are formed on the level of sets.
  \end{enumerate}
\end{example}
\begin{corollary}
  The lifted functor $\powf\colon \Met \to \Met$ has a terminal coalgebra
  $\nu \powf = V_{\omega + \omega}$.
\end{corollary}

We shall see in \cref{E:Haus-ter} that for the Hausdorff functor on $\Met$, the
terminal-coalgebra chain converges in $\omega +\omega$ steps, even though that functor is not
finitary.

\section{Kripke Polynomial Functors}
\label{S:Kripke}

We turn to the Kripke polynomial set functors.  The name stems from \emph{Kripke structures}
used in modal logic. Our definition below is a slight generalization of the (finite) Kripke
polynomial functors presented by Jacobs~\cite[Def.~2.2.1]{jacobs}. We admit arbitrary products
in lieu of just arbitrary exponents. (Kripke polynomial functors using the full power-set
functor were originally introduced by R\"{o}{\ss}iger~\cite{Roessiger00}.)
\begin{defn}\label{D:Kripke}
  The \emph{Kripke polynomial functors} $F$ are the set functors built
  from the finite power-set functor, constant functors and the
  identity functor, by using product, coproduct, and composition. In
  other words, Kripke polynomial functors are built according to the
  following grammar:
  \[
    \textstyle
    F ::= \powfin \mid  A \mid \Id \mid \prod_{i\in I} F_i \mid
    \coprod_{i\in I} F_i \mid FF,
  \]
  where $A$ ranges over all sets (and is interpreted as a constant
  functor) and $I$ is an arbitrary index set.
\end{defn}
\begin{remark}
\label{R:omit}
The constant functors could be omitted from the grammar since they are obtainable from the
rest of the grammar.  The constant functor with value $1$ is the empty product. Every set
(thus every constant set functor) is a coproduct of copies of $1$.
\end{remark}
\begin{example}\label{E:fin-branch} 
\begin{enumerate}
\item  The Kripke polynomial functor $FX = \powfin(A \times X)$ is the type
  functor of finitely branching labelled transition systems with a set
  $A$ of actions.
\item The closely-related functor $FX = \powfin^{+}X$ of all nonempty finite sets is not a Kripke polynomial functor.
This follows from \Cref{P:Kripke-only}.    
\item 
The functor  $FX = \powfin^{+}X$  is finitary.
The functor $FX   = X^\N$, where $\N$ is the set of natural numbers, is a Kripke polynomial functor which is not finitary.
\end{enumerate}  
\end{example}

\begin{proposition}\label{P:Kripke-only}
  The only Kripke polynomial functors $F$ with $|F1| = 1$ are the powers of the identity functor.
\end{proposition}

\begin{proof}
  We prove by induction on the Kripke polynomial functor $F$ that if
  $|F1| = 1$, then there is some $n$ such that $F \simeq \Id^n$.  (With
  $n = 0$, we mean that $F\simeq C_1$.)  Here and below, we write
  $F\simeq G$ to indicate that $F$ and $G$ are naturally isomorphic
  functors.

  Our base cases are for $\Id$ and $\powf$ (see~\Cref{R:omit}).
  These cases are clear.  Thus, it is sufficient to prove that given
  $F$ as a product, coproduct or composite of functors $F_i$
  ($i \in I$) satisfying the proposition, then from $|F1|=1$ it
  follows that $F\simeq \Id^n$ for some cardinal $n$.
  \begin{enumerate}  
  \item Let $F = \prod_{i\in I} F_i$.  Since $|F1| = 1$, the same
    holds for all $F_i$.  Thus, we have cardinals $n_i$ for $i\in I$
    such that $F_i \simeq \Id^{n_i}$.  We conclude that
    $F\simeq \Id^n$ for $n = \sum_{i\in I} n_i$.  (The index set $I$
    might well be empty, and in that case $n = 0$ and
    $F \simeq C_1 \simeq \Id^0$.)

  \item Suppose that $F = \coprod_{i\in I} F_i$.  There is some
    $j\in I$ such that $|F_j 1| = 1$, and $F_i 1 = \emptyset$ for
    $i \neq j$.  This implies that $F_i = C_{\emptyset}$ for
    $i \neq j$: given a set $X$ the function $f\colon X\to 1$ is
    mapped to $F_i f\colon FX \to \emptyset$.  Hence
    $F_i X = \emptyset$.  We conclude that
    $F \simeq F_j \simeq \Id^{n_j}$.

  \item Let $F = G\o H$.  If $H 1 = \emptyset$, then the argument in
    the last item shows that $H = C_{\emptyset}$.  Thus,
    $F = G \o C_{\emptyset}$ is a constant functor.  Since $|F1| = 1$,
    we see that $F \simeq C_1$.  If $H 1\neq \emptyset$, then the
    split epimorphism $f\colon H 1 \to 1$ is mapped to an epimorphism
    \[
      Ff\colon F1 \to G 1.
    \]
    Since $Ff$ is surjective, $|G 1| = 1$.  By induction hypothesis,
    there is some $n$ such that $G \simeq \Id^n$. Hence, we have 
    $|H 1|^n = |F1 | =1$, which implies that $|H1| = 1$ or else
    $n=0$. In the first case,  we have some $m$ such that
    $H \simeq \Id^m$. It follows that
    $F \simeq (\Id^n)\o (\Id^m) = \Id^{m n}$. 
    In the second case, we obtain $F\simeq \Id^0 \o H = C_1 \o H = C_1$.  \qedhere
  \end{enumerate}
\end{proof}
\begin{remark}\label{R:finitary-Kripke}
  Recall from \cref{S:lfp} that an endofunctor is finitary if it preserves directed colimits.
  Worrell~\cite{worrell:05} proved that for every finitary set functor, the terminal-coalgebra
  chain converges in $\omega+\omega$ steps. We shall proceed to prove a version of Worrell's result
  but for Kripke polynomial functors.
\end{remark}
\begin{proposition}\label{P:Kripke}
  For every Kripke polynomial functor $F$, the terminal-coalgebra chain
  converges in $\omega+\omega$ steps: $\nu F = V_{\omega + \omega}$.%
\end{proposition}
\begin{proof}
  We use \Cref{P:modest}, taking $\M$ to be the class of all monomorphisms.
  \begin{enumerate}
  \item\label{P:Kripke:1} We first observe that $F$ preserves
    monomorphisms and intersections of monomorphisms.  This is clear
    for constant functors and for $\Id$, and it is easy to see for
    $\powf$. Moreover, these properties are clearly preserved by
    product, coproduct, and composition.

  \item\label{P:Kripke:2} Let $(X_n)_{n < \omega}$ be an
    $\omega^\opp$-chain in $\Set$. Then the canonical morphism
    $m\colon F(\lim X_n) \to \lim FX_n$ is monic.  This is obvious for
    constant functors and $\Id$.  Let us check it for $\powfin$.
    Denote the limit projections by $\ell_n\colon \lim X_n \to X_n$
    and $p_n\colon \lim \powf X_n \to \powf X_n$ ($n < \omega$); the
    canonical morphism $m$ is unique such that
    $p_n \cdot m = \powf\ell_n$. Now given $S \neq T$ in
    $\powf(\lim X_n)$, without loss of generality we can pick
    $x \in T \setminus S$. Using that the $\ell_n$ are jointly monic,
    for every $s \in S$ we can choose $n< \omega$ such that
    $\ell_n(x) \neq \ell_n(s)$. Since $S$ is finite, this choice can
    be performed independently of $s \in S$. Thus,
    $\ell_n(x) \not\in \ell_n[S]$, and hence
    $\powf\ell_n(T)\neq \powf\ell_n(S)$. Thus, $\powf\ell_n$ is a jointly
    monic family.  Since $p_n \cdot m = \powf\ell_n$, we see that $m$
    is monic.

  \item\label{P:Kripke:2.5} An induction on Kripke polynomial functors
    $F$ now shows that $m\colon V_{\omega+1}\to V_{\omega}$ is monic.
    We have seen this for the base case functors in
    \Cref{P:Kripke:2}.  The desired property that $m$ is monic is
    preserved by products, coproducts, and composition. In particular,
    for a composition $FG$ note that the canonical morphism for $FG$
    is the composition
    \[
      FG(\lim X_n) \xra{Fm} F(\lim GX_n) \xra{m'} \lim FGX_n,
    \]
    where $m$ is the canonical morphism for $G$ with respect to the given
    $\omega^\opp$-chain and $m'$ the one for $F$ and the
    $\omega^\opp$-chain $(GX_n)_{n< \omega}$. So this morphism
    $m'\o Fm$  is monic since
    both $m$ and $m'$ are so and $F$ preserves monomorphisms by
    \Cref{P:Kripke:1}. 

  \item\label{P:Kripke:3} Since $F$ preserves monomorphisms, we see
    that $Fm$, $FFm$, \dots\/ are monic. We obtain a decreasing chain of
    subobjects $ V_{\omega+n } \monoto V_\omega$.  Therefore, the
    limit $V_{\omega + \omega} = \lim_{n < \omega} V_{\omega + n}$ is
    simply the intersection of these subobjects. From
    \Cref{P:Kripke:1} we know that $F$ preserves this limit. It
    follows that $\nu F = V_{\omega + \omega}$, as desired.\qedhere
  \end{enumerate}
\end{proof}
\begin{corollary}\label{C:K-I}
  Every Kripke polynomial set functor $F$ has an initial algebra.
\end{corollary}

\noindent
This follows from \Cref{P:Kripke}, \Cref{E:cms}\ref{E:cms:set},
and \Cref{T:initial} since $F$ preserves monomorphisms.

\begin{corollary}
  Kripke polynomial functors have a cofree comonad obtained
  in $\omega + \omega$~steps.
\end{corollary}
\noindent
This follows from \Cref{P:Kripke} and \Cref{C:cofree}: if $F$ is a
Kripke polynomial functor, then so is $F(-) \times Y$ for every set $Y$.

\begin{example}
  \begin{enumerate}
  \item For $FX = X+1$, a cofree comonad
    $F_\sharp$ is obtained in $\omega$ steps: 
    \[
      F_\sharp X = X^* + X^\omega
    \]
    is the set of all finite and infinite words on the set $X$.
    
  \item We recall that, for a (possibly infinitary) signature $\Sigma$, the polynomial set
    functor $H_\Sigma$ is a coproduct of power-functors $(-)^n$, one copy for every $n$-ary
    operation symbol in $\Sigma$. A cofree comonad is obtained in $\omega$ steps:
    $F_\sharp Y$ is the set of all $\Sigma_Y$-trees where $\Sigma_Y$ is the signature that
    contains, for every operation symbol $\sigma$ from $\Sigma$ a copy $\sigma_y$ for every
    $y \in Y$. In other words, $\Sigma_Y$-trees are (isomorphic to) $\Sigma$-trees having a
    extra label from $Y$ on every node.
  \end{enumerate}
\end{example}

\begin{example}
  In locally finitely presentable categories, $\omega^{\opp}$-limits
  need not be finitary in general. For example, the category $\Ab$ of
  abelian groups does not have that property: Consider the chain $A_n$
  of quotients of the additive group $\Z$ modulo the subgroups
  $2^n \Z$, represented by
  \[
    A_n = \set{0, 1, ..., 2^n -1}.
  \]
  The connecting morphisms $f_n\colon A_{n+1} \to A_n$ are given by
  \[
    f_n (t) = t \mathbin{\mathrm{mod}} 2^n \qquad (t=0,...,2^{n+1}-1).
  \]
  Let $\ell_n\colon L \to A_n$ ($n \in \N$) be the limit of this chain.  The quotient maps
  $q_n\colon \Z \epito A_n$ yield a cone, so we obtain a unique morphism $m\colon \Z \to L$
  such that $\ell_n \cdot m = q_n$ for all $n \in \N$. This morphism~$m$ is clearly monic,
  since the $q_n$ are jointly monic: given $k \neq k'$, then $\ell_n(k) \neq \ell_n{k'}$,
  for any $n$ such that $2^n > \max\{k,k'\}$. However, none of the $\ell_n \cdot m = q_n$ is
  monic.
  \takeout{
  The family of elements $1$ of $A_n$ ($n<\omega$) is compatible with
  that chain.  Thus, in the limit~$L$ there is a unique element $x$
  mapped by all $\ell_n \colon L \to A_n$  to $1$:
  $\ell_n(x) = 1 $ for all~$n$.  The subgroup~$M$ of~$L$ generated by
  $x$ is infinite. Indeed, $nx \neq 0$ for all $n$ because
  $l_{n+1}(nx) = n \neq 0 $ in~$A_{n+1}$. Thus, none of the limit maps
  $\ell_n$ restricts to a monomorphism $M\monoto A_n$. But $M$ is
  finitely presentable, since in $\Ab$ this is the same as being
  finitely generated.}
\end{example}

\section{Trees and the Limit of the Terminal-Coalgebra Chain for $\powf$}
\label{S:app-worrell}

As mentioned in the Introduction, Worrell~\cite{worrell:05} described
the terminal coalgebra of $\powfin$ and the limit $V_{\omega}$ using
trees.  Worrell proved that $\powfin$ has a terminal coalgebra
consisting of the finitely branching strongly extensional trees (up to
isomorphism of trees). Moreover, the limit $V_\omega$ consists of all
compactly branching strongly extensional trees, where the notion of
compactness uses a particular pseudo-metric \eqref{psW} on the class
of all strongly extensional trees.  However, the proof of this
characterization result was not spelled out in his paper, or in
related papers such as Abramsky's work~\cite{abramsky}.  We provide a full exposition of these
results. The descriptions will be used in \Cref{S:example}.
%
%
%
%
 
  \subsubsection*{Roadmap}
  This section is long, and so we provide some comments to orient the reader.  We aim to
  characterize $V_{\omega} = \lim_{n<\omega} V_n$ for $\powfin$ in terms of trees.  We begin
  with a general discussion of trees, aiming for the concept of strong extensionality
  (\Cref{D-SE}\ref{partSE}) that is part of the characterization.  The compact branching
  condition comes a little later (\Cref{dcb}), but we must mention that the definition is a
  formal notion.  For a tree $t$, it is defined in terms of maps
  $\treepartial^t_n\colon t \to V_n$. These maps $\treepartial^t_n$ form a canonical
    cone on the $\omega^{\opp}$-chain $(V_n)$; in
    addition, for a node~$x$ of~$t$, $\treepartial^t_n(x)$ can be understood in terms of
    cutting the subtree rooted at $x$ at its $n$th level (\Cref{R:exttree}).  These maps
    give rise to a map $t\to V_{\omega}$; it is
    $x\mapsto (\treepartial^t_n(x))_{n<\omega}$. The compact branching condition has a
  combinatorial flavor rather than a metric one.  In \Cref{S:compactness}, we formulate
  Worrell's pseudometric on the class of trees and show that a tree is compactly branching
  in our sense iff it is compact in appropriate pseudometric sense (\Cref{corWor}).  Our
  characterization of $V_{\omega}$ goes via associating a tree $\temptree_x$ to each
  $x\in V_{\omega}$ (\Cref{N:tx}\ref{N:tx:3}), then showing that each $\temptree_x$ is
  strongly extensional and compactly branching.  An arbitrary tree $t$ has an associated
  tree $t^*\in V_{\omega}$, and assuming that the original~$t$ is strongly extensional and
  compactly branching, it happens that $t^* = t$.  This leads us to the identification of
  the limit set $V_{\omega}$ and also the terminal coalgebra in~\Cref{T:s-ext}.

 We turn from the map to the road itself.
A \emph{tree} is a directed graph $t$ with a distinguished node
$\root(t)$ from which every other node can be reached by a unique
directed path.
Every tree in our sense must have a root, so there is no empty tree.
All of our trees are \emph{unordered}, that is, there is no order on
the children of a node. We always identify isomorphic trees.
\removeThmBraces
\begin{defn}
 \label{D-SE}
  \begin{enumerate}
  \item We use the notation $t_x$ for the subtree of $t$ rooted in the
    node $x$ of $t$.
    \item A tree $t$ is \emph{extensional} if for every node $x$
      distinct children $y$ and $z$ of $x$ give different (that is,
      non-isomorphic) subtrees $t_y$ and $t_z$.

  \item A \emph{graph bisimulation} between two trees $t$ and $u$ is a
    relation between the nodes of $t$ and the nodes of $u$ with the
    property that  whenever $x$ and $y$ are related: (a) every child of $x$
    is related to some child of $y$, and (b) every child of $y$ is
    related to some child of $x$.
    
  \item A \emph{tree bisimulation} between two trees $t$ and $u$ is a
    graph bisimulation such that 
    \begin{enumerate}
    \item The nodes $\root(t)$ and $\root(u)$
    are related; the roots are not
      related to other nodes;
      and
    \item  whenever two nodes are related, their parents are also related.
    \end{enumerate}
    
  \item Two trees are \emph{tree bisimilar} if there is a tree
    bisimulation between them.

  \item\label{partSE}
   A tree $t$ is \emph{strongly extensional}
    if every tree bisimulation on it is a
    subset of the diagonal
    \[
      \Delta_t = \set{(x,x): x\in t}.
    \]
    In other words, $t$ is strongly extensional iff distinct children
    $x$ and $y$ of the same node define subtrees $t_x$ and $t_y$ which
    are \emph{not} tree bisimilar.

  \end{enumerate}
\end{defn}
\resetCurThmBraces
\begin{remark}\label{rem:strongext}
  \begin{enumerate}
  \item Every composition and every union of tree bisimulations is again a tree
    bisimulation. In addition, the opposite relation of every tree bisimulation is a tree
    bisimulation: if $R$ is a tree bisimulation from $t$ to $u$, then $R^{\opp}$ is a tree
    bisimulation from $u$ to $t$. Consequently, the largest tree bisimulation on every tree is
    an equivalence relation.

  \item\label{rem:strongext:5} A subtree $s$ of a strongly extensional tree $t$ is strongly
    extensional. Indeed, if $R$ is a tree bisimulation on $s$, then $R\cup\Delta_t$ is a tree
    bisimulation on $t$.  Since $R\cup\Delta_t\subseteq \Delta_t$, 
    $R\subseteq \Delta_s$.

  \end{enumerate}
\end{remark}

\begin{proposition}\label{P:finite}
  A finite tree is extensional iff it is strongly extensional.
\end{proposition}
\begin{proof}
  It is clear that strong extensionality implies extensionality.  In
  the other direction let~$t$ be a finite extensional tree, and let
  $R$ be a tree bisimulation on it.  We claim that if $x$ and~$y$ are
  nodes and $x \mathbin{R} y$, then the corresponding subtrees~$t_x$
  and~$t_y$ are equal. First notice that every node of $t_x$ must be
  related by~$R$ to some node of $t_y$ (to see this, use induction on
  the depth of nodes, i.e.~their distance from the root) and vice
  versa. Thus, $t_x$ and $t_y$ have the same height, $n$ say. We now
  prove $t_x = t_y$ by induction on $n$. For $n = 0$, the result is
  obvious because nodes of height $0$ are leaves.  Assume our result
  for $n$, and let $x$ and $y$ be related by $R$ and of height
  $n+1$. Then by the induction hypothesis and extensionality of $t$,
  for every child $x'$ of $x$, there is a \emph{unique} child $y'$ of
  $y$ with ($x' \mathbin{R} y'$ and hence) $t_{x'} = t_{y'}$; and
  vice-versa. This implies that $t_x = t_y$.  It now follows that $t$
  is strongly extensional.
\end{proof}

\begin{lemma}
 \label{lem:w1} 
  If $t$ and $u$ are strongly extensional and related by a tree
  bisimulation, then~\mbox{$t = u$}.
\end{lemma}
\begin{proof}
  Let $R$ be a tree bisimulation between $t$ and $u$.  By
  \Cref{rem:strongext}, $R^{\opp} \o R$ is a tree bisimulation on
  $t$, whence $R^{\opp} \o R \subseteq \Delta_t$ by strong
  extensionality. But every node of $t$ is related to at least one
  node of $u$ (use induction on the depth of nodes) implying that
  $R^{\opp} \o R = \Delta_t$. Similarly, $R \o R^{\opp} =
  \Delta_u$. Thus, $R$ (is a function and it) is an isomorphism of
  trees, and we identify such trees.
\end{proof}
\begin{notation}\label{N:partial}
\begin{enumerate}
\item Let $\TT$ be the class of trees.  We define maps
  $\partial_n\colon \TT \to V_n= \powf^n 1$ as follows: $\partial_0$
  is the unique map to $1$, and given the map $\partial_n$ and a tree
  $t$, we put
  \begin{equation}
  \label{eq-partial-np}
    \partial_{n+1}(t) = \set{\partial_n(t_x ): \text{$x$ is a child of
        the root of $t$}}. 
  \end{equation}
  On the right we have a subset of $\powf^n 1$, and this is an element
  of $\powf^{n+1}1$.
  
\item The trees $t$ and $u$ are \emph{Barr equivalent} if 
  $\partial_n t = \partial_n u$ for all $n$.  We write $t\approx u$ in
  this case.
  
\item For every tree $t$, we define maps
  $\treepartial^t_n \colon t \to V_n = \powf^n 1$ in the following
  way: $\treepartial^t_0$ is the unique map $t\to 1$, and for all
  nodes $x$ of $t$,
  $\treepartial^t_{n+1}(x) = \set{\treepartial^t_n(y): \mbox{$y$ is a
      child of $x$ in $t$}}$.  This family of maps $\treepartial^t_n$
  is a cone: we have
  $\treepartial^t_n = v_{m,n} \cdot \treepartial^t_m$ for every
  connecting map $v_{m,n}\colon \powfin^m 1 \to\powfin^n 1$,
  $m \geq n$. Hence, there is a unique map
  $\treepartial^t_{\omega}\colon t\to V_{\omega}$ such that
  $\ell_n\o\treepartial^t_\omega = \treepartial^t_n$ for all $n$.
\end{enumerate}
\end{notation}
\begin{remark}\label{R:exttree}
  Note that $V_n = \powfin^n 1$ is in bijective correspondence with the set of all
  extensional trees of height at most $n$. Indeed, $1$ corresponds to
  the singleton set consisting of the root-only tree, and every finite
  set of extensional trees in $V_{n+1} = \powfin V_n$ corresponds to
  the extensional tree obtained by tree-tupling the trees from the given set.

Under this correspondence, the maps $\treepartial^t_n\colon t \to V_n$ can be
    understood as follows: for a node $x$ of $t$, $\treepartial^t_n(x)$ cuts the subtree
    rooted at $x$ at its $n$th level and then takes its extensional quotient (i.e.~the
    extensional tree obtained by successively identifying isomorphic subtrees rooted at
    children of a node).
  
For example,  let $t$ be the tree on the left below.
    \[
      t\colon
    \begin{tikzpicture}[
      x=25mm,scale=.4,
      dot/.style={circle, fill, minimum size=#1, inner sep = 0, outer sep =0},
      dot/.default=4pt, 
      baseline = (B.south),
      sibling distance=60
      ]
      \node[dot, label=above:$x_0$] {}
      child{node[dot,label=left:$x_1$] (B) {}}
      child{node[dot,label=left:$x_2$] {}
        child{node[dot,label=left:$x_4$] {}}
      }
      child{node[dot,label=right:$x_3$] {}
        child{node[dot,label=right:$x_5$] {}
          child{node[dot,label=right:$x_6$] {}}
        }
      };
    \end{tikzpicture}
    \qquad\qquad
    c_1\colon
    \begin{tikzpicture}[
      x=25mm,scale=.4,
      dot/.style={circle, fill, minimum size=#1, inner sep = 0, outer sep =0},
      dot/.default=4pt, 
      sibling distance=60
      ]
      \node[dot] {};
    \end{tikzpicture}
    \qquad\qquad
    c_2\colon
    \begin{tikzpicture}[
      x=25mm,scale=.4,
      dot/.style={circle, fill, minimum size=#1, inner sep = 0, outer sep =0},
      dot/.default=4pt, 
      baseline = (B.south),
      sibling distance=60
      ]
      \node[dot] (B) {}
      child{node[dot] {}};
    \end{tikzpicture}
    \qquad\qquad
    c_3\colon
    \begin{tikzpicture}[
      x=25mm,scale=.4,
      dot/.style={circle, fill, minimum size=#1, inner sep = 0, outer sep =0},
      dot/.default=4pt, 
      baseline = (B.south),
      sibling distance=60
      ]
      \node[dot] (B) {}
      child{node[dot] {}
          child{node[dot] {}}
      };
    \end{tikzpicture}
    \qquad\qquad
    \treepartial^t_2(x_0)\colon
   \begin{tikzpicture}[
      x=125mm,scale=.4,
      dot/.style={circle, fill, minimum size=#1, inner sep = 0, outer sep =0},
      dot/.default=4pt, 
      baseline = (B.base),
      sibling distance=60
      ]
      \node[dot] {}
      child{node[dot] (B) {}}
      child{node[dot] {}
        child{node[dot] {}}
      };
    \end{tikzpicture}
    \] 
  For $i = 1,2,3$, let $c_i$ be the $i$-element chain, considered as a tree.  
  For all nodes $x$, $\treepartial^t_0(x) = c_1$.   
  For $x\in \set{x_1, x_4, x_6}$ and $i \geq 1$,   $\treepartial^t_i(x) = c_1$ as well.
  For all other $x$,   $\treepartial^t_1(x) = c_2$. 
  For $x \in \set{x_2,x_5}$ and $i \geq 2$, $\treepartial^t_i(x) = c_2$.
  For $i \geq 2$, $\treepartial^t_i(x_3)= c_3$.
  Turning to $x_0$, $\treepartial^t_2(x_0)$ is shown on the right above.
  Finally, for $i \geq 3$, $\treepartial^t_i(x_0) = t$. 
\end{remark}

\begin{remark}\label{R:treepartial}
  \begin{enumerate}
  \item\label{R:treepartial:1} If
    $\treepartial^t_{n+1}(a) = \treepartial^t_{n+1}(b)$, then for all
    children~$a'$ of $a$, there is some child $b'$ of $b$ and
    $\treepartial^t_{n}(a') = \treepartial^t_{n}(b')$.  This is easy
    to see from the definition of $\treepartial^t_{n+1}$.%
  \item For all trees $t$,
    \( \treepartial^t_i(\root(t))= \partial_i(t).  \)
Furthermore, let $b\colon t \to \TT$ be given
    by $b(x) = t_x$.  Then
    $\treepartial^t_i = \partial_i\o b$.%
  \end{enumerate}
\end{remark}
\begin{defn} \label{dcb} Let $x_0, x_1, \ldots, $ be a sequence of
  nodes in a tree $t$, and let $y$ also be a node in $t$.  We write 
  \(
    \lim x_n = y
  \)
  to mean
  that for every $k$ there exist some $m$ such that
    $\treepartial^t_k$ merges all $x_m, x_{m+1}, x_{m+2}, \ldots$ with $y$: 
  $\treepartial^t_k(x_n) = \treepartial^t_k(y)$ for every $n \geq m$.

  A tree $t$ is \emph{compactly branching} if for all nodes $x$ of
  $t$, the set of children of $x$ is \emph{sequentially compact}:
  for every sequence of $(y_n)$ of children of $x$
  there is 
  a subsequence $(w_n)$ of $(y_n)$ and 
  some child $z$ of $x$ such that $\lim {w_n} = {z}$.
\end{defn}

\begin{example}\label{E-sat}
  The following tree $t$ is not compactly branching:
  \[
    t\colon \qquad
    \begin{tikzpicture}[
      x=25mm,scale=.4,
      dot/.style={circle, fill, minimum size=#1, inner sep = 0, outer sep =0},
      dot/.default=4pt, 
      baseline = (B.base),
      sibling distance=60
      ]
      \node[dot] {}
      child{node[dot,label=left:$y_0$] (B) {}}
      child{node[dot,label=left:$y_1$] {}
        child{node[dot] {}}
      }
      child{node[dot,label=right:$\quad\cdots$,label=left:$y_2$] {}
        child{node[dot] {}
          child{node[dot] {}}
        }
      };
    \end{tikzpicture}
    \]
    For use below, note that for  $p \geq m$, $\partial_p(t_{y_m}) = t_{y_m}$.
  This is proved by induction on $m$.

We verifiy that the tree is not compactly branching.  Consider the sequence $y_0$,
  $y_1$, $\ldots$.  We claim that for every subsequence $(y_{k_n})$ of this sequence $(y_n)$
  there is no $y_m$ such that $\lim_n y_{k_n} = y_m$.  To simplify the notation, we only
  verify this for the sequence $(y_n)$ itself.  It does not converge to any fixed element
  $y_m$ because for $p > m$,
    \[
      \treepartial^{t}_p(y_m)
      = 
      \partial_p(t_{y_m})
     = t_{y_m}
      \neq
      t_{y_p}
      =
      \partial_p(t_{y_p})
      =
      \treepartial^{t}_p(y_p).
    \]

    In contrast, the following tree is compactly branching (also observe that
    $t \approx t'$):
    \[
      t'\colon \qquad
      \begin{tikzpicture}[
        x=25mm,scale=.4,
        dot/.style={circle, fill, minimum size=#1, inner sep = 0, outer sep =0},
        dot/.default=4pt, 
        baseline = (B.base),
        sibling distance=70
        ]
        \node[dot] {}
        child{node[dot,label=left:$z$] (B) {}
          child{node [dot] {}
            child{node[dot] {}
              child{node[dot,label=below:$\vdots$] {}
              }
            }
          }
        }
        child{node[dot,label=left:$y_0$] {}}
        child{node[dot,label=left:$y_1$] {}
          child{node[dot] {}}
        }
        child{node[dot,label=right:$\quad\cdots$,label=left:$y_2$] {}
          child{node[dot] {}
            child{node[dot] {}}
          }
        };
      \end{tikzpicture}
    \]
Notice that $\partial_n(t_{z}) = t_{y_n}$ for all $n$.
    To check the compactness, consider a sequence of children of the
    root, say $({x_n})$.  If there is an infinite subsequence which
    is constant, then of course that sequence converges.  If not, then
    there is a subsequence of $({x_n})$, say $({w_n})$, where each
    $w_n$ is $y_k$ for some $k\geq n$.  We claim that in this case,
    $\lim_n ({w_n}) =z$.  This is because for all but
    finitely many $n$,  $\partial_n(t_{w_n}) = t_{y_n}$, since $w_n$ is $y_k$ for some $k \geq n$.
For such $n$,    $\treepartial^t_n(z) =
    \partial_n(t_z) = t_{w_n} = \partial_n(t_{y_n}) = \treepartial^t_n(w_n) $.
\end{example}

\takeout{
For such $n$,
\[
\begin{array}{lcll}
\rho_n^t(z) &  = & \partial_n(t_z)  & \mbox{see~\Cref{N:partial}}\\
 &  =  & t_{y_n } & \mbox{see above}\\ 
 &= & \partial_n(t_{w_n}) & \mbox{since $w_n$ is $y_k$ for some $k \geq n$} \\
 &  = & \rho_n^t(w_n)  & \mbox{see~\Cref{N:partial}}\\
 \end{array}
\]
 }

\begin{lemma}\label{seqK-part}
  If $t$ and $u$ are compactly branching, and if
  $\treepartial^t_{\omega}(\root(t)) = \treepartial^u_{\omega}(\root(u)) $, then
  there is a tree bisimulation between $t$ and $u$ which includes
  \(
    \set{(x,y) \in t\times u \colon \treepartial^t_{\omega}(x) =
      \treepartial^u_{\omega}(y)}.
  \)
\end{lemma}

\begin{proof}
  Given compactly branching trees $t$ and $u$, we define a relation 
  $R\subseteq t\times u$ inductively: 
  \begin{align*}
    x\mathbin{R} y \quad \mbox{iff} \quad &
    \text{(1)~$x = \root(t)$ and $y= \root(u)$, or $x$ and $y$ have $R$-related parents, and} \\
    &
    \text{(2)~$\treepartial^t_{\omega}(x) = \treepartial^u_{\omega}(y) $}.
  \end{align*}
  Let us check that $R$ is a tree bisimulation.  Suppose that $(x,y)$
  are related by $R$ as above, and let $x'$ be a child of $x$ in $t$.
  Using \Cref{R:treepartial}\ref{R:treepartial:1} we see that for each
  $n$, there is some child~$y'_n$ of $y$ in $u$ with
  $\treepartial^t_{n}(x') = \treepartial^u_{n}(y_n')$.  Consider the
  sequence $y'_0$, $y'_1$, $\ldots$.  Now
  $\treepartial^t_{n}(x') = \treepartial^u_{n}(y_m')$ if
  $m \geq n$, since $\treepartial^t_n$ and $\treepartial^u_n$ form
  cones:
  \(
    \treepartial^t_n(x')
    =
    v_{m,n} \cdot \treepartial^t_m(x')
    =
    v_{m,n} \cdot \treepartial^u_m(y_m')
    =
    \treepartial^u_n(y_m').
  \)
  By sequential compactness, there is a
  subsequence $z_0$, $z_1$, $\ldots$, and also some child $z^*$ of $y$
  such that $\lim z_n = z^*$.  Being a subsequence,
  $\treepartial^t_{n}(x') = \treepartial^u_{n}(z_m)$ whenever
  $m \geq n$.  Let us check that for all $n$,
  $\treepartial^t_n(x') = \treepartial^u_n(z^*)$.  To see this, fix
  $n$ and let $m\geq n$ be large enough so that for $p\geq m$,
  $\treepartial^u_n(z_p) = \treepartial^u_n(z^*)$.  Thus,
  $\treepartial^t_{n}(x') = \treepartial^u_{n}(z_m)
  =\treepartial^u_n(z^*) $.
  Hence, $\treepartial^t_\omega(x') = \treepartial^u_\omega(z^*)$,
  which shows $x' \mathbin{R} z^*$, as desired. 

  The other half of the verification that $R $ is a tree bisimulation is similar.
\end{proof}

\begin{corollary}
\label{corccb}
Two compactly branching trees are Barr equivalent (\Cref{N:partial})
iff they are tree bisimilar.
\end{corollary}
\begin{notation}\label{N:tx}
  In this section, $V_{\omega}$ denotes the limit of
  \eqref{eq:op-chain} for the finite power-set functor.
  \begin{enumerate}
  \item We take the elements of $V_{\omega}$ to be compatible
    sequences $(x_n)$. That is, $x_n \in \powfin^n 1$ and
    $\powfin^{n}!(x_{n+1}) = x_n$ for every $n < \omega$. To save on
    notation, we write $x$ for $(x_n)$.  We consider the relation
    $\leadsto$ on $V_{\omega}$ defined by
    \begin{equation}\label{leadsto}
      x\leadsto y \quad\mbox{iff}\quad
      \mbox{for all $n$, $y_{n}\in x_{n+1}$}.
    \end{equation}

  \item Let $L^+$ be the set of nonempty finite sequences of elements of  $V_{\omega}$.
  We write such a sequence with the notation
    $\pair{x^1, \ldots, x^n}$.  We consider the relation $\Rightarrow$
    on $L^+$ defined by
    \[
      \pair{x^1, \ldots, x^n} \Rightarrow \pair{y^1, \ldots, y^m}
      \quad\mbox{iff}\quad \mbox{$m = n+1$, $x^1 = y^1$, $\ldots$,
        $x^n = y^n$, and $x^n \leadsto y^{n+1}$}.
    \]
    In other words, $m = n+1$,
    $ \pair{y^1, \ldots, y^{m-1}} = \pair{x^1, \ldots, x^n}$, and
    $x^n \leadsto y^m$.

  \item\label{N:tx:3} For each $x\in V_{\omega}$, let $\temptree_x$ be the tree
    whose nodes are the sequences $\pair{x,x^2, \ldots, x^n} \in L^+$
    whose first entry is $x$, with root the one-point sequence
    $\pair{x}$, and with graph relation the restriction of
    $\Rightarrow$.  For readers familiar with tree unfoldings of
    pointed graphs, $\temptree_x$ is the tree unfolding of the graph
    $(V_\omega,\leadsto)$ at the point $x$.
    
  \item Finally, let
    \begin{equation}\label{eq:T}
      T = \set{\temptree_{x} : x\in V_{\omega}}.
    \end{equation}
  \end{enumerate}
\end{notation}

Recall the connecting maps $\powfin^n!\colon \powfin^{n+1} 1 \to
\powfin^n 1$.
\begin{lemma}\label{lemma-seqseq}
  Let $x\in V_{\omega}$.  
  \begin{enumerate}
  \item\label{lemma-seqseq:1} For all $k$ and all
    $\pair{x,x^2, \ldots, x^n}\in \temptree_x$,
    $\treepartial^{\temptree_x}_k(\pair{x,x^2, \ldots, x^n}) = x^n_k$.

  \item\label{lemma-seqseq:2} Let $R$ be a tree bisimulation on $\temptree_x$.
    If $\pair{x, x^2, \ldots, x^n}\ R\ \pair{x, y^2, \ldots, y^n}$, then
    for all $k$,
    \[
      \treepartial^{\temptree_x}_k(\pair{x, x^2, \ldots, x^n}) =
      \treepartial^{\temptree_x}_k(\pair{x, y^2, \ldots, y^n}).
    \]
    
  \item\label{lemma-seqseq:3} The tree $\temptree_x$ is strongly
    extensional and compactly branching, and
    $ \partial_\omega(\temptree_x) =
    \treepartial^{\temptree_x}_{\omega}(\pair{x}) = x$.
  \end{enumerate}
\end{lemma}
\begin{proof}
\begin{enumerate}
\item By induction on $k$.  For $k = 0$, our result is clear: the
  codomain of $\treepartial_k$ is $1$.  Assume our result for $k$, fix
  $x\in L^+$ and $\pair{x^1, \ldots, x^n}\in \temptree_x$. We first
  prove that
  \begin{equation}\label{eq:aux}
    \set{ y_k : x^n \leadsto y} = x^n_{k+1}.
  \end{equation}
  Indeed, if $x^n\leadsto y$, then
  $y_k\in x^n_{k+1}$. Conversely, if $a\in x^n_{k+1}$, we construct
  $y\in V_{\omega}$ such that $x^n\leadsto y$ with $y_k = a$.  Note
  that
  \[
    x^n_k = \powf^{k}!(x^n_{k+1}) = \pow\powf^{k-1}!(x^n_{k+1}) =
    \powf^{k-1}![x^n_{k+1}].
  \]
  Since $a\in x^n_{k+1}$,  we have $\powf^{k-1}!(a) \in x^n_k$.  So we let
  $y_{k-1} = \powf^{k-1}!(a)$.  We repeat this argument to define
  $y_{k-2}$, $\ldots$, $y_1$, $y_0$; the point is that
  $y_{k-i} \in x^n_{k-i+1}$ for $i = 0,\ldots, k$.  Choices are needed
  when we go the other way from $k$.  Note that
  \[
    \powf^{k+1}![x^n_{k+2}]
    =
    \powf(\powf^{k+1}!)(x^n_{k+2}) 
    =
    \powf^{k+2}!(x^n_{k+2}) = x^n_{k+1}.
  \]
  Every set functor preserves surjective functions, and so
  $\powf^{k+1}!$ is surjective.  Thus, there is some
  $y_{k+1} \in x^n_{k+2}$ such that $\powf^{k+1}!(y_{k+1}) = y_k$.
  The same argument enables us to find by recursion on $i$ a sequence
  $y_{k+i + 1} \in x^n_{k+i+2}$ such that
  $\powf^{k+i+1}!(y_{k+i+1 }) = y_{k+i}$.  This defines~$y$ such that
  $x^n \leadsto y$ according to~\eqref{leadsto}
  with $y_k = a$.  

  The induction step is now easy:
  \begin{align*}
    \treepartial^{\temptree_x}_{k+1}(\pair{x,x^2, \ldots, x^n})
    &= \set{\treepartial^{\temptree_x}_k(\pair{x,x^2, \ldots, x^n,y}): x^n \leadsto y}\\
    &= \set{ y_k : x^n \leadsto y} &\text{by induction hypothesis} \\
    & = x^n_{k+1} &\text{by~\eqref{eq:aux}}\rlap{\text{.}} 
  \end{align*}

\item This again is an induction on $k$, and the steps are similar to
  what we have just seen.  We also note that tuples in $\temptree_x$
  related by a tree bisimulation must have the same length.

\item Note first that by \Cref{lemma-seqseq:1} with $n = 1$, we have
  $\treepartial^{\temptree_x}_k(\pair{x}) = x_k$ for all $k$. This implies that
  $\treepartial^{\temptree_x}_{\omega}(\pair{x}) = x$.  For the strong extensionality,
  let $R$ be a tree bisimulation on $\temptree_x$.  Suppose that
  $\pair{x, x^2, \ldots, x^n}$ and $\pair{x, y^2, \ldots, y^n}$ are
  related by $R$.  Using
  \renewcommand{\itemautorefname}{items}%
  \Cref{lemma-seqseq:1} and~\ref{lemma-seqseq:2},
  \renewcommand{\itemautorefname}{item}%
  we see that for all $k$, we have $x^n_k = y^n_k$.  Thus, $x^n = y^n$.  In
  addition, since $R$ is a tree bisimulation, the parents of the two
  nodes under consideration are also related by $R$.  So the same
  argument shows that $x^{n-1}= y^{n-1}$.  Continuing in this way
  shows that $x^{n-2}= y^{n-2}$, $\ldots$, $x^2 = y^2$.  Hence
  $\pair{x, x^2, \ldots, x^n}=\pair{x, y^2, \ldots, y^n}$.
  
  Finally, we verify that $\temptree_x$ is compactly branching.  To simplify the notation, we
  shall show this for children of the root $\pair{x}$.  Suppose we have an infinite sequence
  ${\pair{x,y^0}}, {\pair{x,y^1}}, {\pair{x,y^2}}, \ldots$. We shall prove that this
    has a converging subsequence, which we will obtain by successively thinning the
    sequence $y^0, y^1, y^2, \ldots$ as shown in the rows below:
  \[
    \begin{array}{cccccccccccc}
    y(0)^0  &  y(0)^1  & \cdots &  y(0)^{k_0} &   y(0)^{k_0 +1 } & \cdots &   y(0)^{k_1} &   y(0)^{k_1+1} & \cdots 
    & y(0)^{k_2} &\cdots & y(0)^{k_3}\\
  y(1)^0  &               & \cdots &  y(1)^{1} &  & \cdots &   y(1)^{2} &     & \cdots  &  y(1)^{3} & \cdots  & y(1)^{\ell} \\
  y(2)^0 &                 &            &   y(2)^{1}        &  &            & &  & \cdots &  y(2)^2  & & y(2)^3 \\
    y(3)^0 &                 &            &   y(3)^{1}        &  &            & &  & \cdots &  y(3)^2  & &  \\
    \end{array}
  \]
  Row 1, the top row, is the sequence $y^0, y^1, y^2, \ldots$ of second components in our
  given sequence $({\pair{x,y^n}})$.  Each row is a subsequence of the row above it.  We
  also arrange some equalities in the columns.  We will have
  $y(0)^0 = y(1)^0 = y(2)^0 =\cdots$, and $y(1)^1 = y(2)^1 = y(3)^1 =\cdots$, and
  $y(2)^2 = y(3) ^2= \ldots$, etc.  To obtain row $j+1$ from row $j$, set
  $y(j+1)^i = y(j)^i $ for $0 \leq i \leq j$.  Then consider $y(j)^{j}_j$, $y(j)^{j+1}_j$,
  $y(j)^{j+2}_j$, $\ldots$.  These are elements of $V_j$.  Since $V_j$ is finite, we may
  find an element $j^* \in V_j$  and an infinite set $I$ of indices $k \geq j+1$ so that
  for $k\in I$, $y(j)^{k}_j = j^*$.  We take the remaining elements of row $j+1$ to be
  $y(j)^k$ for $k \in I$, listed in the natural way.  

Let $z_n\in V_n$ be $y(n)^n_n$.  Let us check that the sequence
  $z = (z_n)$ is compatible.  To save on notation, we simplify things and check only that
  $v_{3,2}(z_3) = z_2$.  For some $\ell > 2$, we have that $y(3)^3 $ is~$y(2)^{\ell}$.  So
\[
v_{3,2}(z_3) = v_{3,2}(y(3)^3_3) = y(3)^3_2= y(2)^{\ell}_2 =  y(2)^{2}_2 = z_2.
\]
The second equality comes from the fact that $y(3)^3$ is a compatible sequence.  The crucial
step here is that $y(2)^{\ell}_2 = y(2)^{2}_2$.  This is what our construction arranges; in
the notation from above, these are both $2^*$.

Next, we check that $\pair{x,z}$ is a child of $\pair{x}$.  For all $n$, $z_n$ is
  $y^k_n$ for some $k$.  Thus, $x\leadsto z$, and so $\pair{x,z}$ is a child of $\pair{x}$.
  We conclude our proof by showing that for the subsequence $(y(n)^n)$ of
  $(y^n)$, we have $\lim \pair{x, y(n)^n} = \pair{x,z}$: for every $k$ and every $n \geq k$,
  \[
    \treepartial^{\temptree_x}_k(\pair{x,y(n)^n})
    = y(n)^n_k
    = v_{n,k}(y(n)^n_n)
    = v_{n,k}(z_n)
    = z_k
    = \treepartial^{\temptree_x}_k(\pair{x,z}),
  \]
  where we use \Cref{lemma-seqseq:1} in the first and last step of this
  calculation.
\qedhere
\end{enumerate}
\end{proof}

In the next result, recall the set $T$ from~\eqref{eq:T}.  We identify this set in~\Cref{bijection}.

\begin{lemma}\label{forBarr}
For every tree $t$ there is a Barr-equivalent tree
    $t^*\in T$ such that $t^*$ is strongly extensional and compactly
    branching.
\end{lemma}
\begin{proof}
  Given any tree $t$, we have $x = \partial_\omega(t)\in V_{\omega}$.
  For all $n$, $x_n = \partial_n(t)$.  The tree $t^* = \temptree_x$ in
  \Cref{lemma-seqseq}\ref{lemma-seqseq:3} is strongly extensional
  and compactly branching.  Recall that the root of $t^*$ is
  $\pair{x}$.  By \Cref{lemma-seqseq}\ref{lemma-seqseq:1}, we have
  that for all $n<\omega$,
  \[
    \partial_n(t^*)
    =
    \treepartial^{t^*}_n(\root(\temptree_x))
    =
    \treepartial^{t^*}_n(\pair{x})
    =
    x_n
    =
    \partial_n(t).\tag*{\qedhere}
  \]
\end{proof}

\begin{proposition}\label{P:BB}
  Bisimilar trees are Barr-equivalent.
\end{proposition}

\begin{proof}
Let $R$ be a tree bisimulation between trees $t$ and $u$.
We show by induction on $n$ that whenever $x\in t$ and $y\in u$ are related by $R$,
then $\rho^t_n(x) = \rho^t_n(y)$.
This is clear for $n = 0$.  Assuming it for $n = 1$, we have 
\[
\rho^t_{n+1}(x) = \set{\rho^t_{n}(x'): x' \mbox{ is a child of $x$}} =  \set{\rho^u_{n}(y'): y' \mbox{ is a child of $y$}}  = \rho^t_{n+1}(y).
\]
For equality in the middle, let $x'$ be a child of $x$ in $t$.
Since $R$ is a bisimulation, there is some child $y'$ of $y$ in $u$ such that $x'$ and $y'$ are related by $R$.
By induction hypothesis, $\rho^t_{n}(x') = \rho^u_{n}(y')$.  This implies that the set on the left is a subset of the set on the right.
The converse is similar.

Now $\partial_n(t) = \rho^t_n(\root(t))$, and similarly for $u$.  So since $R$ relates the roots of $t$ and $u$, 
$\partial_n(t) = \partial_n(u)$.
This holds for all $n$, and we are done.
\end{proof}

\begin{lemma}\label{bijection}
  The set $T$ defined in~\eqref{eq:T} is the set of all compactly
  branching, strongly extensional trees.  
The map $\temptree\colon  V_{\omega}\to T $ is a bijection,
where $\temptree(x) = \temptree_x$, and 
  $\partial_n\o \temptree^{-1}$ is a limit of the first $\omega$ steps in the terminal-coalgebra chain~\eqref{eq:op-chain}.
\end{lemma}
\begin{proof}
  By \Cref{lemma-seqseq}\ref{lemma-seqseq:3} we know that every tree in $T$ is strongly
  extensional and compactly branching. For the reverse inclusion, let $t$ be compactly
  branching and strongly extensional.  Let $t^*$ be as in \Cref{forBarr} for $t$. By
  \Cref{lem:w1,seqK-part}, $t = t^*$. Thus, we have $t\in T$.  In the second assertion, we
  need only show that if $\temptree_x = \temptree_y$, then $x = y$.  But from the proof
  of~\cref{forBarr},
  $x = \partial_{\omega}(\temptree_x) =\partial_{\omega}(\temptree_y) = y$.  Thus,
    $\temptree\colon V_{\omega}\to T $ is a bijection, so $\partial_n\o \temptree^{-1}$ is a
    limit of~\eqref{eq:op-chain}.
\end{proof}

\begin{defn}\label{D:D}
  Let $D$ be the set of finitely branching strongly extensional trees.
  Let $\delta\colon D\to \powf D$ take a strongly extensional tree $t$
  to the (finite) set of its subtrees $t_x$.
\end{defn}

\noindent
Here we use
\Cref{rem:strongext}\ref{rem:strongext:5}: a subtree of a strongly
extensional tree is strongly extensional.

\removeThmBraces
\begin{lemma}[\cite{worrell:05}] \label{L:Worrell}
The $\powf$-coalgebra $(D,\delta)$ is terminal.
\end{lemma}
\resetCurThmBraces
\begin{proof}

\takeout{
The map $\phi\colon V_{\omega} \to T$ given by
  $\phi(x) = \temptree_x$ is obviously surjective.  Suppose that
  $\temptree_x = \temptree_y$.  The roots of these trees are
  $\pair{x}$ and $\pair{y}$.  For all $n$, we have that 
  \[
    x_n = \treepartial^{\temptree_x}_n(\pair{x}) = \treepartial^{\temptree_y}_n(\pair{y}) = y_n.
  \]
  Thus, $\partial_\omega(\pair{x}) = \partial_\omega(\pair{y})$. By \Cref{lem:w1,seqK-part},
  $x = y$.  So $\phi$ is injective. The formula for~$\partial_n$ comes from
  \Cref{lemma-seqseq}\ref{lemma-seqseq:1}.
}

We use \Cref{P:Kripke}. The map $m\colon V_{\omega + 1} \to V_\omega$ in~\eqref{diag:m}
  assigns to a finite set of trees in~$V_\omega$ their tree-tupling. Its image is the set of
  all strongly extensional, compactly branching trees which are finitely branching at the
  root. An easy induction on $n$ shows that $V_{\omega+n}$ is the set of all compactly
  branching, strongly extensional trees $t$ with the property that the topmost $n$ levels of
  $t$ are finitely branching.  With this description, $V_{\omega+n} \subseteq D$, and the
  limit~$V_{\omega+\omega}$ is simply the intersection $D = \bigcap_n V_{\omega+n}$.  This
  shows that the carrier set of $\nu \powf$ is $D$.  For the structure map $\delta$, note that
  $m\colon \powf V_{\omega} \to V_{\omega}$ in \eqref{diag:m} is tree-tupling, as are the maps
  $\powf m$, $\powf \powf m$, etc.  It follows that in the intersection, $D$, the coalgebra
  structure is the inverse of tree-tupling.\qedhere
\end{proof}
\begin{theorem}\label{T:s-ext}
  \begin{enumerate}
  \item The limit $V_\omega$ in the terminal-coalgebra chain for
    $\powfin$ consists of all compactly branching, strongly
    extensional trees.
  \item\label{T:s-ext:2} The terminal coalgebra for $\powfin$ is the coalgebra of all
    strongly extensional, finitely branching trees. The coalgebra
    structure is the inverse of tree tupling.
  \end{enumerate}
\end{theorem}

\noindent
The first part follows from~\cref{L:Worrell} and~\cref{bijection}, and the second
from~\cref{L:Worrell} and~\cref{D:D}.

\subsection{Compactness and the Worrell pseudo-metric}\label{S:compactness}
We can also explicate the termino\-logy of `compactly branching'
(\Cref{dcb}).  We have formulated this using convergence in a
formal sense, and here we connect this to a metric.  The collection
$\mathcal{C}$ of all strongly extensional trees is a proper class.
Setting this aside, it also comes with \emph{Worrell's pseudo-metric}~$d$,
 given by
\begin{equation}
\label{psW}
  d(t,u) = \inf \set{ 2^{-n} : \partial_n(t) =  \partial_n(u)}.
\end{equation}

\begin{proposition} The relation $\approx$ on the class $\TT$ of all
  trees has $2^{\aleph_0}$ equivalence classes, and $d$ is a metric on
  $\TT/{\approx}$.
\end{proposition}

\begin{proof}
  We have a well-defined injective map
  \[
    b\colon \TT/{\approx} \to V_\omega
  \]
  assigning to an equivalence class of a tree $t$ the sequence
  \(
  (\partial_n t)_{n < \omega}$. Hence,
  $|\TT/{\approx}| \leq 2^{\aleph_0}.
  \)
  We check the reverse
  inequality.  For every set $A\subseteq \Nat \setminus \set{0}$, let
  $t_A$ be an infinite path with additional leaves of depth $n$ for
  each $n\in A$.  Then one can check that for $A \neq B$, $t_A$ and
  $t_B$ are not Barr equivalent. Hence
  $|\TT/{\approx}| \geq 2^{\aleph_0}$ since the set
  $\Nat \setminus \set{0}$ has $2^{\aleph_0}$ subsets.

The general construction of a metric space from a pseudo-metric space 
takes equivalence classes of points of distance $0$.  In our setting, points (trees)
have distance $0$ in $\TT$ exactly when they are Barr-equivalent.
\end{proof}

It is a standard fact that for metric spaces, compactness and
sequential compactness are equivalent.  
 Thus, for all
trees $t$ and all nodes $x$ of $t$,
$\set{t_y : \text{$y$ a child of $x$ in $t$}}$ is compact iff it is sequentially
compact.  This explains why we were able to define the compactly
branching trees using sequential compactness.

\begin{theorem}\label{T:compact}
  The metric space $(\TT/{\approx},d)$ is compact.
\end{theorem}
\begin{proof}
\begin{enumerate}
\item The space $\TT/{\approx}$ is totally bounded: for every
  $\eps > 0$, there is a finite cover by $\eps$-balls. Indeed,
  choose~$n$ with $2^{-n} < \eps$ and take the (finite) set~$A$ of all
  extensional trees of height at most~$n$. Then every tree $t$
  satisfies $\partial_n t \in A$ and
  $d(t, \partial_n t) \leq 2^{-n} < \eps$.

\item We recall that a metric space is compact iff it is totally
  bounded and complete. 
Towards showing that $\TT/{\approx}$ is
  complete, let $([t_n])$ be a Cauchy sequence.  For each fixed~$k$, the
  sequence $(\partial_k(t_n))$ is eventually constant, and it does not depend on
  the choice of representatives from $[t_n]$.  Let
  $x_k\in V_k$ be the unique element of $V_k$ such that for all but finitely many~$n$,
  $\partial_k(t_n) =x_k$.  The sequence $(x_k)$ is
  compatible: 
  For all but finitely many $n$, we have
  \[
    x_k = \partial_k(t_n) = v_{k+1,k}(\partial_{k+1}(t_n)) = v_{k+1,k}(x_{k+1}).
  \]
  Then \Cref{forBarr} provides us with a single tree~$t$ such that
  $\partial_k(t) = x_k$ for all $k$.  The equivalence class $[t]$ is a limit of the
  original sequence~$([t_n])$.
  \qedhere
\end{enumerate}
\end{proof}
\begin{corollary}\label{corWor}
  A set $M \subseteq \TT$ is compact iff for every sequence $(t_n)$
  from $M$, if $\lim_n t_n = s$, there exists $t\in M$ such that
  $t\approx s$.
\end{corollary}
\begin{proof}
  The quotient map $q\colon \TT\to \TT/{\approx}$ is continuous and
  closed, and thus $M$ is compact iff its image $q[M]$ is.  The latter
  holds iff $q[M]$ is a closed set.  This is what the condition in our
  corollary states.
\end{proof}

To conclude, the reader might wonder why this section did not simply mention \Cref{corWor}.
Why was it necessary to consider the graph $L^{+}$ of all finite sequences of trees, and the
associated set $T$ from \eqref{eq:T}?  The answer is that the linchpin of this section is
\Cref{forBarr}, and we were not able to prove it more directly.  In addition, we aimed for
the characterization in~\Cref{bijection}.  For this, the development using $T$ turned out to
be an elegant way to proceed, reminiscent of the way one uses maximal objects of various
types (such as ultrafilters) in proving representation results.

\subsection{The Cofree Comonad on $\powf$}
\label{S:cofree-comonad}
We have introduced the topic of cofree comonads in
\Cref{S:introduce-cofree-comonads}.  We also mentioned
in~\Cref{C:cofree} that for a category with products, an
endofunctor $F$ generates a cofree comonad exactly when for each
object $Y$, a terminal coalgebra for $F(-)\times Y$ exists.  We
illustrate this by describing the terminal coalgebra for
$\powf(-)\times A$ for all sets $A$.

Fix a set $A$ for the remainder of this discussion.

\begin{defn}
\begin{enumerate}
\item An \emph{$A$-tree} is a pair $(t,\ell)$ consisting of a tree $t$
  together with a \emph{labelling} $\labelling\colon t\to A$ of the
  nodes of $t$.  We use $\TT$ for the class of $A$-trees.

\item We again use the notation $t_x$ for the sub-$A$-tree of $t$ rooted
  in the node $x$ of $t$.  The labelling of~$t_x$ is inherited from
  $t$.

\item For most concepts, we obtain an `$A$-version' by appending $A$
  everywhere in the obvious way.  So we obtain the concepts of a
  \emph{isomorphic}, \emph{extensional}, \emph{Barr-equivalent},
  \emph{finitely branching}, and \emph{compactly branching} $A$-trees.

\item A \emph{graph bisimulation} between $A$-trees $(t,\ell)$ and
  $(u,m)$ is a graph bisimulation $R$ of the trees $t$ and $u$ which
  in addition has the property that if $x\, R\, y$, then
  $\ell(x) = m(y)$.  In words, nodes which are related have the same
  labelling.  This gives us the concepts of a \emph{tree bisimulation}
  and a \emph{strongly extensional} $A$-tree.

\item We write $F$ for the endofunctor $FX = \powf X \times A$.  As
  before, $F^n 1$ is written $V_n$.  The maps $\partial_n:\TT\to V_n$
  are given by
  \[
    \partial_{n+1}(t, \ell)
    =
    \big(\set{\partial_n(t_x): \text{$x$ is a child of $\root(t)$}}, \ell(\root(t))\big).
  \]
  The maps $\rho^t_n \colon t \to V_n$ are defined similarly.

\item We again take $V_{\omega}$ to be the set of compatible sequences
  $(x_n)$ of $A$-trees.  For all $1 \leq m \leq n$ and all $x\in V_n$,
  $\ell(x) = \ell(F^{n-m}! (x))$.  This allows us to define a map
  $\ell\colon V_{\omega} \to A$: we put $\ell((x_n)) = \ell(x_n)$ for
  any $n \geq 1$.

  For every $x\in V_{\omega}$, we have an $A$-tree $\temptree_x$ with the labelling defined by
  \[
    \ell(\pair{x, x^2, \ldots, x^n}) = \ell(x^n).
  \]

  The set $T$ is defined as before in~\eqref{eq:T}, and $D$ is the
  set of all finitely branching strongly extensional~$A$-trees.  We have a
  map $\delta\colon D \to \powf D \times A$ taking a strongly
  extensional $A$-tree~$t$ to the set of its subtrees along with the
  label of the root of $t$.  This map is a bijection.
\end{enumerate}
\end{defn}

With these substitutions, one can read all of \Cref{S:app-worrell} again.
The proofs in that section are nearly the same as what we have seen, and the points we make above are all that needs to be changed. We thus obtain the following result:

\removeThmBraces
\begin{theorem} 
  Let $A$ be a set, and consider $FX =  \powf X \times A$:
  \begin{enumerate}
  \item the maps $\partial_n\colon T \to F^n 1$ given by $\partial_n(\temptree_x) = x_n$ form a limit of~\eqref{eq:op-chain}; thus, $V_\omega \cong T$,
    
  \item\label{L:Worrell:3} the coalgebra $(D,\delta)$ is terminal.
  \end{enumerate}
\end{theorem}
\resetCurThmBraces

\begin{corollary}
  The cofree comonad on $P_f$ assigns to every set $A$ the coalgebra
  of all finitely branching, strongly extensional $A$-trees. The
  coalgebra structure is given by the inverse of tree tupling, and the
  counit assigns to a tree the label of its root.
\end{corollary}

\section{Vietoris Polynomial Functors}
\label{S:Vietoris}

We continue the theme of terminal coalgebras of functors in an inductively defined class.
We move from $\Set$ to $\Top$, the category of topological spaces. We
  exchange~$\powfin$ for the Vietoris functor $\V$, defined below.  We are not able to prove
  that for Vietoris polynomial functors on $\Top$, the terminal-coalgebra chain converges in
  $\omega$ steps.  However, we are able to do this for three full subcategories of $\Top$. The
  first is $\Haus$, the category of Hausdorff topological spaces.  Hofmann et
  al.~\cite{HNN19} proved that the terminal-coalgebra chains of Vietoris polynomial functors
  on $\Haus$ converge in $\omega$ steps. Our proof is slightly different from theirs because
  we wish to avoid a result stated by Zenor~\cite{zenor70} whose proof is incomplete.
  Second, we shall consider the categories $\KHaus$ of compact Hausdorff spaces and
  $\Stone$ of Stone spaces (also defined below).

Recall that a \emph{base} of a topology is a collection $\mathcal{B}$ of open sets such that
every open set is a union of members of $\mathcal{B}$.  A \emph{subbase} is a collection of
open sets whose finite intersections form a base.  For every collection $\mathcal{B}$ of
subsets of the space, there is a smallest topology for which $\mathcal{B}$ is a (sub)base, the
family of unions of finite intersections from $\mathcal{B}$.
\begin{defn}\label{D:Vietoris}
  \begin{enumerate}
  \item Let $X$ be a topological space.  We denote by $\V X$  the space of
    compact subsets of $X$ equipped with the `hit-and-miss'
    topology.  This topology has as a subbase all sets of the following forms:
    \begin{equation}\label{eq:predlift}
      \begin{aligned}[c]
        U^{\Diamond}
        & = \set{R \in VX : R \cap U\neq \emptyset}
        & \text{($R$ hits $U$)}, \\
        U^{\Box}
        &= \set{R \in VX : R \subseteq U}
        & \text{($R$ misses $X \setminus U$),}
      \end{aligned}
    \end{equation}
    where $U$ ranges over the open sets of $X$.    We call $\V X$ the \emph{Vietoris space of
      $X$}, also known as the \emph{hyperspace} of $X$.
    
  \item Recalling that the image of a compact set under a continuous
    function is compact, for a continuous function
    $f\colon X \to Y$, we put $\V f (A) = f[A]$ for every compact
    subset $A$ of~$X$.
  \end{enumerate}
\end{defn}
\begin{remark}\label{R:Vietoris}
  \begin{enumerate}
  \item For a compact Hausdorff space $X$, Vietoris~\cite{Vietoris22} defined $\V X$ to consist
    of all \emph{closed} subsets of $X$.  These are the same as the compact subsets in this
    case. In the coalgebraic literature, $\V X$ has also mostly been studied for spaces $X$
    which are compact Hausdorff. However, the `classic Vietoris space' (using closed subsets)
    does not yield a functor on $\Top$ (see Hofmann et al.~\cite[Rem.~2.28]{HNN19}). Hofmann et
    al.~\cite[Def.~2.27]{HNN19} call the functor $\V$ in \Cref{D:Vietoris} the \emph{compact
      Vietoris functor}.

  \item\label{R:Vietoris:2} Michael~\cite[Thm.~4.9.8]{Michael51}
    proved that $X$ is Hausdorff iff so is $\V X$.
         
  \item\label{R:Vietoris:3} Vietoris~\cite{Vietoris22}
    originally proved that for a compact Hausdorff space $X$ (the
    classic Vietoris space)~$\V X$ is compact Hausdorff, too. 

  \item\label{R:Vietoris:4} A \emph{Stone space} is a compact Hausdorff space having a base of
    clopen sets.  If $X$ is a Stone space, so is $\V X$; see~\cite[Thm.~4.9.9]{Michael51}
    or~\cite[Section~III.4]{Johnstone86}.
  \end{enumerate}
\end{remark}
\begin{proposition}\label{prVnat}
  For every continuous function $f\colon X\to Y$ and every open
  $U\subseteq Y$, 
  \[
    (f^{-1}(U))^{\Diamond} = (\V f)^{-1}(U^{\Diamond})
    \qquad\text{and}\qquad
    (f^{-1}(U))^{\Box} = (\V f)^{-1}(U^{\Box}).
  \]
\end{proposition}
\begin{proof}
  Let $R\in \V X$.   Observe that
  \[
    R \cap f^{-1}(U) \neq \emptyset
    \iff 
    f[R] \cap U \neq \emptyset
    \iff
    f[R]\in U^{\Diamond}
    \iff
    R\in (\V f)^{-1}( U^{\Diamond}).
  \]
  This proves our first assertion for all $R$. For the second assertion, we have
  \[
    R \subseteq f^{-1}(U)
    \iff 
    f[R] \subseteq U
    \iff
    f[R]\in U^{\Box}
    \iff
    R\in (\V f)^{-1}(U^{\Box}).
    \tag*{\qedhere}
  \]
\end{proof}
\begin{corollary}\label{V-functor}
  The mappings $X \mapsto \V X$ and $f \mapsto \V f$ form a
 functor~$\V$ on~$\Top$.
\end{corollary}
\noindent
Indeed, \Cref{prVnat} shows that
for every subbasic open set of $\V Y$ its inverse image under~$\V f$
is open in $\V X$.  This establishes continuity of $\V f$. 
\begin{notation}\label{N:V}
  We denote by $\Haus$, $\Khaus$ and $\Stone$ the full subcategories
  of $\Top$ given by all Hausdorff spaces, all compact Hausdorff
  spaces and all Stone spaces, respectively. By
  \Cref{R:Vietoris}\mbox{\ref{R:Vietoris:2}--\ref{R:Vietoris:4}},
  the functor $\V$ restricts to these three full subcategories, and we
  denote the restrictions by $\V$ as well.
\end{notation}
\begin{remark}\label{R:Vietoris-tech}
  \begin{enumerate}     
  \item\label{R:Vietoris-tech:2} The full subcategories $\Haus$,
    $\KHaus$, and $\Stone$ are closed under limits in~$\Top$.
    In particular, the inclusion functors preserve and reflect
    limits. In fact, $\KHaus$ is a full reflective
    subcategory: the reflection of a space is its Stone-\v{C}ech
    compactification.

  \item\label{R:Vietoris-tech:3} If an $\omega^\opp$-chain
    as in~\eqref{Xsequence} consists of surjective continuous maps
    between compact Hausdorff spaces, then each limit projection 
    $\ell_n\colon \lim_{k<\omega} X_k \to X_n$ is surjective, too.  
    Moreover, Eilenberg and Steenrod~\cite[Cor.~3.9]{ES:1952} prove
    the surjectivity of projections for all codirected limits of
    surjections between compact Hausdorff spaces; see also Ribes and
    Zalesskii~\cite[Prop.~1.1.10]{RibesZ10}).

  \item\label{R:Vietoris-tech:4} If $X$ has a base
    $\mathcal B$ which is closed under finite unions, then the sets
    $U^\Diamond$ and $U^\Box$ for $U \in \mathcal B$ already form a
    subbase of $\V X$. Indeed, given a set $\mathcal S$ of open subsets of
    $X$ we have
    $(\bigcup \mathcal S)^\Diamond = \bigcup \set{ U^\Diamond : U \in
      \mathcal S}$. Moreover, it is easy to see that
    \[\textstyle
      (\bigcup \mathcal S)^\Box = \bigcup \set{\big( \bigcup \mathcal
        F\big)^\Box : \text{$\mathcal F \subseteq \mathcal S$ finite} };
    \]
    `$\supseteq$' is trivial, and for `$\subseteq$' use
    compactness of $R \in \V X$.  Hence, if $\mathcal S$ consists of
    basic open sets from~$\mathcal B$, then $\bigcup F \in \mathcal B$ due to
    its closure under finite unions. Thus, $(\bigcup \mathcal S)^\Box$
    is a union of sets of the form $U^\Box$ for $U \in \mathcal B$.
  \end{enumerate}
\end{remark}
%
\begin{proposition}
  \label{P:limit-Haus}
  The functor $\V\colon \Haus \to \Haus$ preserves limits of
  $\omega^{\opp}$-chains.
\end{proposition}
%
\begin{proof}
  Consider an $\omega^\opp$-chain as in~\eqref{Xsequence}. Let
  $M = \lim \V X_n$, with limit cone $r_n \colon M \to \V X_n$.  Let
  $m\colon \V L \to M$ be the unique continuous map such that
  $\V \ell_n = r_n \o m$ for all $n < \omega$. We shall prove that $m$
  is a bijection and then that its inverse is continuous, which proves
  that~$m$ is an isomorphism.
  \begin{enumerate}
  \item\label{P:limit-Haus:1} Injectivity of $m$ follows from the
    fact that $\V\ell_n$ ($n < \omega$) forms a jointly monic family,
    as we will now prove. Suppose that $A, B \in \V L$ satisfy
    $\ell_n[A ] = \ell_n[B]$ for every $n < \omega$. We prove that $A
    \subseteq B$; by symmetry $A = B$ follows. Given $a \in A$, we show that every
    open neighbourhood of~$a$ has a nonempty intersection with
    $B$. Since $B$ is closed, we then have $a \in B$.
    It suffices to prove the desired property for the basic open
    neighbourhoods $\ell_n^{-1}(U)$ of~$a$, for~$U$ open in $X_n$ (\Cref{P:chain}\ref{R:Vietoris-tech:1}). Since
    $\ell_n[A] = \ell_n[B]$ we have some $b \in B$ which satisfies $\ell_n(a) =
    \ell_n(b)$. Then we have $b \in \ell_n^{-1}(U) \cap B$. 

  \item Surjectivity of $m$.  An element of $M$ is a sequence $(K_n)_{n < \omega}$ of
    compact (hence closed) subsets $K_n \subseteq X_n$ such that $f_n[K_{n+1}] = K_n$ for
    every $n < \omega$. We need to find a compact set $K \subseteq L$ such that
    $\ell_n[K] = K_n$ for every $n < \omega$.  With the subspace topology, $K_n$ is itself a
    compact space.  The connecting maps $f_n\colon X_{n+1} \to X_n$ restrict to surjective
    continuous maps $K_{n+1} \epito K_n$.  Thus, the spaces $K_n$ form an
    $\omega^\opp$-chain of surjections in $\KHaus$. Let $K$ be the limit with projections
    $p_n \colon K \epito K_n$; the maps $p_n$ are surjective by
    \Cref{R:Vietoris-tech}\ref{R:Vietoris-tech:3}.\footnote{Note that here we use that we
      are in $\Haus$; in $\Top$ we would only know that the spaces $K_n$ are compact (but
      not that they are Hausdorff), and so we could not apply
      \Cref{R:Vietoris-tech}\ref{R:Vietoris-tech:3}.}
    Then $K$ is a subset of $L$, and each
    projection $p_n$ is the restriction of $\ell_n\colon L \to X_n$.

    Let us check that the topology on $K$ is the subspace topology
    inherited from $L$. A base of the topology on $K$ is the
    family of sets $p_n^{-1}(U)$ as $U$ ranges over the open subset
    of $K_n$.  Each $U$ is of the form $V\cap K_n$ for some open
    $V$ of $X_n$, and $p_n^{-1}(U) = \ell_n^{-1}(V)\cap K$. Thus,
    $p_n^{-1}(U)$ is open in the subspace topology, and the converse
    holds as well.

    Moreover, $K$ is a compact space by \Cref{R:Vietoris-tech}\ref{R:Vietoris-tech:2}.
    Thus, $K$ is the desired compact set in $\V L$ such that $p_n[K] = K_n$ for all $n$.
    
  \item Finally, we prove that the inverse $k\colon M \to \V L$, say,
    of $m$ is continuous. We know that the sets~$\ell_n^{-1}(U)$, for
    $U$ open in $X_n$, form a base of $L$. Moreover, this base is
    closed under finite unions. By
    \Cref{R:Vietoris-tech}\ref{R:Vietoris-tech:4} and using
    \Cref{prVnat}, we obtain that $\V L$ has a subbase given by the following
    sets
    \[
      (\V \ell_n)^{-1}(U^\Diamond) = (\ell_n^{-1}(U))^\Diamond
      \quad
      \text{and}
      \quad
      (\V \ell_n)^{-1}(U^\Box) = (\ell_n^{-1}(U))^\Box
      \quad
      \text{for $U$ open in $X_n$.}
    \]
    It suffices to show that the inverse images of these subbasic open
    sets of $\V L$ are open in $M$. For $\V\ell_n^{-1} (U^\Diamond)$
    with $U$ open in $X_n$ we use that $\V \ell_n \cdot k = r_n$
    clearly holds to obtain
    \[
      k^{-1}\big(\V\ell_n^{-1}(U^\Diamond)\big) = r_n^{-1}(U^\Diamond),
    \]
    which is a basic open set of $M$ by
    \Cref{P:chain}\ref{R:Vietoris-tech:1}. For the subbasic
    open sets $\V\ell_n^{-1}(U^\Box)$, the proof is similar.
    \qedhere
  \end{enumerate}
\end{proof}
\begin{corollary}\label{C:limit-KHaus}
  The restrictions of $\V$ to $\KHaus$ and $\Stone$ preserve limits of
  $\omega^\opp$-chains. 
\end{corollary}
\noindent
Indeed, use \Cref{R:Vietoris-tech}\ref{R:Vietoris-tech:2}.
\begin{remark}\label{r:codirected}
  A \emph{codirected limit} is the limit of a diagram whose scheme is
  of the form $P^\opp$ for a directed poset $P$.
  \Cref{P:limit-Haus} and \Cref{C:limit-KHaus} hold more
  generally for codirected limits.  The argument is the same.  This
  proves a result stated by Zenor~\cite{zenor70}, but with an
  incomplete proof.
\end{remark}

The following definition is due to Kupke et al.~\cite{kkv} for Stone
spaces, whereas Hofmann et al.~\mbox{\cite[Def.~2.29]{HNN19}} use general
topological spaces, but they later essentially restrict constants to be
(compact) Hausdorff, stably compact or spectral spaces.
%
\begin{defn}
  \label{D:Vietoris-poly}
  The \emph{Vietoris polynomial functors} are the endofunctors on
  $\Top$ built from the Vietoris functor $\V$, the constant functors,
  and the identity functor, using product, coproduct, and composition.
  Thus, the Vietoris polynomial functors are built according to the
  following grammar
  \[
    \textstyle
    F ::= \V \mid A \mid \Id \mid \prod_{i\in I} F_i \mid
    \coprod_{i\in I} F_i \mid FF,
  \]
  where $A$ ranges over all topological spaces and $I$ is an index set.
\end{defn} 
\begin{theorem}\label{T:Vietoris}
  Let $F\colon \Top\to \Top$ be a Vietoris polynomial functor, and
  assume that all constants in $F$ are Hausdorff spaces.  Then the terminal-coalgebra chain for
  $F$ converges in~$\omega$ steps, and
  $\nu F = V_\omega$ is a Hausdorff space.
\end{theorem}
\begin{proof}
  An easy induction on Vietoris polynomial functors $F$  
  shows that:
  \begin{enumerate}
  \item The functor $F$ has a restriction $F_0\colon \Haus \to \Haus$.


  \item  The functor $F_0$ preserves limits of
    $\omega^\opp$-chains; the most important step is done in
    \Cref{P:limit-Haus}.
  \end{enumerate}
  The terminal coalgebra result for $F_0$ follows from the fact which
  we have mentioned in
  \Cref{S:preliminaries}: $\nu F$ is
  the limit of the terminal-coalgebra $\omega^\opp$-chain $V^{F_0}_n$
  ($n < \omega$).  Since $\Haus$ is closed under limits in~$\Top$ and
$V^{F_0}_n = V^{F}_n$, the functor $F$ has the same terminal coalgebra
  $\nu F = \lim F^n 1$.
\end{proof}
\begin{corollary}\label{C:initial}
  Let $F\colon \Top\to \Top$ be a Vietoris polynomial functor, and
  assume that all constants in $F$ are Hausdorff spaces. Then $F$ has
  an initial algebra.
\end{corollary}
\noindent
This follows from \Cref{T:Vietoris},
\Cref{E:cms}\ref{E:cms:set} and \Cref{T:initial}, since an easy
induction shows that $F$ preserves monomorphisms.
\begin{corollary}\label{C:KHaus}
  Let $F\colon \Top\to \Top$ be a Vietoris polynomial functor in which
  all constants are \emph{compact} Hausdorff spaces and only finite
  coproducts are used. Then the terminal coalgebra $\nu F$ is a
  compact Hausdorff space.
\end{corollary}
\begin{proof}
  The functor $F$ restricts to an endofunctor on $\KHaus$, since products and finite
    coproducts of compact spaces are compact.
  Thus, the
  terminal-coalgebra $\omega^\opp$-chain~$F^n 1$ lies in
  $\KHaus$. Moreover, $\KHaus$ is closed under limits in $\Top$
  because it is a full reflective subcategory
  (\Cref{R:Vietoris-tech}\ref{R:Vietoris-tech:2}). Thus,
  $\nu F = \lim_{n < \omega} F^n 1$ is compact Hausdorff.
\end{proof}

\begin{corollary}\label{C:Stone}
  Let $F\colon \Top\to \Top$ be a Vietoris polynomial functor
  in which all constants are Stone spaces and only finite
  coproducts are used. Then the terminal coalgebra $\nu F$ is a Stone
  space.
\end{corollary}
\noindent
The proof is similar.

\begin{corollary}\label{cor-V-cofree}
  For a Vietoris polynomial functor $F\colon \Top\to \Top$ in which all constants are
  Hausdorff spaces, a cofree comonad is obtained in~$\omega$ steps.
\end{corollary}
\noindent
Indeed, for every Hausdorff space $A$, the functor $F(-) \times A$ is
also a Vietoris polynomial functor in which all constants are
Hausdorff spaces. Now apply \Cref{C:cofree} and
\Cref{T:Vietoris}.

\begin{remark}
  \renewcommand{\itemautorefname}{Item} \Cref{C:KHaus} essentially
  appears in work by Hofmann et al.~\cite[Thm.~3.42]{HNN19} (except
  for the convergence ordinal). \Cref{C:Stone} is due to Kupke et
  al.~\cite{kkv}.  Our proof using convergence of the
  terminal-coalgebra chain is different than the previous ones.
  \renewcommand{\itemautorefname}{item}%
\end{remark}

\begin{example}\label{ex:Vnu}
  The terminal coalgebra for $\V$ itself was identified by
  Abramsky~\cite{abramsky}.  By what we have shown, it is
  $V_{\omega} = \lim \V^n 1$. An easy induction on $n$ shows that
  $\V^n 1$ is $\powf^n 1$ with the discrete topology; the key point is
  that each set $\powf^n 1$ is finite.  We know that $V_{\omega}$ is
  thus the set $T$ in~\eqref{eq:T} of all compactly branching,
  strongly extensional trees.  The topology was described in
  \Cref{P:chain}\ref{R:Vietoris-tech:1}: it has as a base the sets
  $\partial_n^{-1}(U)$ as $U$ ranges over the subsets of $\powf^n 1$.
  By \Cref{T:Vietoris}, $\nu F$ is a Stone space.
\end{example}
\begin{remark}
  Note that \Cref{T:Vietoris} also holds for Vietoris polynomial
  functors when we take $\Haus$ as our base category.
  Hofmann et al.~\cite{HNN19} consider other full subcategories of
  $\Top$, and they
  also study the completeness of the category of coalgebras for Vietoris polynomial
  functors $F$ (however, they restrict to using finite products and
  finite coproducts in their definition of Vietoris polynomial
  functors). For a Vietoris polynomial functor $F$ on $\Haus$, the category of
  coalgebras is complete~\mbox{\cite[Cor.~3.41]{HNN19}}. Moreover,
  every subfunctor of $F$ has a terminal
  coalgebra~\mbox{\cite[Cor.~4.6]{HNN19}}.
\end{remark}
\begin{remark}
  Hofmann et al.~\cite[Ex.~2.27(2)]{HNN19} also consider a related
  construction called the \emph{lower Vietoris space} of $X$. It is
  the set of all closed subsets of $X$ with the topology generated by
  all sets $U^\Diamond$, cf.~\eqref{eq:predlift}. This again yields a
  functor on $\Top$: a given continuous function is mapped to
  $A \mapsto \ol{f[A]}$, where $\ol{f[A]}$ denotes the closure
  of~$f[A]$. Furthermore, one has a corresponding notion of lower
  Vietoris polynomial functors. They prove that for such functors $F$
  on the category of stably compact spaces
  (defined in op.~cit.), $\Coalg F$ is
  complete~\cite[Thm.~3.35]{HNN19}. Furthermore, if a lower Vietoris
  polynomial functor $F$ on $\Top$ can be restricted to that category,
  then its terminal-coalgebra chain converges in $\omega$ steps:
  $\nu F = V_\omega$~\cite[Thm.~3.36]{HNN19}. Similar results hold for
  spectral spaces and spectral maps.
\end{remark}
\begin{remark}\label{R:notAcc}
  Let us mention a very general result which applies in many
  situations to deliver a terminal coalgebra: Makkai and Par\'e's
  Limit Theorem~\cite[Thm.~5.1.6]{makkaipare}.  It implies that every
  accessible endofunctor $F\colon\A\to\A$ on a locally presentable
  category has an initial algebra and a terminal coalgebra.  (Indeed,
  the theorem implies that the category of $F$-coalgebras is
  cocomplete.)  This result cannot be used here because $\Haus$ is not
  locally presentable:  it does not have
  a small set of objects that is colimit-dense~\cite[Prop.~8.2]{ahrt23}.
\end{remark}
\begin{oproblem}
  \begin{enumerate}
  \item\label{OP:1} Does every Vietoris polynomial functor on the category~$\Top$ have a
    terminal coalgebra?
  \item\label{OP:2} Does every Vietoris polynomial functor on $\KHaus$
    in which all constants are compact Hausdorff and only finite coproducts are used 
    have an initial algebra?
  \end{enumerate}
\end{oproblem}

\noindent
\renewcommand{\itemautorefname}{Item}%
\Cref{OP:1} above is equivalent to asking whether the result that $\nu F$ exists
for every Vietoris polynomial functor would remain true if we allowed
non-Hausdorff constants. Concerning \cref{OP:2} note that an argument similar to
\Cref{C:initial} does not work, since monomorphisms are not smooth in $\KHaus$
(\Cref{E:cms}\ref{E:cms:khaus}), and similarly for $\Stone$.
\renewcommand{\itemautorefname}{item}%

\section{Hausdorff Polynomial Functors}
\label{S:Hausdorff}

Analogously to the Vietoris polynomial functors on $\Top$, we introduce Hausdorff polynomial
functors on $\MS$, the category of
extended metric spaces (distances can be $\infty$) and non-expanding maps.
 Closer to the situation of Kripke polynomial functors on $\Set$ than to
Vietoris polynomial functors on $\Top$, for the Hausdorff polynomial functors on $\MS$, the
terminal-coalgebra chain converges in $\omega+\omega$ steps.
\begin{remark}
  \begin{enumerate}
  \item The functors $\V\colon\Top\to\Top$ and $\H\colon \Met \to \Met$ are
    closely related: for compact metric spaces $X$ the Vietoris space
    $\V X$ is precisely the topological space induced by the Hausdorff
    space $\H X$.
    
  \item Some authors define $\H X$ to consist of all \emph{nonempty} compact subsets of
    $X$. However, Hausdorff~\cite{Hausdorff14} did not exclude $\emptyset$, and the formula in
    \cref{E:Hausdorff}\ref{E:Hausdorff:1} works (as already indicated) without such an
    exclusion.
  \end{enumerate}
\end{remark}

\begin{remark}\label{R:Haus}
  \begin{enumerate}
  \item For a complete metric space, $\H X$ is complete again (see
    e.g.~Barnsley~\mbox{\cite[Thm.~7.1]{Barnsley93})}.  Thus, $\H$
    restricts to an endofunctor on the category $\CMS$ of complete metric
    spaces, which we denote by the same symbol $\H$.

  \item Let $\UMet$ denote the category of (extended) ultrametric
    spaces: the full subcategory of $\MS$ given by spaces satisfying
    the following stronger version of the triangle inequality:
    \[
      d(x,z) \leq \max\set{d(x,y),d(y,z)}.
    \]
    If $X$ is an ultrametric space, then so is $\H X$.  To see this,
    let $S, T, U\in \H X$.  Write $p$ for
    $\max\set{\bar{d}(S,T), \bar{d}(T,U)}$.  For each $x\in S$, there
    is some $y\in T$ such that $d(x,y)\leq \bar{d}(S,T)$.  For this
    $y$, there is some $z\in U$ such that $d(y,z)\leq \bar{d}(T,U)$.
    So
    \[
      d(x,z) \leq \max\set{d(x,y), d(y,z)} \leq \max\set{\bar{d}(S,T),
        \bar{d}(T,U)} = p.
    \]
    It follows that $d(x,U) \leq p$.  This for all $x\in X$ shows that
    $d(S,U) \leq p$.  Note that
    $p = \max\set{\bar{d}(U,T), \bar{d}(T,S)}$.  The same argument
    shows that $\sup_{z\in U} d(z,S) \leq p$.  So we have
    $\bar{d}(S,U) \leq p$.  This proves the ultrametric inequality.

    We again denote the restriction of the Hausdorff functor to
    $\UMet$ by $\H$.

  \item For a discrete metric space $X$ (e.g.~where all distances are $0$
    or $1$), $\H X$ is the discrete space formed by all 
    finite subsets of $X$.

  \item\label{R:Haus:3} For every metric space $X$, the
    nonempty finite subsets of $X$ form a dense set in~$\H X$. Indeed,
    given a nonempty compact set $S \subseteq X$, for every $\eps > 0$, there
    exists a nonempty finite set $T \subseteq S$ such that  $S$ is covered
    by $\eps$-balls around the points in $T$. Therefore
    $d(x, T) \leq \eps$ for all $x \in S$, and we have $d(y,S) = 0$
    for all $y \in T$. This implies that $\bar d(S,T) \leq \eps$.
  \end{enumerate}
\end{remark}
\begin{example}\label{E:Haus-ter}
  For the Hausdorff functor, a terminal coalgebra is carried by the space of all finitely
  branching strongly extensional trees equipped with the discrete metric. This follows from the
  finite power-set functor $\powfin$ having its terminal coalgebra formed by those trees
  (\Cref{T:s-ext}\ref{T:s-ext:2}). Indeed, the terminal-coalgebra chain $V_i$ ($i \in \Ord$)
  for $\H$ is obtained by equipping the sets in the terminal-coalgebra chain for $\powfin$ with
  the discrete metric. Furthermore, since limits in $\Met$ (or $\CMS$) are set-based and limits
  of discrete spaces are discrete again, we see that, since the terminal-coalgebra chain for
  $\powfin$ converges in $\omega+\omega$ steps (by \cref{C:set-kvec}), so is the one for the
  Hausdorff functor; in symbols: $\nu\H = V_{\omega+ \omega}$.
\end{example}

It follows that, unlike the Vietoris functor, the Hausdorff functor does not preserve limits of
$\omega^\opp$-chains: the terminal-coalgebras chain for $\H$ does not converge before
$\omega + \omega$ steps (see \Cref{E:Haus-ter}). Thus, this functor does not preserve the limit
$V_\omega = \lim_{n <\omega} V_n$.

\begin{defn}
  Let $(X_n)_{n <\omega}$ be an $\omega^\opp$-chain in $\MS$. A cone
  $r_n\colon M \to X_n$ is \emph{isometric} if for all $x, y \in M$ we
  have $d(x,y) = \sup_{n < \omega} d (r_n(x), r_n(y))$.
\end{defn}

By \Cref{P:chain}\ref{R:Haus:4}, limit cones of
$\omega^\opp$-chains in $\MS$ are isometric.

\begin{proposition}\label{P:Haus-limit}
  The Hausdorff functor preserves isometric cones of
  $\omega^\opp$-chains.
\end{proposition}
\begin{proof}
  Let $(X_n)_{n < \omega}$ be an $\omega^\opp$-chain with connecting
  maps $f_n\colon X_{n+1} \to X_n$. Given an isometric cone
  $\ell_n\colon M \to X_n$ ($n < \omega$), we prove that the cone
$\H \ell_n\colon \H M \to \H X_n$ is also isometric:
  \[
    \bar d(S, T) = \sup_{n < \omega} \bar d (\ell_n(S), \ell_n(T))
    \qquad
    \text{for all compact subset $S, T \subseteq M$.}
  \]
  We can assume that $S$ and $T$ are nonempty and finite: since finite sets are
  dense in $\H M$ by \Cref{R:Haus}\ref{R:Haus:3}, and the maps
  $\ell_n$ are (non-expanding whence) continuous, the desired equality
  then holds for all pairs in $\H M$. The case where $S$ or $T$ is
  empty is trivial.

  Since every $\ell_n$ is non-expanding, we only need to prove that
  $\bar d(S,T) \leq c$ holds for
  $c = \sup_{n < \omega} \bar d(\ell_n[S], \ell_n[T])$.  For this,
we   show that for every $\eps > 0$, $\bar d(S,T) \leq c +
  \eps$. By the definition of the Hausdorff metric $\bar d$, it
  suffices to prove that for every $x \in S$ we have
  $d(x,T) \leq c + \eps$. By symmetry, we then also have
  $d(y, S) \leq c + \eps$ for every $y \in T$.

  Given $y \in T$ we have $d(x,y) = \sup_{n < \omega} d(\ell_n(x),
  \ell_n(y))$. Thus, there is a $k < \omega$ such that
  \[
    d(x,y) \leq d(\ell_k(x),\ell_k(y)) + \eps. 
  \]
  Since $T$ is finite, we can choose $k$ such that this inequality
  holds for all $y \in T$. By definition, 
  \[
   \bar d(\ell_k(x), \ell_k[T])
    =
    \inf_{y \in T} d(\ell_k(x), \ell_k(y))
    \qquad\text{in $X_k$}.
  \]
  Again using that  $T$ is finite, we can pick some $y \in T$ such that
  $d(\ell_k(x), \ell_k[T]) = d(\ell_k(x),\ell_k(y))$. With this $y$ we
  conclude that 
  \begin{align*}
    d(x,T)
    &
    \leq d(x,y) \leq d(\ell_k(x), \ell_k(y)) + \eps
    \\
    &
    = d(\ell_k(x),
    \ell_k[T]) + \eps \leq \bar d(\ell_k[S], \ell_k[T]) + \eps
    \\
    &
    \leq c + \eps. \tag*{\qedhere}
  \end{align*}  
\end{proof}
\begin{remark}\label{R:Haus-pres}
  The Hausdorff functor preserves isometric embeddings and their
  intersections. Indeed, for every subspace $X$ of a metric space $Y$,
  a set $S \subseteq X$ is compact in $X$ iff it is so in
  $Y$. Moreover, given $S, T \in \H X$, their distances in $\H X$ and
  $\H Y$ are the same. Thus, $\H$ preserves isometric embeddings.

  Given a collection $X_i \subseteq Y$ ($i \in I$) of subspaces, a set
  $S \subseteq \bigcap_{i\in I} X_i$ is compact iff it is so in~$Y$ (and therefore in every $X_i$). Thus, $\H$ preserves that
  intersection.
\end{remark}
\begin{defn}\label{D:Haus-poly}
  The \emph{Hausdorff polynomial functors} are the endofunctors on
  $\Met$ built from the Hausdorff functor, the constant functors, and
  the identity functor, using product, coproduct, and composition.
  Thus, the Hausdorff polynomial functors are built according to the following
  grammar (cf.~\Cref{D:Kripke}):
  \[
    \textstyle
    F ::= \H \mid A \mid \Id \mid \prod_{i\in I} F_i \mid
    \coprod_{i\in I} F_i \mid FF,
  \]
  where $A$ ranges over all metric spaces and $I$ is an arbitrary
  index set.
\end{defn}
\begin{theorem}\label{T:Haus-poly}
  For every Hausdorff polynomial functor $F\colon \Met \to \Met$, the terminal-coalgebra chain
  converges in $\omega+\omega$ steps: $\nu F = V_{\omega +\omega}$.
\end{theorem}
\begin{proof}
We use~\Cref{P:modest}, taking $\M$ to be the
    class of all isometric embeddings.  An easy induction over the
    structure of Hausdorff polynomial functors shows that each such
    functor $F$ preserves isometric embeddings and their
    intersections.  The base case for $\H$ is due to
    \Cref{R:Haus-pres}.  Another induction shows that each such
    functor $F$ preserves isometric cones.  Here the base case comes
    from \Cref{P:Haus-limit}.  From this second fact, we see that
    the image under $F$ of the limit cone $(V_{\omega} \to V_n)_n$ is
    an isometric cone.  It is then easy to see that the canonical
    $m\colon V_{\omega+1}\to V_{\omega}$ is an isometric embedding.
    So we have verified the hypotheses of~\Cref{P:modest}.
\end{proof}
\begin{remark}\label{R:Hausdorff-poly-cms}
  Note that if a Hausdorff polynomial functor $F$ uses only contants
  given by complete metric spaces $A$,
  then it has a restriction to an endofunctor on
  $\CMS$. Indeed, by an easy induction on the
  structure of $F$ one shows that $FX$ is complete whenever~$X$ is
  complete. Similarly, when $F$ uses constants which are ultrametric
  spaces, then $F$ has a restriction on $\UMet$. 
\end{remark}

Since $\CMS$ and $\UMet$ are closed under limits of
$\omega^\opp$-chains in $\Met$, we obtain the following

\begin{corollary}
  For every Hausdorff polynomial functor on $\CMet$ or $\UMet$, the terminal-coalgebra chain
  converges in $\omega+\omega$ steps.
\end{corollary}
\begin{corollary}\label{C:HP-I}
  Every Hausdorff polynomial functor $F$ on $\Met$ or $\CMS$ has an
  initial algebra.
\end{corollary}
\noindent
Indeed, since Hausdorff polynomial functors preserve isometric embeddings, 
this follows from~\Cref{T:Haus-poly},
\Cref{E:cms}\ref{E:cms:met}, and \Cref{T:initial}.

\begin{corollary}
\label{cor:Haus}
  Every Hausdorff polynomial functor $F$ on $\Met$, $\CMS$ or $\UMet$
  generates a cofree comonad obtained in $\omega+\omega$ steps. 
\end{corollary}
%

%
%

\begin{remark}\label{MP2}
  We have mentioned another possible approach to terminal coalgebras in
  \Cref{R:notAcc}.  Let us comment on the situation regarding the
  results on $\Met$ here.  The category~$\Met$ is locally presentable
  (see e.g.~\cite[Ex.~2.3]{AdamekR22}).  The Limit Theorem does
  imply that on $\MS$, the Hausdorff polynomial functors have terminal
  coalgebras.  In more detail, the Hausdorff functor is finitary: this
  was proved for its restriction to $1$-bounded metric
  spaces~\cite[Sec.~3]{AdamekEA15}, and the proof for $\H$ itself is
  the same. An easy induction then shows that every Hausdorff
  polynomial functor is accessible, so that the Limit Theorem can be
  applied.  However, our elementary proof shows that the terminal
  coalgebra chain converges in $\omega + \omega$ steps.  The proof of
  Makkai and Par\'e's Limit Theorem does not yield such a bound.
\end{remark}
\subsection{A Worked Example of a Terminal Coalgebra}\label{S:example}

Coalgebras for the Hausdorff polynomial functor $FX = \H(\Sigma\times
X)$ are metric labelled transition system with non-expanding
transitions. We shall describe its terminal coalgebra for a tiny space
of labels, $\Sigma = \set{0,1}$ be the metric given by $d(0,1) =
\delta< 1$, in terms of \emph{edge-labelled trees}.
We will define these in \cref{D:edge} below and establish
connections between edge-labelled trees and (ordinary) trees.  Our
final characterization of $\nu F$ appears in \cref{T:mainmetric}.

In addition to $F$, we are interested in the Hausdorff functor $\H$ itself.

A metric space is a \emph{$\delta$-space} if all distances are $0$,
$\delta$, or $\infty$.  What we use concerning the tiny space~$\Sigma$
in this example is that it is a finite $\delta$-space.  We write
$\Met_{\delta}$ for the full subcategory of $\Met$ determined by the
$\delta$-spaces.  This category is closed under products and under the
Hausdorff functor $\H$.  Thus, we may regard $F$ and~$\H$ as
endofunctors on $\Met_{\delta}$. There is a natural transformation
$\eps\colon F\to \H$ with the components given by
\begin{equation}\label{eps}
  \eps_X\colon \H(\Sigma\times X)\to\H X, \qquad
  \eps_X(s) = \set{x\in X : (a,x)\in s \mbox{ for some $\sigma\in\Sigma$}}.
\end{equation}
The main reason for moving from $\Met$ to $\Met_{\delta}$ is that
doing so ensures that the components~$\eps_X$ are non-expanding maps,
and so $\eps$ is indeed a natural transformation. As we shall see,
this leads to a natural transformation between the terminal-coalgebra chains
of $F$ and $\H$. 

As earlier, we denote by $V_i = \H^i 1$ the terminal-coalgebra chain
of $\H$, with projections $v_{ij}\colon V_i\to V_j$ for $i \geq j$. We
also denote by $\Vbar_i$ the terminal-coalgebra chain of $F$, with
projections $\vbar_{ij}\colon \Vbar_i\to \Vbar_j$ for $i \geq j$.  In
general, we use the `bar' notation for objects and morphisms related
to~$F$ and the `un-barred' notation for the parallel notions related
to $\H$.
  
It is important to see that~$F$ has the same terminal coalgebra whether it is taken as
endofunctor on $\Met$ or on $\Met_{\delta}$.  The terminal-coalgebra chains of the two
endofunctors are identical, and for the chain on $\Met$, we have
$\nu F = \Vbar_{\omega+\omega}$ by~\autoref{T:Haus-poly}.  So the chain on $\Met_{\delta}$ also
converges in $\omega + \omega$ steps and yields the same space $\nu F = \Vbar_{\omega+\omega}$
with the same structure map $\nu F \to F(\nu F)$.  Then by~\cite[second proposition, in dual
form]{A74}, the terminal coalgebra on $\Met_{\delta}$ exists and is the same coalgebra as on
$\Met$.  The same argument shows the same facts concerning $\H$, taken as either an endofunctor
on $\Met$ or on $\Met_{\delta}$.

The functors $F$ and $\H$ are liftings of
  $F_0X = \powfin(\Sigma\times X)$ and $\powfin$, respectively.  This
  means that the diagrams below commute.
  \[
    \begin{tikzcd}
      \Met_{\delta} \ar[r, "F"] \ar[d,swap,"U"] & \Met_{\delta} \ar[d,"U"]\\
      \Set \ar[r,"F_0"]   & \Set \\
    \end{tikzcd}
    \qquad
    \begin{tikzcd}
      \Met_{\delta} \ar[r, "\H"] \ar[d,swap,"U"] & \Met_{\delta} \ar[d,"U"]\\
      \Set \ar[r,"\powfin"]   & \Set \\
    \end{tikzcd}
  \]
  Here $U\colon \Met_{\delta}\to\Set$ is the forgetful functor taking a
  space to its set of points.  The commutativity relies on the fact
  that for every $\delta$-space $X$, the compact subsets of $X$ are
  exactly the finite subsets.

  We have seen that $\nu \H$ is $V_{\omega+\omega}$, and at this point we
  know that its underlying space is the set~$\nu \powfin$.%

  Our aim is to represent
  $\nu F$ as a set of objects which are trees with additional structure
(an edge labelling)
  also to express the metric  on  $\nu F$ solely using concepts to trees and to edge-labelled trees.

  When we need to refer to the terminal-coalgebra chain of $F_0$, 
  we use notation like $V^{F_0}_i$.
For example,   a transfinite
  induction on the ordinal $i$ shows that $U \Vbar_i = \Vbar^{F_0}_i$.  We are
  using that $U$ preserves limits.

The terminal coalgebra chains of $F$ and $\H$ are functors $\Vbar, V\colon\Ord\to\Met$.
There is a unique natural transformation
 $\psi\colon \Vbar\to V$.
This 
sequence of
morphisms $\psi_i\colon \Vbar_i \to V_i$ is: 
\begin{equation}\label{eq:psi}
\begin{aligned}
  \psi_0 & =  \text{the identity on the one-point space}\\
  \psi_{i+1} & = \eps_{V_i} \o F \psi_i \\
  \psi_i & =  \text{the unique morphism such that $v_{ij} \o \psi_i =
  \psi_j \o \vbar_{ij}$ for all $j< i$}
\end{aligned}
\end{equation}
where the last line above is the case for a limit ordinal $i$.
We leave to the reader the easy verification that  $\psi\colon \Vbar\to V$ is the unique natural transformation.

\begin{lemma}
The metric $d_i$ on each space $\Vbar_{i}$ is
  completely determined by $\psi_i: \Vbar_i\to V_i$, as 
  \begin{equation}
    \label{eq:metrics}
    d_i(s,t) = \left\{
    \begin{array}{ll}
    0 & \text{if $s = t$},\\
      \delta & \text{if $s \neq t$ and $\psi_i(s) = \psi_i(t)$}, \\
      \infty & \text{if $\psi_i(s) \neq \psi_i(t)$.}
    \end{array}
    \right.
  \end{equation}
  \end{lemma}
  
    \begin{proof}
      The proof is by transfinite induction on $i$.
  Since $\Vbar_i$ is a $\delta$-space, this boils down to showing that
  $d_i(s,t) =\delta$ if and only if $s\neq t$ and
  $\psi_i(s) = \psi_i(t)$.

  For
  $i= 0$, the result is trivial.  Assuming \eqref{eq:metrics} for $i$,
  we prove it for $i + 1$.  Suppose that $d_{i+1}(s,t) =\delta$.  Then
  $s\neq t$.  Let us check that
  $\psi_{i+1}(s) \subseteq \psi_{i+1}(t)$.  Let
  $x\in (\eps_{V_i}\o F\psi_i)(s)$.  By definition of
  $\eps_{V_i}$, there is some $(\sigma,s_0)\in \Sigma\times \Vbar_i$
  so that $\psi(s_0) = x$.  By definition of the product metric and
  the Hausdorff metric, there is some $(\tau,t_0)\in t$ such that
  \[ 
    d_{i}(s_0,t_0) \leq 
    d_{\Sigma\times \Vbar_i}((\sigma,s_0),(\tau, t_0)) \leq \delta.\]
  So either $s_0 = t_0$, or both $s_0 \neq t_0$ and $d_{i}(s_0,t_0) = \delta$.
  Either way,  $\psi_i(s_0) = \psi_i(t_0)$.  (The second alternative uses the induction hypothesis on $i$.)
  But $(\tau,t_0)\in \Sigma\times \Vbar_i$, and therefore
  \[
    x = \psi_i(s_0) = \psi_i(t_0) \in  (\eps_{V_i} \o F\psi_i)(t)
    = \psi_{i+1}(t),
  \]
  as desired.     
  The reverse inclusion $\psi_{i+1}(s) \supseteq \psi_{i+1}(t)$ holds
  by the same reasoning. 
  
  Conversely, we assume that $s\neq t$ and
  $\psi_{i+1}(s) = \psi_{i+1}(t)$ and prove that
  \mbox{$d_{i+1}(s,t) =\delta$.} The reasoning is quite
  similar.

  Finally, suppose that $i$ is a limit ordinal and that \eqref{eq:metrics} holds for all
  $j < i$.  Let $s\neq t\in \Vbar_i$ be such that $\psi_i(s) = \psi_i(t)$.  Then for all
  $j < i$, $\psi_j(\vbar_{ij}(s)) = \psi_j(\vbar_{ij}(t))$.  By induction hypothesis,
  $d_j(\vbar_{i,j}(s), \vbar_{i,j}(t)) \leq \delta$.  Since the metric on $\Vbar_i$ is the
  supremum of the projection metrics, we see that $ d_i(s,t) \leq \delta$.  As $s\neq t$, we
  have $ d_i(s,t) = \delta$.  Going the other way, suppose that $d_i(s,t) = \delta$.  Then for
  all $j < i$, $d_j(\vbar_{i,j}(s), \vbar_{i,j}(t))\leq \delta$, and for some $j$, this is an
  equality. Since $d_i$ is the supremum of the projection metrics, $d_i(s, t)= \delta$, which
  implies that $s\neq t$.  By the induction hypothesis, for all $j<i$,
  $\psi_j(\vbar_{ij}(s)) = \psi_j(\vbar_{ij}(t))$.  It follows that $\psi_i(s) = \psi_i(t)$ as desired.
\end{proof}

We want to describe $\nu F = \Vbar_{\omega+\omega}$ in $\Met_\delta$, and since the set
underlying this space is~$\Vbar^{F_0}_{\omega+\omega}$, we shift the discussion from
$\Met_{\delta}$ to $\Set$.  For $\powfin$, we know that its terminal coalgebra may be described
as the set of all finitely branching, strongly extensional trees (\cref{L:Worrell}).  We aim to
describe $\nu F_0$ analgously.  The leading idea is that $\nu F_0$ should be `fairly close' to
$\nu \powfin$.  Indeed, we shall describe $\nu F_0$ as a certain set edge-labelled trees.

  \begin{defn}\label{D:edge}
    \begin{enumerate}
    \item An \emph{(edge-)labelled tree} is a tree whose edges are labelled in
      $\Sigma$. It follows that it is a coalgebra
      $e\colon G\to \pow(\Sigma\times G)$.
    \item If $x\in G$ and $(\sigma,y)\in e(x)$, then we say that $y$ is a
      \emph{$\sigma$-neighbor} of $x$.
    \end{enumerate}
  \end{defn}
  
  The natural transformation $\psi$ from~\eqref{eq:psi} may be recast as a
  natural transformation $\psi\colon \pow(\Sigma\times -)\to \pow$.
  Then $\psi$ induces a functor
  \[
    \erase\colon \Coalg \pow(\Sigma\times -) \to \Coalg \Pow
  \]
  defined on objects by $\erase(G,e) = (G, \psi_G\o e)$ and being the
  identity on coalgebra morphisms.

  We generalize all of the definitions and results concerning trees
  which we saw in Section~\ref{S:app-worrell} to edge-labelled trees.
  For the most part, the generalization is smooth: one replaces
  `neighbors' by `$\sigma$-neighbors' in the appropriate way
  throughout.  For example, an edge-labelled tree $t$ is
  \emph{finitely branching} if for every node $x$ and every
  $\sigma\in \Sigma$, $x$ has only finitely many $\sigma$-neighbors in
  $t$. Since~$\Sigma$ is finite, this amounts to being finitely
  branching in our previous sense.  In contrast, in the notion of
  bisimulation of edge-labelled trees, the labels are more important;
  hence strong extensionality for edge-labelled trees is not the same
  as strong extensionality for trees.  We consider the relation
  between compactly branching edge-labelled trees and trees below.
  Note that the erasure of a strongly extensional edge-labelled tree
  is not necessarily a strongly extensional tree.

Recall that in \cref{N:partial} we introduced the notation $\TT$ for the class of all
  trees.  We have seen the maps $\partial_n\colon\TT\to V_n$
  via their defining equation \eqref{eq-partial-np}. 
   We also have the class $\overline{\TT}$ of
  all edge-labelled trees.  This time, we have
  $\partialbar_n\colon\overline{\TT}\to \Vbar_n$, and the defining
  equation is
\begin{equation}\label{keyeq}
\partialbar_{n+1}(t) = \set{(\sigma,\partialbar_n(t_x)): \mbox{ $x$ is a $\sigma$-child of the root of $t$}}.
\end{equation}
We also have seen that for every (ordinary) tree $t$ there are analagous maps $\rho^{t}_{n}: t\to V_n$ defined by a similar recursion.
For every edge-labelled tree $t$, there are maps $\rhobar^{t}_{n}: t\to \Vbar_n$, defined by
\begin{equation}\label{eq:rho}
  \rhobar^t_{n+1}(x) = \set{(\sigma,\rhobar^t_{n}(y)) : \mbox{ $x$ is a
      $\sigma$-child of $y$}}.
\end{equation}
As before, the connection is that $\partialbar_n(t) = \rhobar^t_n(\root(t))$.

Using the maps $\rhobar^t_n$, we define \emph{limits} as in \cref{dcb}.
An edge-labelled tree is \emph{compactly branching} if for
all nodes $x$ and all $\sigma\in \Sigma$:
 for every sequence of $(y_n)$ of $\sigma$-children of $x$
  there is 
  a subsequence $(w_n)$ of $(y_n)$ and 
  some $\sigma$-child $z$ of $x$ such that $\lim {w_n} = {z}$.

We adapt the functor $\erase$
 to give an operation which turns
 edge-labelled trees into (ordinary) trees.
Concretely,
\[ \erase\colon \overline{\TT}\to \TT\] is given as follows:  For each edge-labelled tree $t$, $\erase(t)$ is the tree with
the same nodes as~$t$, and with $x\to y$ in $\erase(t)$ iff for some $\sigma\in \Sigma$, $x\arrowsigma y$ in $t$.
Pictorially,  $\erase$ `forgets the edge labels.'

We need the following fact:
For all edge-labelled trees $t$ and natural numbers $n$, 
\begin{equation}\label{rhobarn}
\rho^{\erase(t)}_n = \psi_n \o \rhobar^t_n.
\end{equation}
The proof is by induction on $n$. For $n = 0$, this assertion is trivial.
Assume~\eqref{rhobarn} for $n$.  Let~$x$ be a node of $t$.
We have 
\[
\begin{array}{cll}
& \rho^{\erase(t)}_{n+1}(x) \\
= & \set{ \rho^{\erase(t)}_{n}(y) : \mbox{$y$ is a child of $x$ in $\erase(t)$} }  \\
= & 
\set{ \rho^{\erase(t)}_{n}(y) : \mbox{ for some $\sigma\in \Sigma$,
    $y$ is a $\sigma$-child of $x$ in $t$} }  & \text{by def.~of $\erase(t)$}\\
= &\set{\psi_n\rhobar^t_{n}(y) : \mbox{for some $\sigma\in \Sigma$, $y$
    is a $\sigma$-child of $x$ in $t$}} & \text{induction hypothesis} \\
= &  \eps_{V_n} \bigl(\set{(\sigma,\psi_n\rhobar^t_{n}(y)) :
  \mbox{$\sigma\in \Sigma$, and $y$ is a $\sigma$-child of $x$ in $t$}
}\bigr) & \text{by~\eqref{eps}}\\
= &  \eps_{V_n}\o F\psi_n \bigl(\set{(\sigma,\rhobar^t_{n}(y)) :
  \mbox{$\sigma\in \Sigma$, and $y$ is a $\sigma$-child of $x$ in $t$}
}\bigr)  & \text{by def.~of $F$}\\
 = & \psi_{n+1} \o \rhobar^t_{n+1}(x) & \text{by~\eqref{eq:psi} and~\eqref{eq:rho}}\\
\end{array}
\]
This concludes the verification of~\eqref{rhobarn} for all $n < \omega$.
It follows that  $\rho^{\erase(t)}_\omega= \psi_\omega \o \rhobar^t_\omega$.

The $\omega^\opp$-limit $V^{F_0}_{\omega}$ bijectively
  corresponds to the set of all compactly branching strongly
  extensional edge-labelled trees (\Cref{bijection}).  In one
  direction, the correspondence is $t\mapsto \partialbar_{\omega}(t)$,
  and in the other direction it is
  $x\in V^{F_0}_{\omega} \mapsto \temptree_x$
  (\cref{N:tx}\ref{N:tx:3}). It cuts down to a correspondence between
  the terminal coalgebra $\nu F_0$ and the set of finitely branching
  strongly extensional edge-labelled trees.  The proofs are completely
  analogous to those in \Cref{S:app-worrell}.

  We use \Cref{keyeq} to check that for all compactly branching edge-labelled trees~\mbox{$t$,
  $\erase(t)$} is a compactly branching tree.  Fix a node $x$ of $\erase(t)$.  So $x$ is a node
  of~$t$.  Let~$(y_n)$ be a sequence of children of $x$ in $\erase(t)$.  For each~$y_n$, there
  is some $\sigma_n\in\Sigma$ such that~$y_n$ is a $\sigma_n$-child of~$x$ in~$t$.  Since
  $\Sigma$ is finite, we can find a single fixed $\sigma$ and a subsequence~$(z_n)$ of~$(y_n)$
  such that each $z_n$ is a $\sigma$-child of $x$.  Since~$t$ is compactly branching, there is
  a subsequence~$(w_n)$ of~$(z_n)$ and some $\sigma$-child~$w^*$ of $x$ such that
  $\lim_n {w_n} = {w^*}$.  Now~$w^*$ is a child of $x$ in $\erase(t)$, and we claim that in
  $\erase(t)$, $\lim_n {w_n} = {w^*}$.  To see this, fix $n$.  For all sufficiently large $p$,
  $ \rhobar^t_n(w_p) = \rhobar^t_n(w^*)$.  Thus, we have
  \[
    \rho^{\erase(t)}_n(w_p) =  \psi_n(\rhobar^t_n(w_p)) =
    \psi_n(\rhobar^t_n(w^*)) = \rho^{\erase(t)}_n(w^*).
  \]
  This concludes the verification.

  
\begin{lemma}\label{L:mainequiv}
  Let $t$ and $u$ be compactly branching edge-labelled trees, and let $x =
  \partialbar_{\omega}(t)$ and $y = \partialbar_{\omega}(u)$ be the
  corresponding elements of $\Vbar_{\omega}$. The following are equivalent:
  \begin{enumerate}
  \item  $\psi_{\omega}(x) = \psi_{\omega}(y)$.
    \label{psipart}
  \item $\erase(t)$ and $\erase(u)$ are Barr-equivalent trees.
    \label{barrpart}
  \item $\erase(t)$ and $\erase(u)$ are bisimilar trees.
    \label{bisimpart}
  \end{enumerate}
\end{lemma}
\begin{proof}
  \ref{psipart} $\Leftrightarrow$ \ref{barrpart}.  Let us write $x_0$
  for the root of $t$; this is the same as the root of $\erase(t)$.
  Now we have
  \[
    \partial_{\omega}(\erase(t))
    =
    \rho^{\erase(t)}_\omega(x_0)
    =
    \psi_\omega \o \rhobar^t_\omega(x_0)
    = 
    \psi_{\omega}(\partialbar_{\omega}(t) )
    = 
    \psi_{\omega}(x).
  \]
  Of course, we have a similar equation for $y$. Hence, \ref{psipart} is
  equivalent to the assertion that
  $ \partial_{\omega}(\erase(t)) = \partial_{\omega}(\erase(u))$.
  That is, \ref{barrpart} holds.

  The equivalence \ref{barrpart} $\Leftrightarrow$ \ref{bisimpart}
  is due to the version of \Cref{corccb} for edge-labelled~trees.
\end{proof}
 
At long last, we can state our conclusion.
\begin{theorem}\label{T:mainmetric}
  \begin{enumerate}
  \item The limit $\Vbar_\omega$ in the terminal-coalgebra chain for
    $F$ is the set of all compactly branching, strongly
    extensional edge-labelled trees equipped with the metric in~\eqref{eq:metrics}.
  \item The terminal coalgebra for $F$ is the coalgebra of all
    strongly extensional, finitely branching edge-labelled trees. The coalgebra
    structure is the inverse of tree tupling and the metric~is~\eqref{eq:metrics}.
  \end{enumerate}
\end{theorem}
\begin{proof}
  \Cref{L:mainequiv} expresses the metric on $\Vbar_{\omega}$. Indeed,
  using~\eqref{eq:metrics} and \ref{psipart} $\Leftrightarrow$
  \ref{bisimpart}, we have that for all strongly extensional compactly
  branching edge-labelled trees $t$ and $u$,
  \begin{equation}
    \label{eq:ultimate}
    d_{\omega}(t,u) = \left\{
      \begin{array}{ll}
        0 & \text{if $t=u$}\\
        \delta & \text{if  $t\neq u$ but  $\erase(t)$ and $\erase(u)$ are bisimilar trees} \\
        \infty & \text{if $\erase(t)$ and $\erase(u)$ are not bisimilar trees.}
      \end{array}
    \right.
  \end{equation}
  For $\Vbar_{\omega+\omega}$ we use that the inclusion
  $\Vbar_{\omega+\omega}\hookrightarrow \Vbar_\omega $ is an
  isometry. Hence, the same formula works: for strongly extensional
  finitely branching edge-labelled trees $t$ and $u$,
  $d_{\omega+\omega}(s,t)$ is given by~\eqref{eq:ultimate}.
\end{proof}

\subsection{Variation: the Closed Subset Functor on $\Met$}

We have been concerned with the Hausdorff functor taking a metric space $M$ to the space of its
nonempty compact subsets.  For two variations, let us consider the functor
$\powcl\colon \MS\to\MS$ taking $M$ to the set of its \emph{closed} subsets, and its subfunctor
$\powcl'\colon \MS \to \MS$ taking $M$ to the set of its nonempty closed subsets.  Both
$\powcl M$ and $\powcl' M$ are given the Hausdorff metric $\bar{d}$
(\cref{E:Hausdorff}\ref{E:Hausdorff:1}. For a non-expanding map $f\colon X \to Y$, the
non-expanding map $\powcl f\colon \powcl X \to \powcl Y$ sends a closed subset $S$ of $X$ to
the closure of $f[S]$.  This makes $\powcl$ and $\powcl'$ functors.  Due to the empty set,
$\powcl$ is a closer analog of $\H$ than $\powcl'$. It is natural to ask whether the results of
\Cref{S:Hausdorff} hold for these functors $\powcl$ and $\powcl'$.  As proved by van
Breugel~\cite[Prop.~8]{vanBreugel}, the functor $\powcl$ has no terminal coalgebra.  Turning to
$\powcl'$, this functor has an initial algebra
given by the empty metric space and a terminal coalgebra carried
by a singleton metric space. But $\powcl'$ has no
other fixed points (see van Breugel et al.~\cite[Cor.~5]{vBCW04}), where an object $X$ is a
\emph{fixed point} of an endofunctor $F$ if $FX \cong X$. We provide below a different, shorter
proof.
\begin{remark}\label{R:ubd-below-i}
\begin{enumerate}
\item A subset $X$ of a metric space is \emph{$\delta$-discrete} if
  whenever $x\neq y$ are elements of $X$, $d(x,y) \geq \delta$.  Every
  subset of a $\delta$-discrete set is $\delta$-discrete, and every
  such set is closed.  Moreover, if~$C$ and $D$ are different subsets
  of a $\delta$-discrete set, then $\bar d(C,D)\geq \delta$.
\item
  A subset $S$ of an ordinal $i$ is \emph{cofinal} if for all $j<i$
  there is some $k\in S$ with $j \leq k < i$. If $S$ is not cofinal,
  then its complement $i\setminus S$ must be so. (But it is possible
  that both~$S$ and~$i\setminus S$ are cofinal in $i$.)
\end{enumerate}
\end{remark}

\removeThmBraces
\begin{theorem}\label{T:Pcl}
  There is no isometric embedding $\powcl' M \to M$ when $|M| \geq
  2$.
  \end{theorem}
\resetCurThmBraces
\begin{proof}%
  Suppose towards a contradiction that $\iota\colon \powcl M \to M$
  were an isometric embedding where $|M| \geq 2$.  If all distances in
  $M$ are $0$ or $\infty$, then $\powcl' M$ is the nonempty power-set
  of~$M$.  In this case, our result follows from the fact that for
  $|M|\geq 2$, $M$ has more nonempty subsets than
  elements.
  Thus, we fix distinct points $a, b\in M$ of finite
  distance, and put $\delta = d(a,b)/2$.  Let
  $A = \set{x\in M : d(x,a) \leq \delta}$, and let $B = M\setminus A$.
  (In case $d(a,b) = \infty$, we need to adjust this by setting
  $\delta = \infty$, and $B$ to be the points whose distance to $a$ is
  finite.  But we shall not present the argument in this case.)

  We proceed to define an ordinal-indexed sequence of elements
  $x_i\in M$. We also prove that each set $S_i = \set{ x_j : j < i}$
  is $\delta$-discrete, and we put
  \[
    X_i =
    \begin{cases}
      A & \text{if $\set{j < i : x_j \in A}$ is cofinal in $i$} \\
      B & \text{else}.
    \end{cases}
  \]   
  For $i =0$, put $x_0 = \iota(\set{a,b})$.  Given an ordinal $i>0$,
  we put 
  \[
    x_i = \iota(X_i \cap S_i).
  \]
  Being nonempty (since $i>0$) and
  $\delta$-discrete, $X_i \cap S_i$ lies in~$\powcl' M$.

  The remainder of our proof consists of showing that for every ordinal $i$:
  \[
    d(x_j, x_k) \geq \delta
    \qquad\text{for $0 \leq j < k \leq i$}.
  \]
  We proceed by transfinite induction. Assuming that our claim holds for every $k < i$, we then
  prove it for $i$. The base case $i = 0$ is trivial. For $i > 0$, note first that it follows
  from the induction hypothesis that $S_i$ is $\delta$-discrete.

  Hence, we only need to verify that $d(x_j, x_i) \geq \delta$ when
  $0\leq j < i$.  We argue the case $X_i = A$; when $X_i = B$, the
  argument is similar, mutatis mutandis.  For $j= 0$, recall that
  $x_0 = \iota(\set{a,b})$ and $x_i = \iota(A\cap S_i)$.  Since $b$ has
  distance at least $\delta$ from every element of~$A$, we obtain
  $\bar d(\set{a,b}, A\cap S_i)\geq \delta$. As $\iota$ is an
  isometric embedding, this distance is also $d(x_0,x_i)$.  Now let $j >0$.
  Since we have $X_i = A$, let $k$ be such that $j \leq k < i$ and
  $x_k\in A$ using cofinality.
  Now either $x_j = \iota(A \cap S_j)$ or else
  $x_j = \iota(B \cap S_j)$.

  In the first case, note that $x_k\in A\cap S_i$ since $k < i$, and
  $x_k\notin S_j$ by the definition of $S_j$ since $k \geq j$.  So
  $A\cap S_j$ and $A\cap S_i$ are different nonempty subsets of the
  $\delta$-discrete set $S_i$. Hence, the distance between these sets
  is at least $\delta$, and therefore we have
  $d(x_j, x_i) \geq \delta$.
 
  In the second case, $B\cap S_j$ is a nonempty subset of $B$, and
  thus again it is not equal to~$A\cap S_i$.  So again we see that
  $d(x_j, x_i) = \bar d( B\cap S_j, A\cap S_i) \geq \delta$.

  We now obtain the desired contradiction since $(x_i)$ is an
  ordinal-indexed sequence of pairwise distinct elements of $M$.
\end{proof}

\begin{corollary}\label{C:Pcl}
  \begin{enumerate}
  \item The functor $\powcl'\colon \MS\to \MS$ has no fixed points
    except the empty set and the singletons.
  \item\label{C:Pcl:2} The functor $\powcl\colon \MS \to \MS$ admits
    no isometric embedding $\powcl M \to M$, whence has no fixed point.
  \end{enumerate}
\end{corollary}
\begin{proof}
  The first item is immediate from \Cref{T:Pcl}. For the second
  one, observe that the inclusion map $e\colon \powcl' M \to \powcl M$
  is an isometric embedding. Assuming that there were an isometric
  embedding $\iota\colon \powcl M \to M$, we see that $M$ cannot be
  empty (since $\powcl M$ is nonempty) or a singleton (since then
  $|\powcl M| = 2$). Hence $|M| \geq 2$. Moreover, we obtain an
  isometric embedding \mbox{$\iota \cdot e\colon \powcl' M \to M$},
  contradicting \Cref{T:Pcl}.%
\end{proof}

\section{Summary}

This paper has two parts, both giving sufficient conditions for the terminal coalgebra for an
endofunctor to be obtained in $\omega$ or $\omega+\omega$ steps of the well-known iterative
construction.  The first part generalizes Worrell's theorem that states that for finitary set
functors, the terminal-coalgebra chain converges in $\omega+\omega$ steps.  That
generalization concerns DCC-categories; examples include sets, vector spaces, posets, and many
others. For finitary endofunctors preserving nonempty binary intersections, the
terminal-coalgebra chain is proved to converge in~$\omega+\omega$ steps.

In the second part, we have worked with variations of the Kripke polynomial functors. These
need not be finitary.  We have proved again that the terminal-coalgebras chain of these
functors converges in $\omega+\omega$ steps. More precisely, we have investigated versions of
the finite power-set functor on the categories $\Haus$ and $\Met$.  Our main results are that
the Vietoris functor $\V$, and indeed all Vietoris polynomial functors, have terminal
coalgebras obtained in $\omega$ steps of the terminal-coalgebra chain.  The same holds for the
Hausdorff polynomial functors on $\Met$, but the iteration takes $\omega+\omega$ steps and so
the underlying reasons are different.

Our work on the Kripke and Hausdorff polynomial functors highlights a technique which we feel
could be of wider interest.  To prove that a terminal coalgebra exists in a situation where the
limit of the $\omega^\opp$-chain~\eqref{eq:op-chain} is not preserved by the functor, one could
try to find preservation properties which imply that the limit of the $\omega^\opp$-chain
$(V_{\omega+n})_n$ was preserved.  In $\Set$, we used finitarity and preservation of
monomorphisms and intersections, and in $\Met$ we have used preservation of intersections, isometric
embeddings, and isometric cones.

We have also seen that for the functor $\powcl$ on $\Met$, there is no
fixed point and hence no terminal coalgebra.  We leave open the
question of whether every Vietoris polynomial functor on $\Top$ has a
terminal coalgebra.


%
%

%
%

\paragraph{Acknowledgements.}
We are grateful to Pedro Nora for discussions on the proof of
\cref{P:limit-Haus}, to Alex Kruckman for discussion of an early version of this
paper, and to Ryan Tay for pointing out a mistake in the proof of \cref{T:Vietoris}.
In addition we thank the referees whose comments improved the presentation of our paper. 

\bibliographystyle{plainurl}
\bibliography{refs}

%
%



\end{document}
